\documentclass[acmsmall,screen,nonacm]{acmart}

\acmJournal{PACMPL}
\acmVolume{1}
\acmNumber{CONF} 
\acmArticle{1}
\acmYear{2020}
\acmMonth{1}
\acmDOI{} 
\startPage{1}

\setcopyright{none}

\usepackage{booktabs}   
\usepackage{subcaption} 

\usepackage{graphicx}
\usepackage{amssymb}
\usepackage{epstopdf}
\usepackage{amsmath}
\usepackage{color}
\usepackage{mathpartir}
\usepackage{todonotes}
\usepackage{comment}
\usepackage{enumitem}
\usepackage{listings}
\usepackage{amsthm}
\usepackage{xspace}
\usepackage{fdsymbol}
\usepackage{suffix}

\makeatletter
\newcommand*{\textoverline}[1]{$\overline{\hbox{#1}}\m@th$}
\makeatother

\makeatletter
\providecommand*{\dashv}{%
  \mathrel{%
    \mathpalette\@dashv\vDash
  }%
}
\newcommand*{\@dashv}[2]{%
  \reflectbox{$\m@th#1#2$}%
}
\makeatother

\newcommand{\ok}[1]{\ensuremath{#1\textbf{\text{ ok}}}}

\newcommand{\keyword}[1]{\textbf{#1}\xspace}

\newcommand{\asset}{\keyword{asset}}
\newcommand{\Owned}{\keyword{Owned}}
\newcommand{\Unowned}{\keyword{Unowned}}
\newcommand{\Shared}{\keyword{Shared}}
\newcommand{\contract}{\keyword{contract}}
\newcommand{\stateExpr}{\keyword{state}}
\newcommand{\interface}{\keyword{interface}}
\newcommand{\implements}{\ensuremath{\triangleleft}}
\newcommand{\new}{\keyword{new}}
\newcommand{\letExpr}[4]{\ensuremath{\keyword{let} \ {#1}: {#2} = {#3} \ \keyword{in} \ {#4}}}

\newcommand{\inExpr}{\keyword{in}}
\newcommand{\ifExpr}[5]{\ensuremath{\keyword{if} \ {#1} \  \inExpr\textsubscript{{#2}} \ {#3} \ \keyword{then} \ {#4} \ \keyword{else} \ {#5}}}
\newcommand{\assertExpr}[2]{\ensuremath{\keyword{assert} \  #1 \  \keyword{in} \ #2}}
\newcommand{\trans}{\ensuremath{\operatorname{\rAngle}}}
\newcommand{\ty}[6]{\ensuremath{#1; #2 \vdash_{#3}{}{} #4 : #5 \dashv{}{} #6}}
\newcommand{\subtype}[3]{\ensuremath{#1 \vdash #2 <: #3}}
\newcommand{\subperm}[3]{\ensuremath{#1 \vdash #2 <:_* #3}}
\newcommand{\notsubperm}[3]{\ensuremath{#1 \vdash #2 \not<:_* #3}}
\newcommand{\this}{\keyword{this}}

\newcommand{\disown}[1]{\textbf{disown} \ #1}
\newcommand{\Unit}{\keyword{unit}}
\newcommand{\unit}{\ensuremath{()}}
\newcommand{\pack}{\keyword{pack}}
\newcommand{\lStronger}[3]{\ensuremath{#1 <^l_{#2} #3}}

\newcommand{\splitType}[3]{\ensuremath{#1 \Rrightarrow #2 / #3}}
\newcommand{\okIn}[2]{\ensuremath{#1 \textbf{ ok in } #2}}

\newcommand{\compatibleTypes}[2]{\ensuremath{#1 \leftrightarrow #2}}

\newcommand{\sameOwnership}[2]{\ensuremath{#1 \approx #2}}

\newcommand{\phibox}[2]{\ensuremath{\fbox{$#1$}_{#2}}}
\newcommand{\psibox}[2]{\ensuremath{\fbox{$#1$}^{#2}}}

\newcommand{\code}[1]{\texttt{#1}}

\newcommand{\stepsTo}[2]{\ensuremath{#1 \rightarrow #2}}
\WithSuffix\newcommand\stepsTo*[2]{\ensuremath{#1 \rightarrow^* #2}}

\newtheorem{theorem}{Theorem}[section]
\newtheorem{lemma}{Lemma}[section]
\newtheorem{definition}{Definition}[section]

\newcommand{\envUpdate}[3]{\ensuremath{[#1 / #2]\, #3}}

\DeclareMathOperator*{\concat}{+\!\!+}

\newcommand{\bnfdef}{\ensuremath{\Coloneqq}}
\newcommand{\bnfalt}{\ensuremath{\mid}\xspace}

\newcommand{\generic}[2]{\ensuremath{#1\langle #2 \rangle}}
\newcommand{\generics}[2]{\generic{#1}{\overline{#2}}}

\newcommand{\bounded}{\implements} 

\newcommand{\bound}[2]{\ensuremath{#1 \vdash \text{bound}\left( #2 \right)}}
\newcommand{\boundPerm}[2]{\ensuremath{#1 \vdash \text{bound}_{*}\left( #2 \right)}}
\newcommand{\typeBounds}{\ensuremath{\Gamma}}
\newcommand{\subsOk}[3]{\ensuremath{\text{subsOk}_{#1}\left(#2, #3\right)}}
\newcommand{\genericsOk}[2]{\ensuremath{\text{genericsOk}_{#1}\left(#2\right)}}

\newcommand{\states}[1]{\ensuremath{\text{states}\left( #1 \right)}}
\newcommand{\transactionNames}[1]{\ensuremath{\text{transactionNames}\left( #1 \right)}}
\newcommand{\transactionName}[1]{\ensuremath{\text{transactionName}\left( #1 \right)}}
\newcommand{\stateNames}[1]{\ensuremath{\text{stateNames}\left( #1 \right)}}
\newcommand{\stateName}[1]{\ensuremath{\text{stateName}\left( #1 \right)}}
\newcommand{\isVar}[1]{\ensuremath{\text{isVar}\left( #1 \right)}}
\newcommand{\nonVar}[1]{\ensuremath{\text{nonVar}\left( #1 \right)}}
\newcommand{\implementOk}[3]{\ensuremath{\text{implementOk}_{#1}\left(#2, #3\right)}}
\newcommand{\params}[1]{\ensuremath{\text{params}\left( #1 \right)}}
\newcommand{\possibleStates}[2]{\ensuremath{\text{possibleStates}_{#1}\left( #2 \right)}}
\newcommand{\maybeOwned}[1]{\ensuremath{\text{maybeOwned}\left( #1 \right)}}
\newcommand{\notOwned}[1]{\ensuremath{\text{notOwned}\left( #1 \right)}}
\newcommand{\isAsset}[2]{\ensuremath{#1 \vdash \text{isAsset}\left( #2 \right)}}
\newcommand{\nonAsset}[2]{\ensuremath{#1 \vdash \text{nonAsset}\left( #2 \right)}}
\newcommand{\nonAssetState}[2]{\ensuremath{#1 \vdash \text{nonAssetState}\left( #2 \right)}}
\newcommand{\genericParam}[4]{#1\ #2.#3\ \bounded\ #4}
\newcommand{\genericParamOpt}[3]{\genericParam{[\textbf{asset}]}{#1}{#2}{#3}}
\newcommand{\substitute}[2]{\ensuremath{\sigma\left(#1\right)\left(#2\right)}}
\newcommand{\substitution}[1]{\ensuremath{\sigma\left(#1\right)}}
\newcommand{\disposable}[2]{\ensuremath{#1 \vdash \text{disposable}\left( #2 \right)}}
\newcommand{\nonDisposable}[2]{\ensuremath{#1 \vdash \text{nonDisposable}\left( #2 \right)}}
\newcommand{\transaction}[3]{\ensuremath{\text{transaction}_{#1}\left( #2, #3 \right)}}
\newcommand{\permVarMap}{\xi}
\newcommand{\PermVar}[1]{\ensuremath{\text{PermVar}\left( #1 \right)}}

\newcommand{\ToPermission}[1]{\ensuremath{\text{ToPermission}\left( #1 \right)}}

\newcommand{\lookup}[2]{\ensuremath{\text{lookup}_{#1}\left( #2 \right)}}
\newcommand{\funcArg}[3]{\ensuremath{\text{funcArg}\left(#1, #2, #3\right)}}
\newcommand{\funcArgResidual}[3]{\ensuremath{\text{funcArgResidual}\left(#1, #2, #3\right)}}

\newcommand{\ra}[1]{\renewcommand{\arraystretch}{#1}}

\lstdefinelanguage{obsidian}{
	keywords=[1]{contract,state,transaction,disown,new,asset,this,in,return,returns,remote,main,private},
	keywordstyle=[1]\color{blue}\bfseries,
	keywords=[2]{int},
	keywordstyle=[2]\color{teal}\bfseries,
	keywords=[3]{->},
	keywordstyle=[3]\color{blue}\bfseries,
	keywords=[4]{Owned,Unowned,Shared},
	keywordstyle=[4]\color{black}\bfseries,
	comment=[l]{//},
    morecomment=[s]{/*}{*/},
	commentstyle=\color{gray}\ttfamily
}

\lstdefinelanguage{Solidity}{
	keywords=[1]{anonymous, assembly, assert, balance, break, call, callcode, case, catch, class, constant, continue, constructor, contract, debugger, default, delegatecall, delete, do, else, emit, event, experimental, export, external, false, finally, for, function, gas, if, implements, import, in, indexed, instanceof, interface, internal, is, length, library, log0, log1, log2, log3, log4, memory, modifier, new, payable, pragma, private, protected, public, pure, push, require, return, returns, revert, selfdestruct, send, solidity, storage, struct, suicide, super, switch, then, this, throw, transfer, true, try, typeof, using, value, view, while, with, addmod, ecrecover, keccak256, mulmod, ripemd160, sha256, sha3}, 
	keywordstyle=[1]\color{blue}\bfseries,
	keywords=[2]{address, bool, byte, bytes, bytes1, bytes2, bytes3, bytes4, bytes5, bytes6, bytes7, bytes8, bytes9, bytes10, bytes11, bytes12, bytes13, bytes14, bytes15, bytes16, bytes17, bytes18, bytes19, bytes20, bytes21, bytes22, bytes23, bytes24, bytes25, bytes26, bytes27, bytes28, bytes29, bytes30, bytes31, bytes32, enum, int, int8, int16, int24, int32, int40, int48, int56, int64, int72, int80, int88, int96, int104, int112, int120, int128, int136, int144, int152, int160, int168, int176, int184, int192, int200, int208, int216, int224, int232, int240, int248, int256, mapping, string, uint, uint8, uint16, uint24, uint32, uint40, uint48, uint56, uint64, uint72, uint80, uint88, uint96, uint104, uint112, uint120, uint128, uint136, uint144, uint152, uint160, uint168, uint176, uint184, uint192, uint200, uint208, uint216, uint224, uint232, uint240, uint248, uint256, var, unit, ether, finney, szabo, wei, days, hours, minutes, seconds, weeks, years},	
	keywordstyle=[2]\color{teal}\bfseries,
	keywords=[3]{block, blockhash, coinbase, difficulty, gaslimit, number, timestamp, msg, data, gas, sender, sig, value, now, tx, gasprice, origin},	
	keywordstyle=[3]\color{violet}\bfseries,
	identifierstyle=\color{black},
	sensitive=false,
	comment=[l]{//},
	morecomment=[s]{/*}{*/},
	commentstyle=\color{gray}\ttfamily,
	stringstyle=\color{red}\ttfamily,
	morestring=[b]',
	morestring=[b]"
}

\newenvironment{absolutelynopagebreak}
  {\par\nobreak\vfil\penalty0\vfilneg
   \vtop\bgroup}
  {\par\xdef\tpd{\the\prevdepth}\egroup
   \prevdepth=\tpd}

\begin{document}

\title[Obsidian]{Obsidian: Typestate and Assets for Safer Blockchain Programming}         


\author{Michael Coblenz}
\orcid{0000-0002-9369-4069}             
\affiliation{
  \position{Doctoral Candidate}
  \department{Computer Science Department}              
  \institution{Carnegie Mellon University}            
  \streetaddress{5000 Forbes Ave.}
  \city{Pittsburgh}
  \state{PA}
  \postcode{15213}
  \country{USA}                    
}
\email{mcoblenz@cs.cmu.edu}          

\author{Reed Oei}
\affiliation{
  \department{Computer Science Department}              
  \institution{University of Illinois at Urbana-Champaign}            
  \streetaddress{201 North Goodwin Avenue MC 258}
  \city{Urbana}
  \state{IL}
  \postcode{61801}
  \country{USA}                    
}
\email{reedoei2@illinois.edu}          
\authornotemark[1]

\author{Tyler Etzel}
\affiliation{
  \department{Computer Science Department}              
  \institution{Cornell University}            
  \streetaddress{616 Thurston Ave.}
  \city{Ithaca}
  \state{NY}
  \postcode{14853}
  \country{USA}                    
}
\email{tyleretzel1@gmail.com}          
\authornotemark[1]

\author{Paulette Koronkevich}
\affiliation{
  \department{Computer Science Department}              
  \institution{University of British Columbia}            
  \streetaddress{2329 West Mall}
  \city{Vancouver}
  \state{BC}
  \postcode{V6T 1Z4}
  \country{Canada}                    
}
\email{paulette@koronkevi.ch}          
\authornotemark[1]

\author{Miles Baker}
\affiliation{
  \department{Computer Science Department}              
  \institution{Reed College}            
  \streetaddress{3203 SE Woodstock Blvd}
  \city{Portland}
  \state{OR}
  \postcode{97202}
  \country{USA}                    
}
\email{milesabaker@gmail.com}          
\authornotemark[1]

\author{Yannick Bloem}
\affiliation{
  \institution{Apple, Inc}            
  \streetaddress{One Apple Park Way}
  \city{Cupertino}
  \state{CA}
  \postcode{95014}
  \country{USA}   
}
\email{yannickbloem@gmail.com}          
\authornote{Work performed while at Carnegie Mellon University.}

\author{Brad A. Myers}
\orcid{0000-0002-4769-0219}             
\affiliation{
  \position{Professor}
  \department{Human-Computer Interaction Institute}              
  \institution{Carnegie Mellon University}            
  \streetaddress{5000 Forbes Ave.}
  \city{Pittsburgh}
  \state{PA}
  \postcode{15213}
  \country{USA}                    
}
\email{bam@cs.cmu.edu}          

\author{Joshua Sunshine}
\orcid{0000-0002-9672-5297}             
\affiliation{
  \position{Systems Scientist}
  \department{Institute for Software Research}              
  \institution{Carnegie Mellon University}            
  \streetaddress{5000 Forbes Ave.}
  \city{Pittsburgh}
  \state{PA}
  \postcode{15213}
  \country{USA}                    
}
\email{joshua.sunshine@cs.cmu.edu}          

\author{Jonathan Aldrich}
\orcid{0000-0003-0631-5591}             
\affiliation{
  \position{Professor}
  \department{Institute for Software Research}              
  \institution{Carnegie Mellon University}            
  \streetaddress{5000 Forbes Ave.}
  \city{Pittsburgh}
  \state{PA}
  \postcode{15213}
  \country{USA}                    
}
\email{jonathan.aldrich@cs.cmu.edu}          

\begin{abstract}
Blockchain platforms are coming into broad use for processing critical transactions among participants who have not established mutual trust. Many blockchains are programmable, supporting \textit{smart contracts}, which maintain persistent state and support transactions that transform the state. Unfortunately, bugs in many smart contracts have been exploited by hackers. Obsidian is a novel programming language with a type system that enables static detection of bugs that are common in smart contracts today. Obsidian is based on a core calculus, Silica, for which we proved type soundness. Obsidian uses \textit{typestate} to detect improper state manipulation and uses \textit{linear types} to detect abuse of assets. We describe two case studies that evaluate Obsidian's applicability to the domains of parametric insurance and supply chain management, finding that Obsidian's type system facilitates reasoning about high-level states and ownership of resources. We compared our Obsidian implementation to a Solidity implementation, observing that the Solidity implementation requires much boilerplate checking and tracking of state, whereas Obsidian does this work statically.
\end{abstract}



\maketitle

\thispagestyle{empty}

\renewcommand{\shortauthors}{Coblenz et al.}

\section{Introduction}
Blockchains have been proposed to address security and robustness objectives in contexts that lack shared trust. By recording all transactions in a tamper-resistant \textit{ledger}, blockchains attempt to facilitate secure, trusted computation in a network of untrusted peers. Blockchain programs, sometimes called \textit{smart contracts} \citep{Szabo}, can be deployed; once deployed, they can maintain state in the ledger. For example, a program might represent a bank account and store a quantity of virtual currency. Clients could conduct transactions with bank accounts by invoking the appropriate interfaces. Each transaction is appended permanently to the ledger. In this paper, we refer to a deployment of a smart contract as an \textit{object} or \textit{contract instance}.

Proponents have suggested that blockchains be used for a plethora of applications, such as finance, health care \cite{HealthCare}, supply chain management \cite{SupplyChain}, and others \cite{Elsden18:Making}. For example \citep{Provenance}, an electronics manufacturer might accept shipments of components from a variety of manufacturers; if any of those components have been replaced with fraudulent components somewhere in the chain of custody, then the manufactured systems might include defects, including security vulnerabilities. A blockchain could provide a tamper-resistant mechanism for recording signed transactions showing every entity that was ever responsible for each component.

Unfortunately, some prominent blockchain applications have included security vulnerabilities, for example through which over \$80 million worth of virtual currency was stolen \cite{DAO, CNBC}. In addition to the potentially severe consequences of bugs, platforms require that contracts be immutable, so bugs cannot be fixed easily. If organizations are to adopt blockchain environments for business-critical applications, there needs to be a more reliable way of writing smart contracts.

Many techniques promote program correctness, but our focus is on programming language design so that we can prevent bugs as early as possible --- potentially by aiding the programmer's reasoning processes before code is even written. Because of our interest in developing a language that would be effective for programmers, we designed a surface language, \textit{Obsidian}, in addition to a core calculus, \textit{Silica}. Obsidian stands for \textit{Overhauling Blockchains with States to Improve Development of Interactive Application Notation}. Our design is based on formative studies with programmers, and although those studies are not the focus of this paper, our goal of \textit{usability} drove us to focus on features that provide powerful safety guarantees while maintaining as much simplicity as possible. In this paper, we focus on the design of the language itself and make only brief mention of our observations in our user studies. For more detail regarding the user studies, readers may refer to \cite{barnaby}.

Obsidian is a programming language for smart contracts that provides strong compile-time features to prevent bugs. Obsidian is based on a novel type system that uses \textit{typestate} to statically ensure that objects are manipulated correctly according to their current states, and uses \textit{linear types} \citep{wadler1990linear} to enable safe manipulation of assets, which must not be accidentally lost. We prove key soundness theorems so that Silica can serve as a trustworthy foundation for Obsidian and potentially other typestate-oriented languages.

We make the following contributions:
\begin{enumerate}
\item We show how typestate and linear types can be combined in a user-facing programming language, using a rich but simple permission system that captures the required restrictions on aliases.
\item We show an integrated architecture for supporting both smart contracts and client programs. By enabling both on-blockchain and off-blockchain programs to be created with the same language, we ensure safety properties of the language are available for data structures that must be transferred off-blockchain as well as for those stored in the blockchain.
\item We describe Silica, the core calculus that underlies Obsidian. We prove type soundness and asset retention for Silica. Asset retention is the property that owning references to assets (objects that the programmer has designated have value) cannot be lost accidentally. Silica is the first typestate calculus (of which we are aware) that supports assets.
\item As case studies, we show how Obsidian can be used to implement a parametric insurance application and a supply chain. By comparison to Solidity, we show how leveraging typestate can move checks from run time to compile time. Our case studies were implemented by programmers who were not the designers of the language, showing that the language is usable by people other than only the designers.
\end{enumerate}

After summarizing related work, we introduce the Obsidian language with an example (\S \ref{sec:lang-intro}). Section \ref{sec:detailed-design} focuses on the design of particular aspects of the language and describes how qualitative studies influenced our design.  We describe how the language design fits into the Fabric blockchain infrastructure in \S \ref{sec:architecture}. Section \ref{silica} describes Silica, the core calculus underlying Obsidian, and its proof of soundness (though the proof itself is in the Appendix). We discuss two case studies in \S \ref{sec:evaluation}, showing how we have collaborated with external stakeholders to demonstrate the expressiveness and utility of Obsidian. Future work is discussed in \S \ref{sec:future-work}. We conclude in \S \ref{sec:conclusions}.

\section{Related Work}
\label{related-work}
Researchers have previously investigated common causes of bugs in smart contracts \cite{Luu2016, Delmolino2015:Step, atzei2016survey}, created static analyses for existing languages \cite{Kalra18:Zeus}, and worked on applying formal verification \cite{Bhargavan2016}. Our work focuses on preventing bugs in the first place by designing a language in which many commonplace bugs can be prevented as a result of properties of the type system. This enables programmers to reason more effectively about relevant safety properties and enables the compiler to detect many relevant bugs.

There is a large collection of proposals for new smart contract languages, cataloged by \citet{Harz2018:Towards}. One of the languages most closely related to Obsidian is Flint \citep{Schrans2019:Flint}. Flint supports a notion of typestate, but lacks a permission system that, in Obsidian, enables flexible, static reasoning about aliases. Flint supports a trait called \code{Asset}, which enhances safety for resources to protect them from being duplicated or destroyed accidentally. However, Flint models assets as traits rather than as linear types due to the aliasing issues that this would introduce \citep{Schrans2008:Assets}. This leads to significant limitations on assets in Flint. For example, in Flint, assets cannot be returned from functions. Obsidian addresses these issues with a permission system, and thus permits any non-primitive type to be an asset and treated as a first-class value.

There are also proposals for blockchain languages that are more domain-specific. For example, Hull et al. propose formalizing a notion of business artifacts for blockchains \citep{hull2016towards}. DAML \citep{DAML} is more schema-oriented, requiring users to write schemata for their data models. In DAML, which was inspired by financial agreements, contracts specify who can conduct and observe various operations and data. 

The problem of aliasing in object-oriented languages has led to significant research on ways to constrain and reason about aliases \citep{Clarke2013:Aliasing}. Unfortunately, these approaches can be very complex. For example, fractional permissions \citep{Boyland:2003:CIF:1760267.1760273} provide an algebra of permissions to memory cells. These permissions can be split among multiple references so that if the references are combined, one can recover the original (whole) permission. However, aside from the simple approach of reference counting, general fractional permissions have not been adopted in practical languages, perhaps because using them requires understanding a complex algebraic system. 

A significant line of research has focused on \textit{ownership types} \citep{Clarke1998:Ownership}, which refers to a different notion of ownership than we use in Obsidian. Ownership types aim to enforce \textit{encapsulation} by ensuring that the implementation of an object cannot leak outside its owner. In Obsidian, we are less concerned with encapsulation and more focused on sound typestate semantics. This allows us to avoid the strict nature of these encapsulation-based approaches while accepting their premise: typically, good architecture results in an aliasing structure in which one "owner`` of a particular object controls the object's lifetime and, likely, many of the changes to the object.

\citet{xu2017taxonomy} gives a taxonomy of blockchain systems, but the focus is on blockchain platforms, i.e. systems that maintain blockchains and process transactions. The architectural implications of Obsidian are at the application layer, not at the system layer.

Fickle \citep{Drossopoulou2001:Fickle} was a one approach to allowing objects to change class at runtime, but Fickle did not allow references to include any type specifications pertaining to the states of the referenced objects. DeLine investigated using typestate in the context of object-oriented systems \cite{DeLine2004}, finding that subclassing causes complicated issues of partial state changes; we avoid that problem by not supporting subclassing. Plaid \cite{sunshine2011first} and Plural \cite{Bierhoff:2008:PCP:1370175.1370213} are the most closely-related systems in terms of their type systems' features. Both languages were complex, and the authors noted the complexity in certain cases, e.g., fractional permissions make the language harder to use but were rarely used, and even then primarily for concurrency \cite{bierhoff2011checking}. Sunshine et al. showed typestate to be helpful in documentation when users need to understand object protocols \cite{Sunshine2014}; we used that conclusion as motivation for our language design.

\citet{Gordon2012} describes a type system for enabling safe concurrency. In addition to not supporting reasoning about object protocols, these systems sometimes introduce significant restrictions in order to handle concurrency. This is warranted in those contexts, but because blockchain systems are now always sequential, this complexity is not needed. For example, \textit{isolated} references in \citet{Gordon2012} do not allow readonly (\textit{readable}) aliases to mutable objects reachable from the isolated references; owned references in Obsidian have no such restriction because Obsidian is not designed to support concurrency.

Linear types, which facilitate reasoning about \textit{resources}, have been studied in depth since Wadler's 1990 paper \cite{wadler1990linear}, but have not been adopted in many programming languages. Rust \cite{rust} is one exception, using a form of linearity to restrict aliases to mutable objects. This limited use of linearity did not require the language to support as rich a permission system as Obsidian does; for example, Rust types cannot directly express states of referenced objects. Alms \cite{Tov:2011:PAT:1926385.1926436} is an ML-like language that supports linear types; unlike Obsidian, it is not object-oriented. Session types \cite{Caires10:Session} are another way of approaching linear types in programming languages, as in Concurrent C0 \cite{Willsey17:Design}. However, session types are more directly suited for communicating, concurrent processes, which is very different from a sequential, stateful setting as is the case with blockchains.

Silica, the core of Obsidian, is based on Featherweight Typestate (FT) \cite{Garcia:2014:FTP:2684821.2629609}. However, since Silica is designed as the core of a user-facing programming language, there are significant differences because we wanted to formalize particular operations that we found convenient in the surface language. For example, Silica replaces FT's atomic field swap with field assignment. This allows fields that temporarily have modes that differ from their declarations, facilitating a style of programming that our participants preferred in our formative user studies. This approach is related to the approach taken in \citet{naden2012type}, where fields can be unpacked, but in Silica, unpacking is only possible via the \this reference in order to maintain encapsulation.

Support for dynamic state tests is an important feature in order to facilitate practical programming. Support for these dynamic state tests has been found to be critical for expressiveness in other object-oriented contexts as well \cite{Bierhoff2009:Practical}. Unlike in FT, Silica supports expressions that execute dynamic state tests so that programs can branch according to the result.

Other differences with FT are shown in Table \ref{FT-vs-Silica}.

\begin{table}[h]
\ra{1.3}
\begin{tabular}{p{5cm} p{8cm}}
\toprule
\textbf{Featherweight Typestate} & \textbf{Silica}\\
\midrule
\texttt{pure} references cannot be used to mutate fields of referenced objects & \Unowned{} references are only as restricted as necessary for soundness: they cannot be used to change nominal state but can be used to write fields\\
No dynamic state tests & Dynamic state tests\\
Types integrate state guarantees, but do not separate state from class & Separate \contract{} and \stateExpr{} constructs \\
Inheritance & No inheritance \\
Typestate is integrated with class & Typestate implies ownership \\
No linear assets & Linear assets with explicit \disown \\ 
\bottomrule
\end{tabular}
\caption{Key differences between Featherweight Typestate and Silica.}
\label{FT-vs-Silica}
\end{table}

Although the user-centered design aspects of Obsidian are not the focus of this paper, others have had success applying these methods to tools for developers. For example, \cite{Myers2016:Programmers} argued that human-centered methods could be used in a variety of different tools for software engineers. Pane, Myers, and Miller used HCI techniques to design a programming language for children \citep{Pane2002:Using}. Stefik and Siebert used an empirical, quantitative approach regarding the design of syntax \citep{Stefik2013:Empirical}.

\section{Introduction to the Obsidian language}
\label{sec:lang-intro}

Obsidian is based on several guidelines for the design of smart contract languages identified in \citet{Coblenz2019:Smarter}. Briefly, those guidelines are:
\begin{itemize}
\item Strong static safety: bugs are particularly serious when they occur in smart contracts. In general, it can be impossible to fix bugs in deployed smart contracts because of the immutable nature of blockchains. Obsidian emphasizes a novel, strong, static type system in order to detect important classes of bugs at compile time. Among common classes of bugs is \textit{loss of assets}, such as virtual currency.
\item User-centered design: a proposed language should be as usable as possible. We integrated feedback from users in order to maximize users' effectiveness with Obsidian.
\item Blockchain-agnosticism: blockchain platforms are still in their infancies and new ones enter and leave the marketplace regularly. Being a significant investment, a language design should target properties that are common to many blockchain platforms.
\end{itemize}

We were particularly interested in creating a language that we would eventually be able to evaluate with users, while at the same time significantly improving safety relative to existing language designs. In short, we aimed to create a language that we could show was more effective for programmers. In order to make this practical, we made some relatively standard surface-level design choices that would enable our users to learn the core language concepts more easily, while using a sophisticated type system to provide strong guarantees. Where possible, we chose approaches that would enable static enforcement of safety, but in a few cases we moved checks to runtime in order to enable a simple design for users or a more precise analysis (for example, in dynamic state checks, \S \ref{sub:dynamic-state-checks}).

Typestate-oriented programming \cite{Aldrich:2009:TP:1639950.1640073} has been proposed to allow specification of protocols in object-oriented settings. For example, a \code{File} can only be read when it is in the \code{Open} state, not when it is in the \code{Closed} state. By lifting these specifications into types, typestate-oriented programming languages allow static checking of adherence to protocols and improve the ability of programmers to reason effectively about how to use the interfaces correctly \cite{Sunshine2014}. Featherweight Typestate \citep{Garcia:2014:FTP:2684821.2629609} is a core calculus for a class of typestate languages. However, we found in user studies that our early prototypes of Obsidian, which were based on a simplified version of this calculus, led to significant user confusion. In order to address these problems, we elicited requirements for blockchain languages that motivated the design of a new formalism. We designed \textit{Silica}, a new typestate calculus that, despite its simplicity, still allows users to express nearly all the properties that earlier typestate calculi enabled. Silica also supports key features that we observed users expected to have, such as dynamic state tests and field assignment.

We selected an object-oriented approach because smart contracts inevitably implement state that is mutated over time, and object-oriented programming is well-known to be a good match to this kind of situation. This approach is also a good starting point for our users, who likely have some object-oriented programming experience. However, in order to improve safety relative to traditional designs, Obsidian omits inheritance, which is error-prone due to the \textit{fragile base class problem} \cite{Mikhajlov:1998:SFB:646155.679700}. We leveraged some features of the C-family syntax, such as blocks delimited with curly braces, dots for separating references from members, etc.,  to improve learnability for some of our target users. Following blockchain convention, Obsidian uses the keyword \contract rather than \code{class}. Because of the transactional semantics of invocations on blockchain platforms, Obsidian uses the term \code{transaction} rather than \code{method}. Transactions can require that their arguments, including the receiver, be in specific states in order for the transaction to be invoked.

Since smart contracts frequently manipulate assets, such as cryptocurrencies, we designed Obsidian to support linear types \citep{wadler1990linear}, which allow the compiler to ensure that assets are neither duplicated nor lost accidentally. These linear types integrate consistently with typestate, since typestate-bearing references are affine (i.e. cannot be duplicated but can be dropped as needed). A particular innovation in this approach is the fusion of linear references to assets with affine references to non-assets. Whether a reference is linear or affine depends on the declaration of the type to which the reference refers.

The example in Fig. \ref{tiny-vending-machine} shows some of the key features of Obsidian. \code{TinyVendingMachine} is a \code{main} contract, so it can be deployed independently to a blockchain. A \code{TinyVendingMachine} has a very small inventory: just one candy bar. It is either \code{Full}, with one candy bar in inventory, or \code{Empty}. Clients may invoke \code{buy} on a vending machine that is in \code{Full} state, passing a \code{Coin} as payment. When \code{buy} is invoked, the caller must initially \textit{own} the \code{Coin}, but after \code{buy} returns, the caller no longer owns it. \code{buy} returns a \code{Candy} to the caller, which the caller then owns. After \code{buy} returns, the vending machine is in state \code{Empty}. 

\begin{figure}
\begin{lstlisting}[numbers=left, framexleftmargin=1em, xleftmargin=2em, basicstyle=\footnotesize\ttfamily]
// This vending machine sells candy in exchange for candy tokens.
main asset contract TinyVendingMachine {
    // Fields defined at the top level are in scope in all states.
    Coins @ Owned coinBin; 

    state Full {
        // inventory is only in scope when the object is in Full state.
        Candy @ Owned inventory; 
    }
    state Empty; // No candy if the machine is empty.

    TinyVendingMachine() {
        coinBin = new Coins(); // Start with an empty coin bin.
        ->Empty;
    }

    transaction restock(TinyVendingMachine @ Empty >> Full this,
                        Candy @ Owned >> Unowned c) {
        ->Full(inventory = c);
    }

    transaction buy(TinyVendingMachine @ Full >> Empty this,
                    Coin @ Owned >> Unowned c) returns Candy @ Owned {
        coinBin.deposit(c);
        Candy result = inventory;
        ->Empty;
        return result;
    }

    transaction withdrawCoins() returns Coins @ Owned {
        Coins result = coinBin;
        coinBin = new Coins();
        return result;
    }
}
\end{lstlisting}
\caption{A tiny vending machine implementation, showing key features of Obsidian.}
\label{tiny-vending-machine}
\end{figure}

Smart contracts commonly manipulate \textit{assets}, such as virtual currencies. Some common smart contract bugs pertain to accidental loss of assets \cite{Delmolino2015:Step}. If a contract in Obsidian is declared with the \code{asset} keyword, then the type system requires that every instance of that contract have exactly one owner. This enables the type checker to report an error if an owned reference goes out of scope. For example, assuming that \code{Coin} was declared as an \code{asset}, if the author of the \code{buy} transaction had accidentally omitted the \code{deposit} call, the type checker would have reported the loss of the asset in the \code{buy} transaction. Any contract that has an \code{Owned} reference to another asset must itself be an asset.

To enforce this, references to objects have types according to both the \texttt{contract} of the referenced object and a \textit{mode}, which denotes information about ownership. Modes are separated from contract names 
 with an \code{@} symbol. Exactly one reference to each asset contract instance must be \Owned; this reference must not go out of scope. For example, an owned reference to a \code{Coin} object can be written \code{Coin@Owned}. Ownership can be transferred between references via assignment or transaction invocation. The compiler outputs an error if a reference to an \code{asset} goes out of scope while it is \Owned. Ownership can be explicitly discarded with the \code{disown} operator.

\Unowned{} is the complement to \Owned{}: an object has at most one \Owned{} reference but an arbitrary number of \Unowned{} references. \Unowned{} references are not linear, as they do not convey ownership. They are nonetheless useful. For example, a \code{Wallet} object might have owning references to \code{Money} objects, but a \code{Budget} object might have \Unowned{} aliases to those objects so that the value of the \code{Money} can be tracked (even though only the \code{Wallet} is permitted to transfer the objects to another owner). Alternatively, if there is no owner of a non-asset object, it may have \Shared{} and \Unowned{} aliases. Examples of these scenarios are shown in Fig. \ref{aliasing} to provide some intuition.

In Obsidian, the \textit{mode} portion of a type can change due to operations on a reference, so transaction signatures can specify modes both before and after execution. As in Java, a first argument called \code{this} is optional; when present, it is used to specify initial and final modes on the receiver. The $\gg$ symbol separates the initial mode from the final one. In the example above, the signature of \code{buy} (lines 22-23) indicates that \code{buy} must be invoked on a \code{TinyVendingMachine} that is statically known to be in state \code{Full}, passing a \code{Coin} object that the caller owns. When \code{buy} returns, the receiver will be in state \code{Empty} and the caller will no longer have ownership of the \code{Coin} argument.

Obsidian contracts can have constructors (line 10 above), which initialize fields as needed. If a contract has any states declared, then every instance of the contract must be in one of those states from the time each constructor exits.

\begin{figure}
\includegraphics[width=13cm]{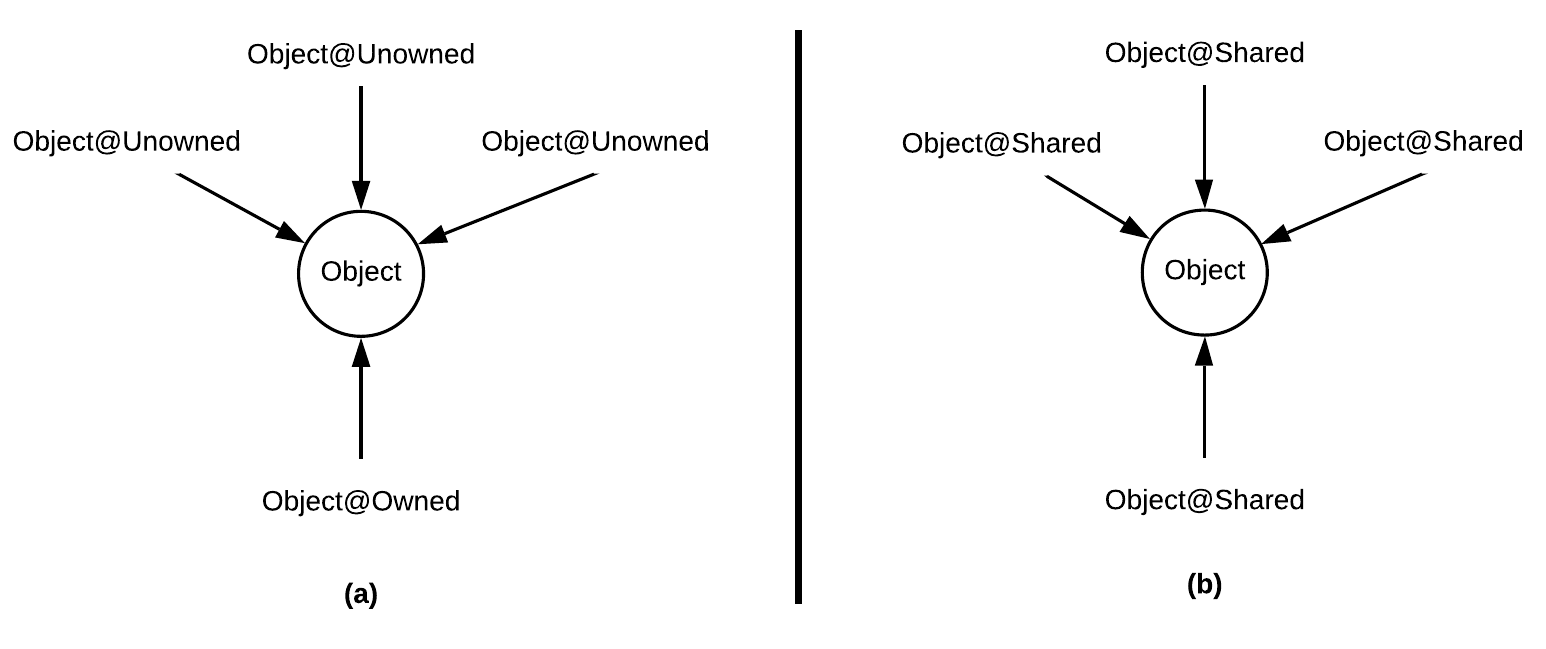}
\caption{Some common aliasing scenarios. (a) shows an object with one owner; (b) shows a shared object.}
\label{aliasing}
\end{figure}

Objects in smart contracts frequently maintain high-level state information \citep{Solidity-Patterns}, with the set of permitted transactions depending on the current state. For example, a \code{TinyVendingMachine} might be \code{Empty} or \code{Full}, and the \code{buy} transaction can only be invoked on a \code{Full} machine. Prior work showed that including state information in documentation helped users understand how to use object protocols \citep{Sunshine2014}, so we include first-class support for states in Obsidian. \textit{Typestate} \citep{Aldrich:2009:TP:1639950.1640073} is the idea of including state information in types, and we take that approach in Obsidian so that the compiler can ensure that objects are manipulated correctly according to their states. State information can be captured in a mode. For example, \code{TinyVendingMachine@Full} is the type of a reference to an object of contract \code{TinyVendingMachine} with mode \code{Full}. In this case, the mode denotes that the referenced object is statically known to be in state \code{Full}. 

State is mutable; objects can transition from their current state to another state via a transition operation. For example, \code{->Full(inventory = c)} might change the state of a \code{TinyVendingMachine} to the \code{Full} state, initializing the \code{inventory} field of the \code{Full} state to \code{c}. This leads to a potential difficulty: what if a reference to a \code{TinyVendingMachine} with mode \code{Empty} exists while the state transitions to \code{Full}? To prevent this problem, typestate is only available with references that also have ownership. Because of this, there is no need to separately denote ownership in the syntax; we simply observe that every typestate-bearing reference is also owned. Then, Obsidian restricts the operations that can be performed through a reference according to the reference's mode. In particular, if an owned reference might exist, then non-owning references cannot be used to mutate typestate. If no owned references exist, then all references permit state mutation. A summary of modes is shown in Table \ref{modes}.

\begin{table}[ht]
\renewcommand{\arraystretch}{1.2}
\begin{tabular}{@{}l p{8.5cm} p{2 cm}@{}}\toprule
\textbf{Mode} & \textbf{Meaning} & \textbf{Typestate mutation}\\
\midrule
\Owned & This is the only reference to the object that is owned. There may be many \Unowned{} aliases but no \Shared aliases. & Permitted\\
\Unowned & There may or may not be any owned aliases to this object, but there may be many other \Unowned or \Shared aliases. & Forbidden\\
\Shared & This is one of potentially many \Shared references to the object. There are no owned aliases. & Permitted \\
\textit{state name(s)} & This is an owned reference and also conveys the fact that the referenced object is in one of the specified states. There may be \Unowned aliases but no \Shared or \Owned aliases. & Permitted\\
\bottomrule
\end{tabular}
\caption{A summary of modes in Obsidian}
\label{modes}
\end{table}

\section{Obsidian language design process and details}
\label{sec:detailed-design}
Obsidian is the first object-oriented language (of which we are aware) to integrate linear assets and typestate. This combination --- and, in fact, even just including typestate --- could result in a design that was hard to use, since typical typestate languages require users to understand a complex permissions model. In designing the language, we focused on \textit{simplicity} in service of usability \citep{Coblenz18:Interdisciplinary}. We maintained static safety where possible, but moved certain checks to runtime where needed to maintain a high level of expressiveness. We also aimed to simplify the job of the programmer relative to existing blockchain programming languages by eliminating onerous, error-prone programming tasks, such as writing serialization and deserialization code. In this section, we describe how we designed language features to improve user experience, in some cases driven by results of formative user studies \citep{barnaby}. Some other system features, such as serialization, are discussed in \S \ref{architecture}. Rather than relying only on our own experience and intuition, we invited participants to help us assess the tradeoffs of different design options. This enabled us to take a more data-driven approach in our language design, as suggested by \citet{LanguageWars} and  \citet{Coblenz18:Interdisciplinary}. We take the perspective that we should integrate \textit{qualitative} methods in addition to quantitative methods in order to drive language design in a direction that is more likely to be beneficial for users.

\subsection{Type declarations, annotations, and static assertions}
Obsidian requires type declarations of local variables, fields, and transaction parameters. In addition to providing familiarity to programmers who have experience with other object-oriented languages, there is a hypothesis that these declarations may aid in usability by providing documentation, particularly at interfaces \cite{Coblenz2014:Considering}. Traditional declarations are also typical in prior typestate-supporting languages, such as Plaid \cite{sunshine2011first}. Unfortunately, typestate is incompatible with the traditional semantics of type declarations: programmers normally expect that the type of a variable always matches its declared type, but mutation can result in the typestate no longer matching the initial type of an identifier. This violates the \textit{consistency} usability heuristic \citep{nielsen1990heuristic} and is a potential source of reduced code readability, since determining the type of an identifier can require reading all the code from the declaration to the program point of interest. 

To alleviate this problem, we introduced \textit{static assertions}. These have the syntax \code{[e @ mode]}. For example, \code{[account @ Open]} statically asserts that the reference \code{account} is owned and refers to an object that the compiler can prove is in \code{Open} state. Furthermore, to avoid confusion about the meanings of local variable declarations, Obsidian forbids mode specifications on local variable declarations.

Static assertions have no implications on the dynamic semantics (and therefore have no runtime cost); instead, they serve as checked documentation. The type checker verifies that the given mode is valid for the expression in the place where the assertion is written. A reader of a typechecked program can be assured, then, that the specified types are correct, and the author can insert the assertions as needed to improve program understandability.

\subsection{State transitions}
Each state definition can include a list of fields, which are in scope only when the object is in the corresponding state (see line 8 of Fig. \ref{tiny-vending-machine}). What, then, should be the syntax for initializing those fields when transitioning to a different state? Some design objectives included:
\begin{itemize}
\item When an object is in a particular state, the fields for that state should be initialized.
\item When an object is \textit{not} in a particular state, the fields for that state should be out of scope.
\item According to the \textit{user control and freedom} heuristic \citep{nielsen1990heuristic} and results by Stylos et al. \citep{Stylos2007}, programmers should be able to initialize the fields in any order, including by assignment. Under this criterion, it does not suffice to only permit constructor-style simultaneous initialization.
\end{itemize}

In order to allow maximum user flexibility without compromising the integrity of the type system, we implemented a flexible approach. When a state transition occurs, all fields of the target state must be initialized. However, they can be initialized either \textit{in} the transition (e.g., \code{->S(x = a)} initializes the field \code{x} to \code{a}) or \textit{prior to} the transition (e.g., \code{S::x = a; ->S}). In addition, fields that are in scope in the current state but will not be in scope in the target state must \textit{not} be owned references to assets. In those cases, ownership must be transferred to another reference or discarded before the transition.

\subsection{Transaction scope}
Transactions in Obsidian are only available when the object is in a particular state. Correspondingly, other typestate-oriented languages support defining methods inside states. For example, Plaid \citep{sunshine2011first} allows users to define the \code{read} method inside the \code{OpenFile} state to make clear that \code{read} can only be invoked when a \code{File} is in the \code{OpenFile} state. However, this is problematic when methods can be invoked when the object is in several states.

\citet{barnaby} considered this question for Obsidian and observed that study participants, who were given a typestate-oriented language that included methods in states, asked many questions about what could happen during and after state transitions. They were unsure what \code{this} meant in that context and what variables were in scope at any given time. One participant thought it should be disallowed to call transactions available in state \code{S1} while writing a transaction that was lexically in state \code{Start}. For this reason, we designed Obsidian so that transactions are defined lexically \textit{outside} states. Transaction signatures indicate (via type annotations on a first argument called \code{this}) from which states each transaction can be invoked. This approach is consistent with other languages, such as Java, which also allows type annotations on a first argument \code{this}.

\subsection{Field type consistency}
In traditional object-oriented languages, fields always reference either \code{null} or objects whose types are subtypes of the fields' declared types. This presents a difficulty for Obsidian, since the mode is part of the type, and the mode can change with operations. For example, a \code{Wallet} might have a reference of type \code{Money@Owned}. How should a programmer implement \code{swap}? One way is shown in Fig. \ref{transition-obsidian}.

\begin{figure}[h]
\begin{lstlisting}[numbers=left, framexleftmargin=1em, xleftmargin=1em]
contract Wallet {
	Money@Owned money;
	
	transaction swap (Money @ Owned m) returns Money @ Owned {
		Money result = money;
		money = m;
		return result;
	}
}
\end{lstlisting}
\caption{Obsidian's approach for handling transitions.}
\label{transition-obsidian}
\end{figure}

The problem is that line 5 changes the type of the \code{money} field from \code{Owned} to \code{Unowned} by transferring ownership to \code{result}. Should this be a type error, since it is inconsistent with the declaration of \code{money}? If it is a type error, how is the programmer supposed to implement \code{swap}? One possibility is to add another state, as shown in Fig. \ref{transition-alternative}.

\begin{figure}[h]
\begin{lstlisting}[numbers=left, framexleftmargin=1em, xleftmargin=1em]
contract Wallet {
	state Empty;
	state Full {
		Money @ Owned money;
	}
	
	transaction swap (Wallet@Full this, Money @ Owned m) 
	                 returns Money @ Owned 
	{
		// Suppose the transition returns the contents of the old field.
		Money result = ->Empty; 
		->Full(money = m);
		return result;
	}
}
\end{lstlisting}
\caption{An alternative approach for handling transitions.}
\label{transition-alternative}
\end{figure}

Although this approach might seem like a reasonable consequence of the desire to keep field values consistent with their types, it imposes a significant burden. First, the programmer is required to introduce additional states, which leaks implementation details into the interface (unless we mitigate this problem by making the language more complex, e.g., with \code{private} states or via abstraction over states). Second, this requires that transitions return the newly out-of-scope fields, but it is not clear how: should the result be of record type? Should it be a tuple? What if the programmer neglects to do something with the result? Plaid \citep{sunshine2011first} addressed the problem by not including type names in fields, but that approach may hamper code understandability \citep{Coblenz2014:Considering}. 

In Obsidian, we permit fields to \textit{temporarily} reference objects that are not consistent with the fields' declarations, but we require that at the end of transactions (and constructors), the fields refer to appropriately-typed objects. This approach is consistent with the approach for local variables, with the additional postcondition of type consistency. Both local variables and fields of nonprimitive type, and transaction parameters must always refer to instances of appropriate contracts; the only discrepancy permitted is of mode. Obsidian forbids re-assigning formal parameters to refer to other objects to ensure soundness of this analysis.

Re-entrancy imposes a significant problem here: re-entrant calls from the middle of a transaction's body, where the fields may not be consistent with their types, can be dangerous, since the called transactions are supposed to be allowed to assume that the fields reference objects consistent with the fields' types. To address this, Obsidian distinguishes between \textit{public} and \textit{private} transactions. Obsidian forbids re-entrant calls to \textit{public} transactions at the object level of granularity (i.e., a transaction cannot invoke another transaction on an object for which there is already an invocation on the stack). The Obsidian runtime detects illegal re-entrant calls and aborts transactions that attempt them. However, to facilitate helper transactions, Obsidian also supports \code{private} transactions, which declare the expected types of the fields before and after the invocation. For example:
\begin{lstlisting}[numbers=none, framexleftmargin=0em, xleftmargin=0em]
contract AContract {
    state S1;
    state S2;
	
    AContract@S1 c;
    private (AContract@S2 >> S1 c) transaction t1() {#\ldots#}
}
\end{lstlisting}

Transaction \code{t1} may only be invoked by transactions of \code{AContract}, only on \code{this}, and only when \code{this.c} temporarily has type \code{AContract@S2}. When \code{t1} is invoked, the typechecker checks to make sure field \code{c} has type \code{C@S2}, and assumes that after \code{t1} returns, \code{c} will have type \code{AContract@S1}. Of course, the body of \code{t1} is checked assuming that \code{c} has type \code{C@S2} to make sure that afterward, \code{c} has type \code{C@S1}.

Avoiding unsafe re-entrancy has been shown to be important for real-world smart contract security, as millions of dollars were stolen  in the DAO hack via a re-entrant call exploit \cite{DAO-details}.

\begin{figure}[t]
\begin{lstlisting}[numbers=left, framexleftmargin=1em, xleftmargin=1em, basicstyle=\footnotesize\ttfamily]
main asset contract GiftCertificate {
    Date @ Unowned expirationDate;

    state Active {
        Money @ Owned balance;
    }

    state Expired;
    state Redeemed;

    GiftCertificate(Money @ Owned >> Unowned b, Date @ Unowned d) 
    {
        expirationDate = d;
        ->Active(balance = b);
    }

    transaction checkExpiration(GiftCertificate @ Active >> (Active | Expired) this) 
    {
        if (getCurrentDate().greaterThan(expirationDate)) {
            disown balance;
            ->Expired;
        }
    }
    transaction redeem(GiftCertificate @ Active >> (Expired | Redeemed) this)
                returns Money@Owned 
    {
        checkExpiration();

        if (this in Active) {
            Money result = balance;
            ->Redeemed;
            return result;
        }
        else {
            revert "Can't redeem expired certificate";
        }
    }
    transaction getCurrentDate(GiftCertificate @ Unowned this) 
                returns Date @ Unowned 
    {
        return new Date();
    }
}
\end{lstlisting}
\caption{A dynamic state check example.}
\label{GiftCertificate}
\end{figure}

\subsection{Dynamic State Checks}
\label{sub:dynamic-state-checks}
The Obsidian compiler enforces that transactions can only be invoked when it can prove statically that the objects are in appropriate states according to the signature of the transaction to be invoked. In some cases, however, it is impossible to determine this statically. For example, consider \code{redeem} in Fig. \ref{GiftCertificate}. After line 24, the contract may be in either state \code{Active} or state \code{Expired}. However, inside the dynamic state check block that starts on line 29, the compiler assumes that \code{this} is in state \code{Active}. The compiler generates a dynamic check of state according to the test. However, regarding the code in the block, there are two cases. If the dynamic state check is of an \Owned reference $x$, then it suffices for the type checker to check the block under the assumption that the reference is of type according to the dynamic state check. However, if the reference is \Shared, there is a problem: what if code in the block changes the state of the object referenced by $x$? This would violate the expectations of the code inside the block, which is checked as if it had ownership of $x$. We consider the cases, since the compiler always knows whether an expression is \Owned, \Unowned, or \Shared:
\begin{itemize}
\item If the expression to be tested is a variable with \Owned mode, the body of the if statement can be checked assuming that the variable initially references an object in the specified state, since that code will only execute if that is the case due to the dynamic check.
\item If the expression to be tested is a variable with \Unowned mode, there may be another owner (and the variable cannot be used to change the state of the referenced object anyway). In that case, typechecking of the body of the if proceeds as if there had been no state test, since it would be unsafe to assume that the reference is owned. However, this kind of test can be useful if the desired behavior does not statically require that the object is in the given state. For example, in a university accounting system, if a \code{Student} is in \code{Enrolled} state, then their account should be debited by the cost of tuition this semester. The \code{debit} operation does not directly depend on the student's state; the state check is a matter of policy regarding who gets charged tuition.
\item If the expression to be tested is a variable with \Shared mode, then the runtime maintains a state lock that pertains to other shared references. The body is checked initially assuming that the variable owns a reference to an object in the specified state. Then, the type checker verifies that the variable still holds ownership at the end and that the variable has not been re-assigned in the body. However, at runtime, if any \textit{other} \Shared reference is used to change the state of the referenced object (for example, via another alias used in a transaction that is invoked by the body of the dynamic state check block), then the  transaction is aborted (recall that the blockchain environment is sequential, so there is only one top-level transaction in progress at a time). This approach enables safe code to complete but ensures that the analysis of the type checker regarding the state of the referenced object remains sound. This approach also bears low runtime cost, since the cost of the check is borne only in transitions via \Shared references. An alternative design would require checks at invocations to make sure that the referenced object was indeed in the state the typechecker expected, but we expect our approach has significantly lower runtime cost. Furthermore, our approach results in errors occurring immediately on transition. The alternative approach would give errors only when the referenced object was used, which could be substantially after the infringing transition, which would require the programmer to figure out which transition caused the bug.
\item If the expression to be tested is not a variable, the body of the if statement is checked in the same static context as the if statement itself. It would be unsafe for the compiler to make any assumptions about the type of future executions of the expression, since the type may change. This case only occurs in Obsidian, not in the underlying Silica formalism, which is in A-normal form \citep{Sabry1992:Reasoning}.
\end{itemize}

The dynamic state check mechanism is related to the \textit{focusing} mechanism of \citet{Fahndrich2002:Adoption}. Unlike focusing, Obsidian's dynamic state checks detect unsafe uses of aliases precisely rather than conservatively, enabling many more safe programs to typecheck. Furthermore, Obsidian does not require the programmer to specify \textit{guards}, which in focusing enable the compiler to reason conservatively about which references may alias.

\subsection{Parametric Polymorphism}
Parametric polymorphism is particularly important for Obsidian in order to maintain safety of collections and avoid needless code duplication. Requiring users to cast objects retrieved from containers to appropriate type would defeat the point of the language, which is to provide strong static guarantees, since those casts would have to be checked dynamically. Furthermore, there would have to be separate containers for different modes, since a container's elements would need to be either \Unowned, \Shared, or \Owned. In Obsidian, a contract can have \textit{two} type parameters: one for a contract and one for a mode. For example, part of the polymorphic LinkedList implementation is as follows:
\begin{lstlisting}[numbers=left, framexleftmargin=1em, xleftmargin=1em, basicstyle=\footnotesize\ttfamily]
contract LinkedList[T@s] {
    state Empty;
    state HasNext {
        LinkedList[T@s]@Owned next;
        T@s value;
    }
    transaction append(LinkedList@Owned this, T@s >> Unowned obj) {
    	#\ldots#
	}
}
\end{lstlisting}

Line 1 shows that the contract type is parameterized by the contract variable \code{T}, and the mode is parameterized by the mode variable \code{s}. In line 4, the \code{next} field is an \Owned reference to an object of type \code{LinkedList[T@s]} -- that is, a node whose type parameters are the same as the containing contract's type parameters. An object of type \code{LinkedList[Money@Owned]} is a container that holds a list of \code{Money} references, each of which the container owns. Using a separate parameter for the mode allows parameterization over states, e.g. a \code{LinkedList[LightSwitch@On]} owns references to \code{LightSwitch} objects that are each in the \code{On} state. In line 7, appending an element to a LinkedList always takes any ownership that was given, and the parameter \code{obj} must conform to the type specified by the type parameter \code{T@s}.

\section{System design and implementation}
\label{sec:architecture}
Our current implementation of Obsidian supports Hyperledger Fabric \cite{Fabric}, a permissioned blockchain platform. In contrast to public platforms, such as Ethereum, Fabric permits organizations to decide who has access to the ledger, and which peers need to approve (\textit{endorse}) each transaction. This typically provides higher throughput and more convenient control over confidential data than public blockchains, allowing operators to trade off between distributed trust and high performance. Fabric supports smart contracts implemented in Java, so the Obsidian compiler translates Obsidian source code to Java for deployment on Fabric peer nodes. The Obsidian compiler prepares appropriately-structured directories with Java code and a build file. Fabric builds and executes the Java code inside purpose-build Docker containers that run on the peer nodes. The overall Obsidian compiler architecture is shown in Fig. \ref{architecture}.

\begin{figure}
\includegraphics[width=\textwidth]{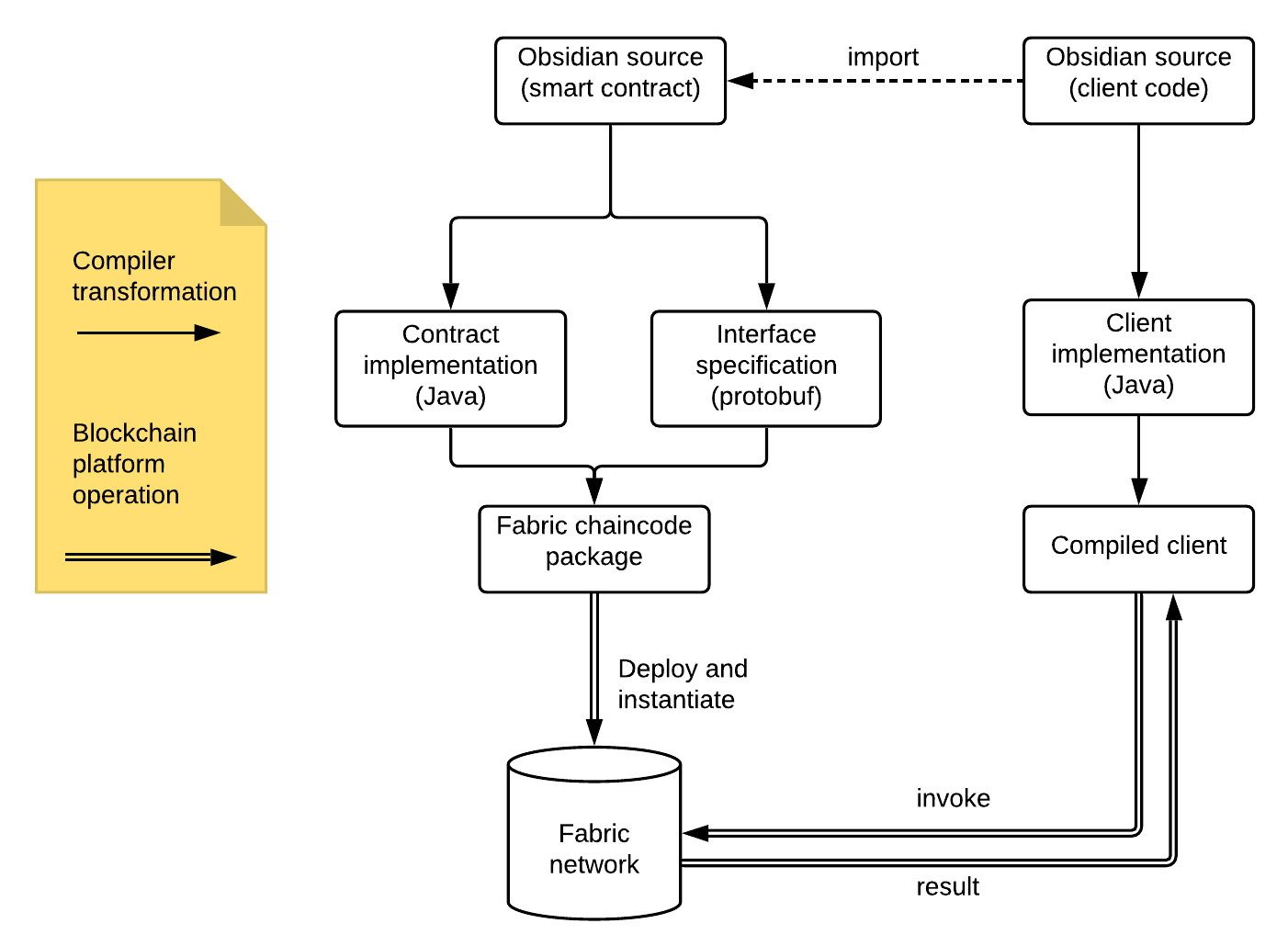}
\caption{Obsidian system architecture}
\label{architecture}
\end{figure}

\subsection{Storage in the ledger}
Fabric provides a key/value store for persisting the state of smart contracts in the ledger. As a result, Fabric requires that smart contracts serialize their state in terms of key/value pairs. In other smart contract languages, programmers are required to manually write code to serialize and deserialize their smart contract data. In contrast, Obsidian automatically generates serialization code, leveraging \textit{protocol buffers} \citep{protobuf} to map between message formats and sequences of bytes. When a transaction is executed, the appropriate objects are lazily loaded from the key/value store as required for the transaction's execution. Lazy loading is \textit{shallow}: the object's fields are loaded, but objects that fields reference are not loaded until \textit{their} fields are needed. After executing the transaction, Obsidian's runtime environment automatically serializes the modified objects and saves them in the ledger. This means that aborting a transaction and reverting any changes is very cheap, since this entails \textit{not} setting key/value pairs in the store, flushing the heap of objects that have been lazily loaded, and (shallowly) re-loading the root object from the ledger. This lazy approach decreases execution cost and frees the programmer from needing to manually load and unload key/value pairs from the ledger, as would normally be required on Fabric.

\subsection{Obsidian client programs}
\label{client-programs}
The convention for most blockchain systems is that smart contracts are written in one language, such as Solidity, and client programs are written in a different language, such as JavaScript. Unfortunately, in Solidity, transaction arguments and outputs must be primitives, not objects; arrays of bytes can be transferred, but the client and server must each implement corresponding serialization and deserialization code. The interface for a given contract is specified in an Application Binary Interface (ABI), documented in a schema written in JavaScript. If there are any incompatibilities between the semantics of the JavaScript serialization code and the semantics of the Solidity contract that interprets the serialized message, there can be bugs. 

Obsidian addresses this problem by allowing users to write client programs in Obsidian. Client programs can reference the same contract implementations that were instantiated on the server, obviating the need for two different implementations of data structures. Clients use the same automatically-generated serialization and deserialization code that the server does. As a result, Obsidian permits arbitrary objects (encoded via protocol buffers) to be passed as arguments and returned from transactions. Since the protocol buffer specifications are emitted by the Obsidian compiler, any client (even non-Obsidian clients) can use these specifications to correctly serialize and deserialize native Obsidian objects in order to invoke Obsidian transactions and interpret their results. 

The Obsidian client program has a \code{main} transaction, which takes a \code{remote} reference. The keyword \code{remote}, which modifies types of object references, indicates that the type refers to a remote object. The compiler implements remote references with stubs, via an RMI-like mechanism. When a non-\code{remote} reference is passed as an argument to a remote transaction, the referenced object is serialized and sent to the blockchain. Afterward, the reference becomes a \code{remote} reference, so that only one copy of the object exists (otherwise mutations to the referenced object on the client would not be reflected on the blockchain, resulting in potential bugs). This change in type is similar to how reference modes change during execution. Fig. \ref{client} shows a simple client program that uses the \code{TinyVendingMachine} above. The \code{main} transaction takes a \code{remote} reference to the smart contract instance.

Every Obsidian object has a unique ID, and references to objects can be transmitted between clients and the blockchain via object ID. There is some subtlety in the ID system in Obsidian: all blockchain transactions must be deterministic so that all peers generate the same IDs, so it is impossible to use traditional (e.g., timestamp-based or hardware-based) UUID generation. Instead, Obsidian bases IDs on transaction identifiers, which Fabric provides, and on an index kept in an ID factory. Since transaction IDs are unique, each transaction can have its own ID factory and still avoid collisions. The initial index is reset to zero at the beginning of each transaction so that no state pertaining to ID generation needs to be stored between transactions. Blockchains provide a sequential execution environment, so there is no need to address race conditions in ID generation. When clients instantiate contracts, they generate IDs with a traditional UUID algorithm, since clients operate off the blockchain.

Although serializing objects according to their Protobuf specifications is better than requiring programmers to manually write their own serialization code, if a client is written in a traditional language, the client does not obtain the safety benefits of the Obsidian type system. Obsidian addresses this problem by tracking which objects are owned by some client. Although the blockchain cannot know which client owns each object or prevent non-Obsidian clients from duplicating or losing assets that are owned by clients, the Obsidian blockchain runtime aborts transactions that attempt to pass ownership from client to blockchain of assets that no client owns.

Blockchains allow clients to interleave their transactions arbitrarily. This does not suffice to ensure safety in arbitrary Obsidian client programs because client programs need to rely on state information that they obtain dynamically. The current implementation of Obsidian assumes that Obsidian clients will not race with other clients. In the future, however, Obsidian will address this issue in a platform-appropriate manner. Once the programmer identifies a critical section, one approach is for the client to wrap the section in a lambda so that the server can execute it in one transaction. This approach might work well on Ethereum, where clients must pay for the costs of executing code on the blockchain. However, on Fabric, this approach is problematic because the security policy is such that clients should not force the blockchain to execute arbitrary code (for example, including non-terminating code). An approach that may be more effective is to use \textit{optimistic concurrency} \citep{Kung1981:On}, in which smart contracts on the blockchain defer commitment of changes from clients until the client's critical section is done; then, either the transaction is committed, or the changes are discarded because of intervening changes that occurred.


\begin{figure}[ht]
\begin{lstlisting}[numbers=left, framexleftmargin=1em, xleftmargin=2em, basicstyle=\footnotesize\ttfamily]
import "TinyVendingMachine.obs"

main contract TinyVendingMachineClient {
    transaction main(remote TinyVendingMachine@Shared machine) {
        restock(machine);

        if (machine in Full) {
            Coin c = new Coin();
            remote Candy candy = machine.buy(c);
            eat(candy);
        }
    }

    private transaction restock(remote TinyVendingMachine@Shared machine) {
        if (machine in Empty) {
            Candy candy = new Candy();
            machine.restock(candy);
        }
    }

    private transaction eat(remote Candy @ Owned >> Unowned c) {
        disown c;
    }
}
\end{lstlisting}
\caption{A simple client program, showing how clients reference a smart contract on the blockchain. Note that the blockchain-side smart contract has been modified (relative to Fig. \ref{tiny-vending-machine}) to have \Shared receivers, since top-level objects are never owned by clients.}
\label{client}
\end{figure}

\subsection{Ensuring safety with untrusted clients}
If a client program is written in a language other than Obsidian, it may not adhere to Obsidian's type system. For example, a client program may obtain an owned reference to an object and then attempt to transfer ownership of that object to multiple references on the blockchain. This is called the \textit{double-spend} problem on blockchains: a program may attempt to consume a resource more than once. To address this problem, the Obsidian runtime keeps a list of all objects for which ownership has been passed outside the blockchain. When a transaction is invoked on an argument that must be owned, the runtime aborts the transaction if that object is not owned outside the blockchain, and otherwise removes the object from the list. Likewise, when a transaction argument or result becomes owned by the client after the transaction (according to the transaction's signature), the runtime adds the object to the list. Of course, Obsidian has no way of ensuring safe manipulation of owned references in non-Obsidian clients, but this approach ensures that each time an owned reference leaves the blockchain, it only returns once, preventing double-spending attacks. Obsidian cannot ensure that non-Obsidian clients do not lose their owned references, so we hope that most client code that manipulates assets will be written in Obsidian.

%

\section{Silica}
\label{silica}
In this section, we describe Silica, the core calculus that forms a foundation for Obsidian. Silica is so named because silica comprises 70\% or more of obsidian glass \citep{Obsidian-chemical}. Silica is designed in the style of Featherweight Typestate \citep{Garcia:2014:FTP:2684821.2629609}, which is itself designed in the style of Featherweight Java \citep{Igarashi:2001:FJM:503502.503505}. Since Obsidian is a more traditional object-oriented, imperative language, the syntax, type checker, etc. implemented in the system differs slightly from the rules for Silica.

Silica uses A-normal form \citep{Sabry1992:Reasoning} as a simplification to avoid nested expressions in most cases. The production $s$ in the grammar stands for a reference to an object, which can then be used to build other expressions. In our initial explanation, $s$ ranges over variables, but in \S \ref{dynamic-semantics} we will introduce \textit{locations}, which facilitate our proof of correctness. 

Silica differs significantly from Featherweight Typestate (FT). Silica avoids class-level inheritance to simplify reasoning about programs. In order to formalize expression of (a) fields that are common to more than one state; and (b) a type system that is aware of all possible (nominal) states of a particular object, Silica defines a notion of \textit{state} in addition to a notion of \textit{contract}. FT only has a notion of \textit{class} and expects the programmer to simulate states by specifying multiple classes that interoperate. Unlike FT, Silica permits the typestate of a field to differ temporarily from its declaration, as long as consistency is restored and not visible outside the contract. This facilitates patterns of use we saw our participants use in user studies. 

Silica fuses the notions of typestate and permission into one type construct, unlike FT, which has separate notions of permission and state guarantee. With this approach, the syntax of Silica exactly expresses the set of possible reference types. Silica also distinguishes between asset contracts and non-asset contracts; owning references to asset contracts are treated \textit{linearly} rather than in an affine way. FT has no way of treating references linearly.

Silica supports parametric polymorphism, a key feature requested by our industrial stakeholders. Although this makes the language more complex, we think this complexity is outweighed by the benefit of the feature, enabling (for example) reusable containers.

As a result, although some aspects of the system are more complex than FT, Silica is more expressive in the above respects. Silica serves as a sound foundation for Obsidian but could be used or adapted for other typestate-oriented languages.

Fig. \ref{silica-syntax} shows the syntax of Silica.

\begin{figure}[htp]
\begin{align*}
C &\in \textsc{ContractNames} & m &\in \textsc{transactionNames} \\
I &\in \textsc{InterfaceNames} & S &\in \textsc{StateNames}\\
D &\in \textsc{ContractNames} \cup \textsc{InterfaceNames} & p &\in \textsc{PermissionVariables}\\
X &\in \textsc{DeclarationVariables} & f &\in \textsc{FieldNames} \\
x &\in \textsc{IdentifierNames} \\
\end{align*}

\begin{tabular}{l r l l}
T & \bnfdef & T\textsubscript{C}.T\textsubscript{ST} & (types of contract references)\\

T\textsubscript{C} & \bnfdef & \generics{D}{T} & (types of concrete contracts/interfaces)\\
&   \bnfalt & X & (declaration variables) \\

T\textsubscript{ST} & \bnfdef & \textoverline{S}  & (nonempty disjunction of states)\\
& 	\bnfalt & p & (permission/state variables)\\
&  	\bnfalt & P & \\

P &	\bnfdef & \Owned \bnfalt \Unowned \bnfalt \Shared & \\

T\textsubscript{G} & \bnfdef & \genericParamOpt{X}{p}{\generics{I}{T}.T\textsubscript{ST}} & (generic type parameter) \\

CON & \bnfdef & \contract \ \generics{C}{T_G}\ \implements\ \generics{I}{T} \{ \textoverline{ST} \textoverline{M} \} \\

IFACE & \bnfdef & \interface\ \generics{I}{T_G} \{ \textoverline{ST}  \textoverline{M\textsubscript{SIG}} \} & \\

ST & \bnfdef & [\textbf{asset}] S  \textoverline{F} \\

F & \bnfdef & T f \\

M\textsubscript{SIG} & \bnfdef & T \generics{m}{T_G} (\textoverline{T \trans T\textsubscript{ST} x}) T\textsubscript{ST} \trans \ T\textsubscript{ST} & (arguments cannot change class)\\
&	\bnfalt  & { \textoverline{T\textsubscript{ST} \trans T\textsubscript{ST} f}} T \generics{m}{T_G}(\textoverline{T \trans T\textsubscript{ST} x}) T\textsubscript{ST} \trans \ T\textsubscript{ST} & (fields have pre- and post-specifications) \\

M & \bnfdef & M\textsubscript{SIG} e \\

e & \bnfdef & s \\
&		\bnfalt & s.f & (field access) \\
&		\bnfalt & s.\generics{m}{T}(\textoverline{x}) \\
&		\bnfalt & \letExpr{x}{T}{e}{e} \\
&		\bnfalt & \new{} \generics{C}{T}.S(\textoverline{s}) &(contract fields; state fields)\\
&		\bnfalt & s $\nearrow_{\Owned \bnfalt \Shared}$ S(\textoverline{s}) & (State transition) \\
&		\bnfalt & s.f := s & (field update, with 1-based indexing)\\
&		\bnfalt & \assertExpr{s}{T\textsubscript{ST}} & {(static assert)} \\
&		\bnfalt & \ifExpr{s}{P}{T\textsubscript{ST}}{e}{e} & (dynamic state test, owned or shared s) \\
&		\bnfalt & \disown s & (drop ownership of an owned reference) \\
&		\bnfalt & \pack \\

s & \bnfdef 	& x & (simple expressions)\\

\end{tabular}
\caption{Syntax of Silica}
\label{silica-syntax}
\end{figure}

\subsection{Silica Static Semantics}

\framebox{$\ty{\typeBounds}{\Delta}{s}{e}{T}{\Delta'}$} \textbf{Well-typed expressions}

Unlike some traditional typing judgments, in addition to an \textit{input} typing context $\Delta$, Silica's typing judgment includes an \textit{output} typing context $\Delta'$. This is because an expression can change the mode of object references. For example, using a variable that references an object may consume ownership of the object.

Note that $\overline{e} : \overline{T}$ is defined to mean a sequence $\overline{e : T}$. Expressions are typechecked in the context of an indirect reference $l$ or variable $x$, which represents $this$. Initial programs are written using $this$, but evaluation of invocations will substitute locations for instances of $this$. The subscript on the turnstile tracks the value of $this$ in the current invocation.

T-lookup relies on the split judgment (\splitType{T_1}{T_2}{T_3}), which describes how a permission in $T_1$ can be split between $T_2$ and $T_3$.
\begin{mathpar}
\inferrule*[right=T-lookup]{\splitType{T_1}{T_2}{T_3}}{\ty{\typeBounds}{\Delta, s': T_1}{s}{s'}{T_2}{\Delta, s': T_3}}
\end{mathpar}

In \keyword{let}, the bound variable can be an owning reference to an asset, but if so, $e_2$ must consume the ownership (as indicated by \textit{disposable}).
\begin{mathpar}
\inferrule*[right=T-let]
{
    \ty{\typeBounds}{\Delta}{s}{e_1}{T_1}{\Delta'} \and
    \ty{\typeBounds}{\Delta', x: T_1}{s}{e_2}{T_2}{\Delta'', x : T_1'} \and
    \disposable{\typeBounds}{T_1'}
}
{
    \ty{\typeBounds}{\Delta}{s}{\letExpr{x}{T_1}{e_1}{e_2}}{T_2}{\Delta''}
}
\end{mathpar}

The subsOk judgment, which is used in T-new, ensures that the given type parameters are suitable according to the declaration of $C$.
\begin{mathpar}
\inferrule*[right=T-new]{\ty{\typeBounds}{\Delta}{s}
{
	\overline{s'}}{\overline{T_{s'}}}{\Delta'} \and
    \overline{\subtype{\typeBounds}{T_{s'}}{stateFields(\generics{C}{T}, S)}} \and
    \overline{\subsOk{\typeBounds}{T}{T_G}}
    \\\\
    def(C) = \contract \ \generics{C}{T_G} \ \implements\ \generics{I}{T_I}\ \{ \ldots \} }
{
	\ty{\typeBounds}{\Delta}{s}
        {\new \ \generics{C}{T}.S(\overline{s'})}{\generics{C}{T}.S}{\Delta'}
}
\end{mathpar}

When accessing a field of \keyword{this} (note that the $s$ in the expression is identical to the $s$ subscript in the judgement), there are two cases. In the first case, the type of the field is consistent with the declared type of the field, in which case we make sure that the field is in scope in all possible current states of the referenced object (via \textit{intersectFields}). In the second case, the field type has been updated due to an assignment, so the field type comes from an override in the context. In both cases, any ownership that was present is consumed from the field.
\begin{mathpar}
\inferrule*[right=T-this-field-def]{
	s.f \notin Dom(\Delta) \and
	T_1 \ f \in intersectFields(T) \and
	\splitType{T_1}{T_2}{T_3}
}
{
    \ty{\typeBounds}{\Delta, s: T}{s}{s.f}{T_2}{\Delta, s: T, s.f: T_3}
}

\inferrule*[right=T-this-field-ctxt]
{
	\splitType{T_1}{T_2}{T_3}
}
{
    \ty{\typeBounds}{\Delta, s: T, s.f: T_1}{s}{s.f}{T_2}{\Delta, s: T, s.f: T_3}
}
\end{mathpar}

A field can be overwritten only if the current reference is disposable, since otherwise assignment might overwrite owning references to assets.
\begin{mathpar}
\inferrule*[right=T-fieldUpdate]
{
    \ty{\typeBounds}{\Delta}{s}{s.f}{T_C.T_{ST}}{\Delta'} \and
    \ty{\typeBounds}{\Delta'}{s}{s_f}{T_C.T_{ST}'}{\Delta''} \and
	 \\\\
    \disposable{\typeBounds}{T_C.T_{ST}}
}
{
    \ty{\typeBounds}{\Delta}{s}{s.f := s_f}{\text{\Unit}}{\Delta'', s.f: T_C.T_{ST}'}
}
\end{mathpar}

In invocations (of both public and private transactions), if the type of an argument differs from the declared type of the formal parameter, the final type of the argument may differ from the declared final type of the parameter. The function \textit{funcArg} defines the resulting final types of the arguments. The subtyping antecedents ensure that the arguments are suitable for the declared types of the formal parameters. Invocations of public transactions can only occur when field types are consistent with their declarations.
\begin{mathpar}
\inferrule*[right=T-inv]
{
    \transaction{\typeBounds}{\generics{m}{T_M}}{\generics{D}{T}} = T \ \generics{m}{T_M'}(\overline{T_{C_x}.T_x \trans T_{xST} \ x}) \ T_{this} \trans \ T_{this}' e
    \\\\
        \bound{\typeBounds}{T_C.T_{STs1}'} = \generics{D}{T}.T_{STs1} \and
    \subperm{\typeBounds}{T_{STs1}}{T_{this}}
    \\\\
    \overline{\subtype{\typeBounds}{T_{s2}}{T_{C_x}.T_x}} \and \forall f, s.f \notin \Delta
		\\\\
	T_{s1}' = \funcArg{T_C.T_{STs1}}{T_C.T_{this}}{T_C.T_{this}'} \and
    \overline{T_{s2}' = \funcArg{T_{s2}}{T_x}{T_{C_x}.T_{xST}}}
}
{
    \ty{\typeBounds}{\Delta, s_1 : T_C.T_{STs1}', \overline{s_2: T_{s2}}}{s}{s_1.\generics{m}{T_M}(\overline{s_2})}{T}{\Delta, s_1: T_{s1}', \overline{s_2: T_{s2}'}}
}
\end{mathpar}

Private invocations differ from public invocations because the current types of the fields must be checked against the transaction's preconditions and the field types must be updated after invocation.
\begin{mathpar}
\inferrule*[right=T-privInv]
{
    \transaction{\typeBounds}{\generics{m}{T_M}}{\generics{D}{T}} = \overline{T_{C_f}.T_{fdecl} \trans T_{fST} \ x} \ T \ m(\overline{T_{C_x}.T_x \trans T_{xST} \ x}) \  T_{this} \trans \ T_{this}' \ e
	\\\\
		\bound{\typeBounds}{T_C.T_{STs1}'} = \generics{D}{T}.T_{STs1} \and
    \subperm{\typeBounds}{T_{STs1}}{T_{this}}
    \\\\
    \overline{\subtype{\typeBounds}{T_{s2}}{T_{C_x}.T_x}} \and
    \overline{\subtype{\typeBounds}{T_f}{T_{C_f}.T_{fdecl}}}
	\\\\
	T_{s1}' = \funcArg{C.T_{STs1}}{C.T_{this}}{C.T_{this}'} \and
  \overline{T_{s2}' = \funcArg{T_{s2}}{T_x}{T_{C_x}.T_{xST}}}
	\\\\
	\overline{T_f' = \funcArg{T_{f}}{T_{C_f}.T_{fdecl}}{T_{C_f}.T_{fST}}}
}
{
	\ty{\typeBounds}{\Delta, s_1 : T_C.T_{STs1}, \overline{s_2: T_{s2}}, \overline{s.f: T_f}}{s_1}{s_1.\generics{m}{T_M}(\overline{s_2})}{T}{\Delta, s_1: T_{s1}', \overline{s_2: T_{s2}'}, \overline{s.f: T_{f}'}}
}
\end{mathpar}

$T-\nearrow_p$ allows changing the nominal state of \keyword{this}. Unlike transitions in FT, $T-\nearrow_p$ does not permit arbitrary changes of class; it restricts the change to states within the object's current contract. It requires giving away ownership of all possible fields of \keyword{this} first. Since the current state may not be uniquely known statically, \textit{unionFields} is used to identify all possible current fields.

Obsidian permits assignment to fields in target states \textit{before} the transition has occurred. This is not directly supported in Silica, but can be represented indirectly.

\begin{mathpar}
\inferrule*[right=T-$\nearrow_p$]{
    \subperm{\typeBounds}{T_{ST}}{p} \and p \in \{\Shared, \Owned\} \\\\ 
    \ty{\typeBounds}{\Delta}{s}{\overline{x}}{\overline{T}}{\Delta'} \and
    \overline{\subtype{\typeBounds}{T}{type(stateFields(\generics{C}{T_A}, S'))}}
		\\\\
    unionFields(\generics{C}{T_A}, T_{ST}) = \overline{T_{fs} \ f_s}
    \and fieldTypes_s(\Delta; \overline{T_{fs} \ f_s}) = \overline{T_{fs}'}
	\and
    \overline{\disposable{\typeBounds}{T_{fs}'}}
	}
    {\ty{\typeBounds}{\Delta, s: \generics{C}{T_A}.T_{ST}}{s}{s \nearrow_{p} S'(\overline{x})}{\Unit}{\Delta', s : \generics{C}{T_A}.S'}}
\end{mathpar}


\begin{mathpar}
\inferrule*[right=T-assertStates]
    {\overline{S} \subseteq \overline{S'}}
    {\ty{\typeBounds}{\Delta, x : T_C.\overline{S}}{s}{\assertExpr{x}{\overline{S'}}}{\Unit}{\Delta, x : T_C.\overline{S}}}

\inferrule*[right=T-assertPermission]
{
    T_{ST} \in \{ \Owned, \Unowned, \Shared \}
}
{
    \ty{\typeBounds}{\Delta, x : T_C.T_{ST}}{s}{\assertExpr{x}{T_{ST}}}{\Unit}{\Delta, x : T_C.T_{ST}}
}
\end{mathpar}

When asserting that a variable is in a state corresponding to a type variable, $\text{bound}_{*}$ is used to compute the most specific concrete mode for the variable.
\begin{mathpar}
\inferrule*[right=T-assertInVar]
    {\ensuremath{\nonVar{T_{ST}}} \and  \boundPerm{\typeBounds}{p} = T_{ST}
     \\\\
     \ty{\typeBounds}{\Delta, x : T_C.T_{ST}}{s}{\assertExpr{x}{T_{ST}}}{\Unit}{\Delta, x : T_C.T_{ST}}}
    {\ty{\typeBounds}{\Delta, x : T_C.T_{ST}}{s}{\assertExpr{x}{p}}{\Unit}{\Delta, x : T_C.T_{ST}}}

\inferrule*[right=T-assertInVarAlready]
    { }
    {\ty{\typeBounds}{\Delta, x : T_C.p}{s}{\assertExpr{x}{p}}{\Unit}{\Delta, x : T_C.p}}
\end{mathpar}

Dynamic state tests are typechecked according to the ownership of the variable to be checked. T-isInStaticOwnership can be used when a variable is an owning reference but does not provide a particular state specification that the programmer wants. In contrast, IsIn-Dynamic applies when there is no ownership.

\begin{mathpar}
\inferrule*[right=T-IsInStaticOwnership]
{
    \ty{\typeBounds}{\Delta, x: T_C.\overline{S}}{s}{e_1}{T_1}{\Delta'} \and
    \overline{S} \subseteq states(T_C)
    \\\\
    \subperm{\typeBounds}{T_{ST}}{Owned} \and
    \overline{S_x} = \possibleStates{\typeBounds}{T_C.T_{ST}}
	\\\\
    \ty{\typeBounds}{\Delta, x: T_C.(\overline{S_x} \setminus \overline{S})}{s}{e_2}{T_1}{\Delta''} \and
    \Delta_f = merge(\Delta', \Delta'')
}
{
    \ty{\typeBounds}{\Delta, x: T_C.T_{ST}}{s}{\ifExpr{x}{owned}{\overline{S}}{e_1}{e_2}}{T_1}{\Delta_f}
}
\end{mathpar}

In \ifExpr{x}{shared}{\overline{S}}{e_1}{e_2}, $e_1$ is permitted to change the state of the object referenced by $x$, but it is not permitted to allow another reference to obtain permanent ownership of the object. While $e_1$ is evaluating, all state changes to the object referenced by $x$ that occur via Shared aliases will cause program termination, so it is up to the programmer to ensure that this is impossible.

\begin{mathpar}
\inferrule*[right=T-IsIn-Dynamic]
{   
    \overline{S} \subseteq states(T_C) \and
    \ty{\typeBounds}{\Delta, x: T_C.\overline{S}}{s}{e_1}{T_1}{\Delta', x: T_C.T_{ST}'} 
    \\\\
        \boundPerm{\typeBounds}{T_{ST}'} \neq Unowned
    \\\\
        \ty{\typeBounds}{\Delta, x: T_C.Shared}{s}{e_2}{T_1}{\Delta'', x: T_C.Shared}
	 \\\\
	    \Delta_f = merge(\Delta', \Delta''), x: T_C.Shared
}
{
    \ty{\typeBounds}{\Delta, x: T_C.Shared}{s}{\ifExpr{x}{shared}{\overline{S}}{e_1}{e_2}}{T_1}{\Delta_f}
}

\end{mathpar}

If the test is against a permission variable, we check $e_1$ in a context that gives $x$ the permission variable's permission, which will result in relying on the bound on $p$ in $\Gamma$. 
\begin{mathpar}
\inferrule*[right=T-IsIn-PermVar]
{
    \ty{\typeBounds}{\Delta, x: T_C.p}{s}{e_1}{T_1}{\Delta'} \and
    \ty{\typeBounds}{\Delta, x: T_C.T_{ST}}{s}{e_2}{T_1}{\Delta''}
    \\\\
    \Delta_f = merge(\Delta', \Delta'') \and
    \text{Perm} = \ToPermission{T_{ST}}
}
{
    \ty{\typeBounds}{\Delta, x: T_C.T_{ST}}{s}{\ifExpr{x}{\text{Perm}}{p}{e_1}{e_2}}{T_1}{\Delta_f}
}
\end{mathpar}

In T-IsIn-Perm-Then and T-IsIn-Perm-Else, the compiler knows which branch will be taken: either $T_{ST}$ satisfies the given condition or it does not. If $T_{ST}$ is a variable, then we treat it as if it were owned (via ToPermission).
\begin{mathpar}
\inferrule*[right=T-IsIn-Perm-Then]
{
    \text{Perm} \in \{ \Owned, \Unowned, \Shared \} \and
    P = \ToPermission{T_{ST}}
    \\\\
    \subperm{\typeBounds}{P}{\text{Perm}} \and
    \ty{\typeBounds}{\Delta, x: T_C.T_{ST}}{s}{e_1}{T_1}{\Delta'}
}
{
    \ty{\typeBounds}{\Delta, x: T_C.T_{ST}}{s}{\ifExpr{x}{P}{\text{Perm}}{e_1}{e_2}}{T_1}{\Delta'}
}

\inferrule*[right=T-IsIn-Perm-Else]
{
    \text{Perm} \in \{ \Owned, \Unowned, \Shared \} \and
    P = \ToPermission{T_{ST}}
    \\\\
    \notsubperm{\typeBounds}{P}{\text{Perm}} \and
    \ty{\typeBounds}{\Delta, x: T_C.T_{ST}}{s}{e_2}{T_1}{\Delta'}
}
{
    \ty{\typeBounds}{\Delta, x: T_C.T_{ST}}{s}{\ifExpr{x}{P}{\text{Perm}}{e_1}{e_2}}{T_1}{\Delta'}
}
\end{mathpar}

The case where we test to see if an unowned reference is in a particular state is included because it can arise via substitution.
\begin{mathpar}
\inferrule*[right=T-IsIn-Unowned]
{
    \ty{\typeBounds}{\Delta, x: T_C.\Unowned}{s}{e_2}{T_1}{\Delta'}
}
{
    \ty{\typeBounds}{\Delta, x: T_C.\Unowned}{s}{\ifExpr{x}{\Unowned}{\overline{S}}{e_1}{e_2}}{T_1}{\Delta'}
}
\end{mathpar}

\vspace{1em}

Disown discards ownership of its parameter. Existing ownership is split; in $\splitType{T_C.T_{ST}}{T}{T'}$, $T$ retains ownership and $T'$ lacks it, so the output context uses $T'$ as the new type of $s'$. Note that the split is not a function; one can see by inspection of the definition of split that $T'$ is not owned, but may be either shared or unowned.
\begin{mathpar}
\inferrule*[right=T-disown]
    {\splitType{T_C.T_{ST}}{T}{T'} \and
     \subperm{\typeBounds}{T_{ST}}{\Owned} }
    {\ty{\typeBounds}{\Delta, s': T_C.T_{ST}}{s}{\disown \ s'}{\Unit}{\Delta, s': T'}}
\end{mathpar}

\code{Pack} updates $\Delta$, removing all type overrides of fields of \code{this}. It is only appropriate, of course, when the existing overrides are consistent with the field declarations. There is no corresponding \textit{unpack}; instead, field assignment and field reading can cause a future need to invoke \code{pack}. Note that \code{pack} exists only in the formal model and is not needed in user programs because the compiler can insert them where required (at the ends of transactions and before invocations of public transactions).

\begin{mathpar}
\inferrule*[right=T-pack]
{
	s.f \notin dom(\Delta) \and
    contractFields(T) = \overline{T_{decl} \ f} \and
    \overline{\subtype{\typeBounds}{T_f}{T_{decl}}} \and
    \typeBounds \vdash \overline{\sameOwnership{T_f}{T_{decl}}}
}
{
    \ty{\typeBounds}{\Delta, s: T, \overline{s.f: T_f}}{s}{\pack}{\Unit}{\Delta, s: T}
}
\end{mathpar}

\begin{absolutelynopagebreak}
\framebox{$\okIn{M}{C}$} \textbf{Well-typed transaction}

\begin{mathpar}
\inferrule*[right=PublicTransactionOK]{
	params(C) = \overline{T_G} \and
    \overline{\text{Var}(T_G) = T} \and
    \typeBounds = \overline{T_G}, \overline{T_M}
    \\\\
	\ty
        {\typeBounds}
        {this: \generics{C}{T}.T_{this},
			{\
				\overline{x: C_x.T_x}}}
		{this}
		{e}
		{T}
		{this: C.T_{this}', \overline{x: C_x.T_{x}'}}
	}
    {\okIn{T \ \generics{m}{T_M}(\overline{{C_x}.T_x \trans T_{x}' \ x}) \ T_{this} \trans \ T_{this}' \ e
}{C} }
\end{mathpar}

Note that all fields of \this{} must end the transaction with types consistent with their declarations; otherwise, there would be occurrences of s.f in $\Delta'$.

\begin{mathpar}
\inferrule*[right=PrivateTransactionOK]
    {params(C) = \overline{T_G} \and
    \overline{\text{Var}(T_G) = T} \and
     contractFields(\generics{C}{T}) = \overline{T_f \ f}
     \\\\
     \Delta = s: \generics{C}{T}.T_{ST},
        {\
            \overline{s.f: contract(T_f).S_{f1}},
            \overline{x: C_x.T_x}} 
     \\\\
     \Delta' = s: \generics{C}{T}.T_{ST}', \overline{s.f: contract(T_f).S_{f2}},\overline{x: C_x.T_{x}'} 
          \\\\
	\ty
        {\typeBounds}{\Delta}
		{s}
		{e}
		{T}
		{\Delta'} \and
     \typeBounds = \overline{T_G}, \overline{T_M}
	}
    {\okIn{\overline{S_{f1}>>S_{f2} f}\ T \ \generics{m}{T_M}(C_x.T_x \trans T_{x}' x) T_{ST} \trans \ T_{ST}' \ e
}{C} }
\end{mathpar}
\end{absolutelynopagebreak}

The difference between public and private transactions is that private transactions may begin and end with fields inconsistent with their declarations. In both cases, inside $e$, it is possible to set fields of \this so that they do not match their declared types. However, while this is the case, additional public transactions cannot be invoked, ensuring that only private transactions are exposed to the inconsistent state.

There may be aliases to \this{}. However, if the fields of \this{} are inconsistent with their types, no public transactions can be invoked, so the inconsistency cannot be visible outside this transaction and any private transactions that it invokes. Furthermore, the state of \this{} can only be changed if the permission on \this{} allows that operation (see \textsc{This-state-transition}).

In Obsidian, field pre- and post- types are optional; when they are omitted, specifications match the field type declarations.

\framebox{ST \textbf{ok}} \textbf{Well-formed State}

All fields must have distinct names, and if any field is an asset, then the state must be labeled $\asset$.
\begin{mathpar}
\inferrule*
{
	\forall {i, j} \ \  i \neq j \Rightarrow f_i \neq f_j \and
    \overline{\nonAsset{\typeBounds}{T}}
}
{
	\typeBounds \vdash S \ \overline{T \ f} \ \textbf{ok}
}

\inferrule*
{
	\forall {i, j} \ \  i \neq j \Rightarrow f_i \neq f_j
}
{
	\typeBounds \vdash \asset \ S \ \overline{T \ f} \ \textbf{ok}
}
\end{mathpar}

\framebox{CL \textbf{ok}} \textbf{Well-typed Contract}

\begin{mathpar}
\inferrule*
{
	\okIn{\overline{M}}{C}  \and
    \overline{T_G} \vdash \ \overline{ST \ \textbf{ok}} \and
	|\overline{ST}| > 0 \and
    \\\\
    \transactionNames{I} \subseteq \transactionNames{C} \and
    \stateNames{I} \subseteq \stateNames{C}
    \\\\
    \forall T \in \overline{T}, \isVar{T} \implies T \in \overline{Var(T_G)}
    \\\\
    \forall M \in \overline{M}, \transactionName{M} \in \transactionNames{I} \implies \implementOk{\overline{T_G}}{\generics{I}{T}}{M}
    \\\\
    \forall S \in \overline{ST}, \stateName{S} \in \stateNames{I} \implies \implementOk{\overline{T_G}}{\generics{I}{T}}{S}
    \\\\
    \overline{\genericsOk{\overline{T_G}}{T_G}} \and
    \overline{\subsOk{\overline{T_G}}{T}{\params{I}}}
}
{
	\contract \ \generics{C}{T_G} \ \implements\ \generics{I}{T} \ \{ \overline{ST} \ \overline{F} \ \overline{M} \} \ \textbf{ok}
}

\end{mathpar}

\framebox{IFACE \textbf{ok}} \textbf{Well-typed Interface}

\begin{mathpar}
\inferrule*
{
    \overline{\genericsOk{\overline{T_G}}{T_G}}
}
{
    \interface \ \generics{I}{T_G}  \{ \overline{ST} \ \overline{M_{SIG}} \} \ \textbf{ok}
}
\end{mathpar}

\framebox{PG \textbf{ok}} \textbf{Well-typed Program}

\begin{mathpar}
\inferrule*
    {\ok{\overline{CON}} \and
     \ok{\overline{IFACE}} \and
     \ty{\cdot}{\cdot}{s}{e}{T}{\cdot}
    }
    {\langle \overline{IFACE}, \overline{CON}, e \rangle \ \textbf{ok}}
\end{mathpar}

\
\subsection{Auxiliary Judgements}

\subsubsection{Program structure} \mbox{}

We assume that the contracts and interfaces defined in a program are ambiently available via the $def$ function, which retrieves the definition of a contract or interface (definition) by name.
Likewise, the definition of a state $S$ of contract or interface $D$ can be retrieved via $sdef(D, S)$, and the definition of a transaction can be retrieved via $tdef(D, m)$.
Note that for declaration variables $def(X)$ is the interface bound on $X$; similarly, $sdef(X, S)$ is the state in the bound on $X$.
That is $sdef(X,S) = sdef(def(X), S)$.

\fbox{$stateFields(D, S)$}

On individual states, $stateFields$ gives only the fields defined directly in those states:
\begin{mathpar}
\inferrule*{
	def(C) = \ \contract \ \generics{C}{T_G} \ \implements \ \generics{I}{T} \ \{ \overline{ST} \ \overline{M} \}
    \\\\
	S \ \overline{F} \in \overline{ST}
}
{stateFields(C, S) = \overline{F}}

\inferrule*{
}
{stateFields(I, S) = \cdot}
\end{mathpar}

\fbox{$unionFields(T)$}

The $unionFields$ function looks up the fields that are defined in ANY of the states in a set of states. Note that the syntax guarantees that any field has consistent types in all states in which it is defined. This is useful when it is known that one of two different types captures the state of an object, but it is not known which one.
\begin{mathpar}
\inferrule*
{
	F = \cup_{S \in \overline{S}} {stateFields(D, S)}
}
{
	unionFields(D.\overline{S}) = F
}

\inferrule*
{
	T_{ST} \in \{Shared, Owned, Unowned\} \\\\
	cdef(C) = \ \contract \ C  \{ \overline{[\asset] S \ F_S} \ \overline{M} \}
	\\\\
	F = \cup_{S \in \overline{F_S}} {stateFields(C, S)}
}
{
	unionFields(D.T_{ST}) = F
}\end{mathpar}

\fbox{$intersectFields(T)$}

The $intersectFields$ function looks up the fields that are defined in ALL of the states in a set of states. Note that the syntax guarantees that any field has consistent types in all states in which it is defined.
\begin{mathpar}
\inferrule*
{
	F = \cap_{S \in \overline{S}} {stateFields(D, S)}
}
{
	intersectFields(D.\overline{S}) = F
}

\inferrule*
{
	T_{ST} \in \{Shared, Owned, Unowned\}
		\\\\
	cdef(C) = \contract \ C  \{ \overline{[\asset] S \ F_S} \ \overline{M} \}
	\\\\
    F = \cap_{S \in \overline{F_S}} {stateFields(D, S)}
}
{
	intersectFields(C.T_{ST}) = F
}
\end{mathpar}

\framebox{\textbf{contract($T_C$)}}

The \textit{contract} function relates types with their contracts.
\begin{mathpar}
\inferrule*{ }{
	contract(T_C.T_{ST}) = T_C\\
}

\end{mathpar}

\fbox{$contractFields(C)$}

On contracts, contractFields gives the set of field declarations defined in all of a contract's states.
\begin{mathpar}
contractFields(C) \triangleq intersectFields(C.Unowned)
\end{mathpar}

\framebox{$fieldTypes_s(\Delta; \overline{T_{fs} \ f_s})$}

$fieldTypes$ gives the current types of the fields, given that some of them may be overridden in the current context.
\begin{mathpar}
\inferrule*{ }
	{fieldTypes_s(\cdot; \overline{T_{fs} \ f_s}) = \overline{T_{fs}}}

\inferrule*{ f \in \overline{f_s} \and fieldTypes_s(\Delta; \overline{T_{fs} \ f_s}) = \overline{T'}}
	{fieldTypes_s(\Delta, s.f: T; \overline{T_{fs} \ f_s}) = T, \overline{T'}}

\inferrule*{ }
	{fieldTypes_s(\Delta, b: T; \overline{T_{fs} \ f_s}) = fieldTypes_s(\Delta; \overline{T_{fs} \ f_s})}

\end{mathpar}

\subsubsection{Reasoning about types} \mbox{}

\framebox{$\ToPermission{T_{ST}}$}

ToPermission provides a conservative approximation of ownership to ensure that if ToPermission indicates non-ownership, the type is definitely disposable.
\begin{align*}
	\ToPermission{\overline{S}} &\triangleq \Owned & \ToPermission{\Unowned} &\triangleq \Unowned\\
	\ToPermission{p} &\triangleq \Owned & \ToPermission{\Shared} &\triangleq \Shared\\
	\ToPermission{\Owned} &\triangleq \Owned\\
\end{align*}

\framebox{$\possibleStates{\typeBounds}{T_C.T_{ST}} = T_{ST}$}
\begin{mathpar}
\inferrule*{ }
    {\possibleStates{\typeBounds}{T_C.\overline{S}} = \overline{S}}

\inferrule*{ P \in \{ \Owned, \Shared, \Unowned \} }
    {\possibleStates{\typeBounds}{T_C.P} = \stateNames{def(T_C)}}

\inferrule*{ \genericParamOpt{X}{p}{\generics{I}{T}.T_{ST}} \in \typeBounds }
    { \possibleStates{\typeBounds}{T_C.p} = \possibleStates{\typeBounds}{T_C.T_{ST}} }
\end{mathpar}

\framebox{$\mathbf{\isAsset{\typeBounds}{T}}$}
\begin{mathpar}
\inferrule*{
    \asset \ S \ \overline{F} \in \possibleStates{\typeBounds}{\generics{D}{T}.T_{ST}} }
    {\isAsset{\typeBounds}{\generics{D}{T}.T_{ST}}}

\inferrule*{
    \genericParam{\asset}{X}{p}{\generics{I}{T}.T_{ST_i}} \in \typeBounds }
    {\isAsset{\typeBounds}{X.T_{ST}}}
\end{mathpar}

\framebox{$\mathbf{\nonAssetState{\typeBounds}{ST}}$}
\begin{mathpar}
\inferrule*{ }
    {\nonAssetState{\typeBounds}{S \ \overline{F}} }
\end{mathpar}

\framebox{$\mathbf{\nonAsset{\typeBounds}{T}}$}
\begin{mathpar}
\inferrule*{
    \overline{\nonAssetState{\typeBounds}{\possibleStates{\typeBounds}{\generics{D}{T}.T_{ST}}}} }
    {\nonAsset{\typeBounds}{\generics{D}{T}.T_{ST}}}

\inferrule*{
    \genericParam{}{X}{p}{\generics{I}{T}.T_{ST_i}} \in \typeBounds }
    {\nonAsset{\typeBounds}{X.T_{ST}}}
\end{mathpar}

\framebox{$\mathbf{\disposable{\typeBounds}{T}}$}

The $disposable$ judgement describes reference types that are NOT owning references to assets. When applied to a set of states, all states must be disposable in order for the set to be disposable.

\begin{mathpar}
\inferrule*{\notOwned{T}}
    { \disposable{\typeBounds}{T_C.T_{ST}} }

\inferrule*{\maybeOwned{T_C.T_{ST}} \and
    \nonAsset{\typeBounds}{T_C.T_{ST}} }
    { \disposable{\typeBounds}{T_C.T_{ST}} }
\end{mathpar}

\framebox{$\mathbf{\notOwned{T}}$}
\begin{mathpar}
    \inferrule*{ }
    {\notOwned{T_C.Unowned}}

    \inferrule*{ }
    {\notOwned{T_C.Shared}}
    
    \inferrule*{ }
    {\notOwned{\Unit}}
    
\end{mathpar}

\framebox{$\mathbf{\maybeOwned{T}}$}
\begin{mathpar}
    \inferrule*{ T_{ST} <:_* \Owned }
    {\maybeOwned{T_C.T_{ST}}}

\inferrule*{ }
    {\maybeOwned{T_C.p}}
\end{mathpar}

Note that all permission variables could be owned, because we only have upper bounds on permissions.
Therefore, we must treat all permission variables as though they may be owned.

\framebox{$\bound{\typeBounds}{T}$}

\begin{mathpar}
\inferrule*{ }
    { \bound{\typeBounds}{\Unit} = \Unit }

\inferrule*{ \boundPerm{\typeBounds}{T_{ST}} = T_{ST}'}{
    \bound{\typeBounds}{\generics{D}{T}.T_{ST}} = \generics{D}{T}.T_{ST}'}

\inferrule*{ \genericParamOpt{X}{p}{\generics{I}{T}.T_{ST}'} \in \typeBounds \and
    \boundPerm{\typeBounds}{T_{ST}} = T_{ST}'}{
    \bound{\typeBounds}{X.T_{ST}} = \generics{I}{T}.T_{ST}'}
\end{mathpar}

\framebox{$\boundPerm{\typeBounds}{T_{ST}}$}

\begin{mathpar}
\inferrule*{ P \in \{ Owned, Shared, Unowned \} }{
    \boundPerm{\typeBounds}{P} = P
}

\inferrule*{ }{
    \boundPerm{\typeBounds}{\overline{S}} = \overline{S}
}

\inferrule*{ \genericParamOpt{X}{p}{\generics{I}{T}.T_{ST}} \in \typeBounds }{
    \boundPerm{\typeBounds}{p} = T_{ST}
}
\end{mathpar}

The bound of a type $T$ or permission or state $T_{ST}$ is the most specific concrete (i.e., non-variable) type (resp. permission or state) that is a supertype of $T$.
For example, if we know from a type parameter that the type variable $X$ must implement an interface $\generics{I}{T}$ and $p$ must be a subpermission of $\Owned$, then the bound of $X.p$ is $\generics{I}{T}.\Owned$.
However, a concrete type such as $\generics{C}{T}.S$ is already as specific as possible---therefore, its bound is itself.

\framebox{$\nonVar{T}$, $\nonVar{T_C}$, $\nonVar{T_{ST}}$}
\begin{mathpar}
\inferrule*{ }{
    \nonVar{\generics{D}{T}.T_{ST}}
}

\inferrule*{ }{
    \nonVar{\Unit}
}

\inferrule*{ }{
    \nonVar{\generics{D}{T}}
}

\inferrule*{ T_{ST} \in \{ Owned, Shared, Unowned \} }{
    \nonVar{T_{ST}}
}

\inferrule*{ }
    { \nonVar{\overline{S}} }
\end{mathpar}

\framebox{$\isVar{T}$, $\isVar{T_C}$, $\isVar{T_{ST}}$}
\begin{mathpar}
\inferrule*{ }{
    \isVar{X.T_{ST}}
}

\inferrule*{ }{
    \isVar{X}
}

\inferrule*{ }{
    \isVar{p}
}
\end{mathpar}

\framebox{$Var(T_G)$, $PermVar(T_G)$, $Perm(T)$}
\begin{align*}
	Var(\genericParamOpt{X}{p}{\generics{I}{T}.T_{ST}}) &\triangleq X\\
    PermVar(\genericParamOpt{X}{p}{\generics{I}{T}.T_{ST}}) &\triangleq p\\
\end{align*}
\begin{align*}
    Perm(T_C.T_{ST}) &\triangleq T_{ST}\\
    Perm(\Unit) &\triangleq \Unowned\\
\end{align*}

\framebox{$\transactionName{M}$, $\transactionName{M_{SIG}}$, $\transactionNames{\overline{M}}$}
\begin{align*}
    \transactionName{T \generics{m}{T_M}(\overline{T \trans T_{ST} x}) T_{ST} \trans \ T_{ST}} &\triangleq m\\
    \transactionName{T \generics{m}{T_M}(\overline{T \trans T_{ST} x}) T_{ST} \trans \ T_{ST} \ e} &\triangleq m\\
    \transactionName{\overline{T_{ST}>>T_{ST} f} T \generics{m}{T_M}(\overline{T \trans T_{ST} x}) T_{ST} \trans \ T_{ST}} &\triangleq m\\
    \transactionName{\overline{T_{ST}>>T_{ST} f} T \generics{m}{T_M}(\overline{T \trans T_{ST} x}) T_{ST} \trans \ T_{ST} \ e} &\triangleq m\\
    \end{align*}
\[
    \transactionNames{\overline{M}} \triangleq \overline{\transactionName{M}}
\]

\framebox{$\states{D}$}
\begin{mathpar}
\inferrule*{ \contract \ \generics{C}{T_G} \ \implements\ \generics{I}{T} \{ \overline{ST}\ \overline{M} \} }
    { \states{C} = \overline{ST} }

\inferrule*{ \interface \ \generics{I}{T_G} \{ \overline{ST}\ \overline{M_{SIG}} \} }
    { \states{I} = \overline{ST} }
\end{mathpar}

\framebox{$\stateNames{D}$, $\stateName{S}$}
\begin{mathpar}
\inferrule*{ }
    { \stateName{[\asset] \ S \ \overline{F}} = S }

\inferrule*{ \states{D} = \overline{S} }
    { \stateNames{D} = \overline{\stateName{S}} }
\end{mathpar}

\framebox{$\params{D}$, $\params{M}$}
\begin{mathpar}
\inferrule*{ def(C) = \contract \ \generics{C}{T_G} \ \implements\ \generics{I}{T} \{ \overline{ST}\ \overline{M} \} }
    { \params{C} = \overline{T_G} }

\inferrule*{ def(I) = \interface \ \generics{I}{T_G} \{ \overline{ST}\ \overline{M_{SIG}} \} }
    { \params{I} = \overline{T_G} }

\inferrule*{ }
    { \params{T \ \generics{m}{T_M}(\overline{T \trans T_{ST} \ x}) \ T_{ST_i} \trans \ T_{ST_f} \ e} = \overline{T_M} }

\inferrule*{ }
    { \params{\overline{T_{STs1}>>T_{STs2} f} \ T \ \generics{m}{T_M}(\overline{T \trans T_{ST} \ x}) \ T_{ST_i} \trans \ T_{ST_f} \ e} = \overline{T_M} }
\end{mathpar}

\framebox{$\implementOk{\typeBounds}{\generics{I}{T}}{M_{SIG}}$, $\implementOk{\typeBounds}{\generics{I}{T}}{ST}$ }
\begin{mathpar}
\inferrule*{
    \transaction{\typeBounds}{m}{\generics{I}{T}} = T_{ret}' \ \generics{m}{T_M'}(\overline{T' \trans T_{ST}' \ x}) \ T_{ST_i}' \trans \ T_{ST_f}'
    \\\\
    \overline{\subtype{\typeBounds}{T'}{T}} \and
    \overline{\subperm{\typeBounds}{T_{ST}}{T_{ST}'}} \and
    \subperm{\typeBounds}{T_{ST_i}'}{T_{ST_i}}
    \\\\
    \subperm{\typeBounds}{T_{ST_f}}{T_{ST_f}'} \and
    \subtype{\typeBounds}{T_{ret}}{T_{ret}'}
    }
    { \implementOk{\typeBounds}{\generics{I}{T}}{T_{ret} \ \generics{m}{T_M}(\overline{T \trans T_{ST} \ x}) \ T_{ST_i} \trans \ T_{ST_f}} }

\inferrule*{
    sdef(S, \generics{I}{T}) = \asset \ S
    }
    { \implementOk{\typeBounds}{\generics{I}{T}}{[\asset]\ S \ \overline{F} } }

\inferrule*{
    sdef(S, \generics{I}{T}) = S
    }
    { \implementOk{\typeBounds}{\generics{I}{T}}{S \ \overline{F} } }
\end{mathpar}

To check $\implementOk{\typeBounds}{\generics{I}{T}}{S}$, we only need to ensure that if our state is an asset, then the state we are implementing is also an asset.

\framebox{$\subsOk{\typeBounds}{T}{T_G}$}
\begin{mathpar}
\inferrule*{
    \subtype{\typeBounds}{\generics{D}{T_1}.T_{ST}}{\generics{I}{T_2}.T_{ST}'}}
    { \subsOk{\typeBounds}{\generics{D}{T_1}.T_{ST}}{\genericParam{\asset}{X}{p}{\generics{I}{T_2}.T_{ST}'}} }

\inferrule*{
    \subtype{\typeBounds}{\generics{D}{T_1}.T_{ST}}{\generics{I}{T_2}.T_{ST}'} \and
    \nonAsset{\typeBounds}{\generics{D}{T_1}.Owned} }
    { \subsOk{\typeBounds}{\generics{D}{T_1}.T_{ST}}{\genericParam{}{X}{p}{\generics{I}{T_2}.T_{ST}'}} }
\end{mathpar}

We can substitute a non-asset for an asset generic parameter, but not vice versa.
Note that, as we can use type variables without their corresponding permission variable (e.g., we can write $X.Owned$, not just $X.p$), we must check whether the generic parameter is an asset in \emph{any} state, not just its bound. Similarly, we must check if the type we pass is an asset in \emph{any} state, not just the one we pass.

\framebox{$\genericsOk{\typeBounds}{T_G}$}

$\genericsOk{\typeBounds}{T_G}$ expresses whether a use of a type parameter is suitable when the parameter must implement a particular interface.

\begin{mathpar}
\inferrule*{ \forall T \in \overline{T}, \isVar{T} \implies T \in Var(\typeBounds) \and
    \overline{\subsOk{\typeBounds}{T}{\params{I}}} \and
    \nonAsset{\typeBounds}{\generics{I}{T}.Owned}
    \\\\
    \forall T_G \in \typeBounds, \left( Var(T_G) = X \ \text{or} \ PermVar(T_G) = p \right) \implies T_G = \genericParam{}{X}{p}{\generics{I}{T}.T_{ST}}
    \\\\
    T_{ST} = \overline{S} \implies \forall S \in \overline{S}, S \in \stateNames{I} }
    { \genericsOk{\typeBounds}{\genericParam{}{X}{p}{\generics{I}{T}.T_{ST}}} }

\inferrule*{ \forall T \in \overline{T}, \isVar{T} \implies T \in Var(\typeBounds) \and
    \overline{\subsOk{\typeBounds}{T}{\params{I}}}
    \\\\
    \forall T_G \in \typeBounds, \left( Var(T_G) = X \ \text{or} \ PermVar(T_G) = p \right) \implies T_G = \genericParam{}{X}{p}{\generics{I}{T}.T_{ST}}
    \\\\
    T_{ST} = \overline{S} \implies \forall S \in \overline{S}, S \in \stateNames{I} }
    { \genericsOk{\typeBounds}{\genericParam{\asset}{X}{p}{\generics{I}{T}.T_{ST}}} }
\end{mathpar}

\framebox{$\substitute{T/T_G}{e}$}
\begin{mathpar}
\inferrule*{T_G = \genericParamOpt{X}{p}{\generics{I}{T_2}.T_{ST}'} }
    { \substitute{\generics{D}{T}.T_{ST}/T_G}{e} = [\generics{D}{T}/X][T_{ST} / p]e }

\inferrule*{ \overline{T} = T_1, T_2, \ldots, T_n \and
    \overline{T_G} = T_{G_1}, T_{G_2}, \ldots, T_{G_n} }
    { \substitute{\overline{T/T_G}}{e} = \left( \substitution{T_n/T_{G_n}} \circ \substitution{T_{n-1}/T_{G_{n-1}}} \circ \cdots \circ \substitution{T_1/T_{G_1}} \right) (e)}
\end{mathpar}

\framebox{$\transaction{\typeBounds}{\generics{m}{T_M}}{\generics{D}{T})}$}
\begin{mathpar}
\inferrule*{
    tdef(D, m) = M \and
    \overline{T_M} = \params{M} \and
    \overline{T_G} = \params{D}
    \\\\
    \overline{\subsOk{\typeBounds}{T}{T_G}} \and
    \overline{\subsOk{\typeBounds}{T_2}{T_M}} }
    { \transaction{\typeBounds}{\generics{m}{T_2}}{\generics{D}{T}} = \substitute{\overline{T_2/T_M}}{\substitute{\overline{T/T_G}}{M}} }
\end{mathpar}

\framebox{\textbf{merge($\Delta, \Delta') = \Delta''$}}
The \textit{merge} function computes a new context from contexts that resulted from branching. It ensures that ownership is consistent across both branches and takes the union of state sets for each variable.

For brevity, let $d \bnfdef x \bnfalt x.f$.

\begin{mathpar}
\inferrule*[right=Sym]{merge(\Delta; \Delta') = \Delta''} {merge(\Delta'; \Delta) = \Delta''}

\inferrule*[right=$\oplus$]{ merge(\Delta; \Delta') = \Delta'' }
	{merge(\Delta, d: T; \Delta', d: T') = \Delta'', d: (T \oplus T')}

\inferrule*[right=Dispose-disposable]{x \notin Dom(\Delta') \and merge(\Delta, \Delta') = \Delta'' \and \disposable{\typeBounds}{T}}
	{merge(\Delta, x: T; \Delta') = \Delta''}
	
\end{mathpar}

\begin{align*}
&T \oplus T \triangleq T\\
&T_C.Owned \oplus T_C.\overline{S} \triangleq T_C.Owned\\
&T_C.Shared \oplus T_C.Unowned \triangleq T_C.Unowned\\
&T_C.\overline{S} \oplus T_C.\overline{S'} \triangleq T_C.(S \cup S')\\
&\generics{C}{T}.T_{ST} \oplus \generics{I}{T}.T_{ST}' \triangleq \generics{I}{T}.(T_{ST} \oplus T_{ST}') \text{ if } def(C) = \contract \ \generics{C}{T_G} \ \implements\ \generics{I}{T} \{ \ldots \} \\
&\generics{D}{T}.T_{ST} \oplus \generics{D}{T}.T_{ST}' \triangleq \generics{D}{T}.(T_{ST} \oplus T_{ST}') \\
\end{align*}

%
%
%
%

\framebox{\textbf{$\funcArg{T_C.T_{STpassed}}{T_C.T_{STinput-decl}}{T_C.T_{SToutput-decl}}$}}

This function specifies the output permission for a function argument that started with a particular permission and was passed to a formal parameter with given initial and final permission specifications. The function is only defined for inputs that correspond with well-typed invocations.

\begin{mathpar}
\inferrule*[right=funcArg-owned-unowned]{
	\maybeOwned{T_C.T_{STpassed}}
}
{
	\funcArg{T_C.T_{STpassed}}{T_C.\Unowned}{T_C.T_{SToutput-decl}} = T_C.T_{STpassed}
}

\inferrule*[right=funcArg-shared-unowned]{
}
{
	\funcArg{T_C.Shared}{T_C.\Unowned}{T_C.T_{SToutput-decl}} = T_C.T_{Shared}
}

\inferrule*[right=funcArg-other]{
	 T_C.T_{STinput-decl} \neq \Unowned
}
{
	\funcArg{T_C.T_{STpassed}}{T_C.T_{STinput-decl}}{T_C.T_{SToutput-decl}} = T_C.T_{SToutput-decl}
}

\end{mathpar}

\framebox{\textbf{$\funcArgResidual{T_C.T_{STpassed}}{T_C.T_{STinput-decl}}{T_C.T_{SToutput-decl}}$}}

This function specifies the type of the reference that remains after an argument is passed to a function.
\begin{mathpar}
\inferrule*[right=FAR-OU]{
	\maybeOwned{T_C.T_{STpassed}}
}
{
	\funcArgResidual{T_C.T_{STpassed}}{T_C.\Unowned}{T_C.T_{SToutput-decl}} = T_C.T_{STpassed}
}

\inferrule*[right=FAR-SU]{
}
{
	\funcArgResidual{T_C.\Shared}{T_C.\Unowned}{T_C.T_{SToutput-decl}} = T_C.T_{\Shared}
}

\inferrule*[right=FAR-*]{
	 T_C.T_{STinput-decl} \neq \Unowned
}
{
	\funcArgResidual{T_C.T_{STpassed}}{T_C.T_{STinput-decl}}{T_C.T_{SToutput-decl}} = T_C.\Unowned
}

\end{mathpar}

\color{black}
\vspace{1 em}
\framebox{\textbf{Typing contexts $\Delta$} and \textbf{type bound contexts $\typeBounds$}}

The typing context includes both local variables and temporary field types. It is assumed that $\Delta$ and $\typeBounds$ are permuted as needed in order to apply the rules, but when a context is extended with a mapping, the new mapping replaces any previous mapping of the same variable.
$\typeBounds$ is simply a set of generic type variables T\textsubscript{G} as defined in the grammar in Figure~\ref{silica-syntax}.

\begin{tabular}{l r l l}
$\typeBounds$ 	& 	\bnfdef & 	$\cdot$\\
				&	\bnfalt & 	$\typeBounds$, T\textsubscript{G}\\
$\Delta$ 		& 	\bnfdef & 	$\cdot$\\
				&	\bnfalt &	$\Delta$,  x : T\\
				& 	\bnfalt &	$\Delta$,  s.f : T\\
\end{tabular}

\vspace{1em}
\framebox{$\mathbf{\subtype{\typeBounds}{T_1}{T_2}}$} \textbf{Subtyping}

\begin{mathpar}

\inferrule*[right=<:-Unit]{ }{\subtype{\typeBounds}{\Unit}{\Unit}}

\inferrule*[right=<:-Matching-defs]{ \typeBounds \vdash T_{ST} <:_* T_{ST}'}
    {\subtype{\typeBounds}{T_C.T_{ST}}{T_C.T_{ST}'}}

\inferrule*[right=<:-Matching-decls]{
    \typeBounds \vdash T_{ST} <:_* T_{ST}'}
    {\subtype{\typeBounds}{\generics{D}{T}.T_{ST}}{\generics{D}{T}.T_{ST}'}}

\inferrule*[right=<:-Implements-interface]{ \typeBounds \vdash T_{ST} <:_* T_{ST}'
    \\\\
    def(C) = \contract \ \generics{C}{T_G} \ \implements\ \generics{I}{T'} \{ \ldots \} }
    {\subtype{\typeBounds}{\generics{C}{T}.T_{ST}}{\substitute{\overline{T/T_G}}{\generics{I}{T'}.T_{ST}'}} }

\inferrule*[right=<:-Bound]{ \typeBounds \vdash T_{ST} <:_* T_{ST}' \and \bound{\typeBounds}{X.T_{ST}} = T_C.T_{ST}'}
    {\subtype{\typeBounds}{X.T_{ST}}{T_C.T_{ST}'}}
\end{mathpar}

\framebox{$\subperm{\typeBounds}{T_{ST1} }{T_{ST2}}$} \textbf{Subpermissions}

The subpermission judgment is ancillary to the subtyping judgment, and specifies when an expression with one mode can be used where one with the same contract but a potentially different mode is expected.
\begin{mathpar}

\inferrule*[right=$<:_*$-Refl]{ }{ \subperm{\typeBounds}{T_{ST}}{T_{ST}} }

\inferrule*[right=$<:_*$-Trans]{ \subperm{\typeBounds}{T_{ST_1}}{T_{ST_2}} \and \subperm{\typeBounds}{T_{ST_2}}{T_{ST_3}} }
    { \typeBounds \vdash T_{ST_1} <: T_{ST_3} }

\inferrule*[right=$<:_*$-Var]{ \boundPerm{\typeBounds}{p} = T_{ST}}{\subperm{\typeBounds}{p}{T_{ST}}}

\inferrule*[right=$<:_*$-S-S']{ \overline{S} \subseteq \overline{S'} }{\subperm{\typeBounds}{\overline{S} }{\overline{S'}}}

\inferrule*[right=$<:_*$-S-O]{ }{\overline{S} <:_* Owned}

\inferrule*[right=$<:_*$-O-*]{T_{ST2} \neq \overline{S} }{\subperm{\typeBounds}{Owned}{T_{ST2}}}

\inferrule*[right=$<:_*$-U-U]{ }{\subperm{\typeBounds}{T_{ST}}{Unowned} }

\end{mathpar}

\framebox{$\notsubperm{\typeBounds}{T_{ST}}{T_{ST}}$}
\begin{mathpar}
\inferrule*{ \subperm{\typeBounds}{T_{ST_2}}{T_{ST_1}} \and T_{ST_2} \neq T_{ST_1} }
    { \notsubperm{\typeBounds}{T_{ST_1}}{T_{ST_2}} }
\end{mathpar}

\framebox{$\sameOwnership{T_1}{T_2}$} \textbf{Ownership equality}
\begin{mathpar}
\inferrule*[right=$\approx$-Refl]{ }{ \sameOwnership{T_1}{T_1}}

\inferrule*[right=$\approx$-Sym]{\sameOwnership{T_1}{T_2}}{ \sameOwnership{T_2}{T_1}}

\inferrule*[right=$\approx$-Trans]{\sameOwnership{T_1}{T_2} \and \sameOwnership{T_2}{T_3}}{ \sameOwnership{T_1}{T_3}}

\inferrule*[right=$\approx$-O-O] { \maybeOwned{T_1} \and \maybeOwned{T_2} }
    { \sameOwnership{T_1}{T_2} }

\inferrule*[right=$\approx$-U-U] { \notOwned{T_1} \and \notOwned{T_2} }
    { \sameOwnership{T_1}{T_2} }

\end{mathpar}

\framebox{$\splitType{T_1}{T_2}{T_3}$} \textbf{Type splitting}

Type splitting specifies how ownership of objects can be shared among aliases. In \splitType{T_1}{T_2}{T_3}, there is initially one reference of type $T_1$; afterward, there are two references of type $T_2$ and $T_3$.

\begin{mathpar}
\
\inferrule*[right=Split-Unowned]{ T_C = contract(T)}{\splitType{T}{T}{T_C.Unowned}}

\inferrule*[right=Split-shared]{ }{\splitType{T_C.Shared}{T_C.Shared}{T_C.Shared}}

\inferrule*[right=Split-owned-shared]{ \nonAsset{\typeBounds}{T_C.T_{ST}} \and \maybeOwned{T_C.T_{ST}} }{\splitType{T_C.T_{ST}}{T_C.Shared}{T_C.Shared}}


\inferrule*[right=Split-unit]
{
}
{
	\splitType{\Unit}{\Unit}{\Unit}
}
\end{mathpar}

\framebox{\textoverline{S} \textbf{ok}} \textbf{Well-formed state sequence}\\
Well-formed states cannot have conflicts regarding ownership, and if any states are specified, then Owned would be redundant. There must be no duplicates in the list.

\begin{mathpar}
\inferrule*{ \overline{S} \ \textbf{statename-list} }{\overline{S} \  \ok}
\end{mathpar}

\framebox{\textoverline{S} \textbf{statename-list}} \textbf{Well-formed statename sequence}
\begin{mathpar}
\inferrule*{ }{
	S \ \textbf{statename-list}
}

\inferrule*{ \overline{S} \ \textbf{statename-list} \and S' \notin \overline{S}}{
	\overline{S}, S' \ \textbf{statename-list}
}
\end{mathpar}

\framebox{T \textbf{ok}} \textbf{Well-formed type}
\begin{mathpar}
\inferrule*{ \overline{S} \ \mathbf{ok}}{
	T_C.\overline{S} \ \mathbf{ok}
}

\inferrule*{ }{
    \ok{T_C.p}
}

\inferrule*{ }{
    \ok{\Unit}
}

\end{mathpar}

\subsection{Silica Dynamic Semantics}
\label{dynamic-semantics}
In order to express the dynamic semantics, we must first slightly extend the syntax. We introduce a notion of \textit{locations} $l$, which are used only in the formal semantics, not in the implementation, as a tool to prove soundness. Locations, introduced in FT \cite{Garcia:2014:FTP:2684821.2629609}, allow the formal model to track permissions of individual aliases to shared objects. Intuitively, locations typically correspond with local variables that reference objects; we record the permission each indirect reference holds in a context $\rho$.

\begin{align*}
o 				& \in \textsc{ObjectRefs} \\ 
l 			& \in \textsc{IndirectRefs}\\ 
\generics{C}{T}.S (\overline{o}) & \in \textsc{Objects}\\
\permVarMap 	& \in \textsc{PermissionVariables} \rightharpoonup \textsc{StateNames} \cup \{ \Owned, \Unowned, \Shared \}\\
\mu 			& \in \textsc{ObjectRefs} \rightharpoonup \textsc{Objects}\\ 
\rho 			& \in \textsc{IndirectRefs} \rightharpoonup \textsc{Values}\\ 
\end{align*}

\begin{tabular}{l r l l}
e 	& 	\bnfdef &	\ldots \bnfalt o\\
& 		\bnfalt & \phibox{e}{o} & (state-locking mutation detection container)\\
& 		\bnfalt & \psibox{e}{o} & (reentrancy detection container) \\

$s$ &	\bnfdef & 	\ldots \bnfalt $l$\\

$v$ &	\bnfdef &	\unit \bnfalt $o$ &   (values) \\

$\phi$ & \bnfdef & $\cdot$ \bnfalt $\phi$, o & (Objects that are state-locked)\\

$\psi$ & \bnfdef & $\cdot$ \bnfalt $\psi$, o & (Objects that have transactions that are on the stack)\\

$\mathbb{E}$ & \bnfdef & $\Box$ \\
			 & \bnfalt & let $x = \mathbb{E}$ in $e$\\
			 & \bnfalt & $\fbox{$\mathbb{E}$}_o$\\
			 & \bnfalt & $\fbox{$\mathbb{E}$}^o$\\
\end{tabular}

We extend the previous definition of static contexts so that programs can remain well-typed as they execute:

\begin{tabular}{l r l}
$b$ & $\in$ & $x$ \bnfalt $l$ \bnfalt $o$ \\
$\Delta$ & \bnfdef & $\overline{b : T}$ \\
\end{tabular}

We extend the previous \textit{T-lookup} rule to account for this extension:
\begin{mathpar}
\inferrule*[right=T-lookup]{\splitType{T_1}{T_2}{T_3}}{\ty{\typeBounds}{\Delta, b: T_1}{s}{b}{T_2}{\Delta, b: T_3}}
\end{mathpar}

The abstract machine maintains state $\langle \mu, \rho, \phi, \psi, \permVarMap \rangle$. For concision, we abbreviate that tuple as $\Sigma$ and refer to the components as $\Sigma_\mu$, etc. $\mu$ is used as an abbreviation for $\Sigma_\mu$ when there is only one $\Sigma$ in scope. The syntax \envUpdate{X}{\mu}{\Sigma} denotes $\langle X, \rho, \phi, \psi, \permVarMap \rangle$; $\mu$ denotes $\Sigma_\mu$ if it occurs in X.

The dynamic semantics are similar to the dynamic semantics of FT. However, in addition to heap $\mu$ and environment $\rho$, we keep a state-locking environment $\phi$, which is a set of references to objects that are state-locked. $\phi$ is modified for \code{is in} and checked as needed for safety, depending on the static types.

In the scope of an \code{if in} block, we must ensure that other aliases cannot be used to violate the state assumptions of the block. We only check for state modification, not for general field writes, since the typestate mechanism is restricted to nominal states rather than pertaining to all properties of objects.

Although Obsidian lacks the dynamic \code{assert} statement that can cause FT programs to get stuck, Obsidian's dynamic state locking can result in an expression getting stuck. While in the scope of a dynamic state lock, transitions to a different state through a reference that does not hold the ownership endowed by the dynamic check cause the semantics to get stuck.

Reentrancy is checked dynamically at object granularity. Object-level reentrancy aborts the current top-level transaction. However, as a special exception, private transactions are not protected from reentrancy (otherwise they would be useless). Reentrancy is checked via the $\psi$ context, which is a set of all objects that have transaction invocations on the stack.

\framebox{$\stepsTo{\Sigma, e}{\Sigma', e'}$}
\begin{mathpar}
\inferrule*[right=E-lookup]
{
}
{
	\stepsTo{\Sigma, l}{\Sigma, \rho(l)}
}

\inferrule*[right=E-let]
{
	l \notin dom(\rho)
}
{
	\stepsTo{\Sigma, \letExpr{x}{T}{v}{e}}{\envUpdate{\rho[l \mapsto v]}{\rho}{\Sigma}, [l/x]e}
}

\inferrule*[right=E-letCongr]
{
	\stepsTo{\Sigma, e_1}{\Sigma', e_1'}
}
{
	\stepsTo{\Sigma, \letExpr{x}{T}{e_1}{e_2}}{\Sigma', \letExpr{x}{T}{e_1'}{e_2}}
}

\inferrule*[right=E-new]
{
	o \notin dom(\mu) \and
    def(C) = \contract \ \generics{C}{T_G} \ \implements \ \generics{I}{T} \{ \ldots \}
}
{
    \stepsTo{\Sigma, \text{new } \generics{C}{T}.S(\overline{l})}{\envUpdate{\mu[o \mapsto \generics{C}{T}.S(\overline{\rho(l)})]}{\mu}{\Sigma}, o}
}

\inferrule*[right=E-field]
{
    \mu(\rho(l)) = \generics{C}{T}.S(\overline{s})
}
{
	\stepsTo{\Sigma, l.f_i}{\Sigma, s_i}
}

\inferrule*[right=E-fieldUpdate]
{
    \mu(\rho(l)) = \generics{C}{T}.S(\overline{l}) \and
	fields(C.S) = \overline{T \ f}
}
{
    \stepsTo{\Sigma, l.f_i := l'}{\envUpdate{\mu[\rho(l) \mapsto \generics{C}{T}.S(o_1, o_2, \ldots, o_{i-1}, \rho(l'), o_{i+1}, \ldots, o_{|l|})]}{\mu}{\Sigma}, \unit}
}
\end{mathpar}

The two invocation rules are complex. First, we look up the receiver in the heap to find its dynamic state. In invocations of public methods, we also must check that there is not already an invocation on the receiver in progress. Then, we make fresh indirect references $l_1'$ and $\overline{l_2'}$, which will be used to pass ownership to the transaction; residual ownership will remain in the original indirect references $l_1$ and $\overline{l_2}$. Then, since $e$ may use type parameters according to the declarations of $C$ and $tdef(C, m)$, we need to update $\permVarMap$ so that the variables are bound according to the invocation by resolving any type variables to concrete permissions or states (via \textit{lookup}). Then, we proceed by substitution in an environment that tracks an in-progress invocation on the object referenced by $l_1$, which is referenced by $\rho(l_1)$. This object reference must be removed from the context after evaluation, since the dynamic state tests in part function as sequence operators. To arrange this, the rule steps to an expression in a box. Afterward, evaluation will proceed inside the box until the contents of the box reaches a value, at which point the value is unboxed and the reference is removed from $\psi$.

We define $\lookup{\permVarMap}{T_{ST}}$ so that it looks up $T_{ST}$ in $\permVarMap$ if $T_{ST}$ is a variable, and otherwise, simply maps to $T_{ST}$.
This definition ensures that each permission variable maps to a concrete permission or state, rather than a permission variable, eliminating the need for recursive lookups.

\begin{mathpar}
\inferrule*[right=E-inv]
{
    \mu(\rho(l_1)) = \generics{C}{T}.S(\ldots) \and
    \rho(l_1) \notin \psi
	\\\\
    tdef(C, m) = T \ \generics{m}{T_M}(\overline{T_{C_x}.T_x \trans T_{xST} \ x}) T_{this} \ \trans \ T_{this}' \ e \and
	\\\\
	l_1' \nin dom(\rho) \and \overline{l_2' \nin dom(\rho)} \and
    \params{C} = \overline{T_D}
    \\\\
    \permVarMap' = \permVarMap, \overline{PermVar(T_D) \mapsto \lookup{\permVarMap}{Perm(T)}}, \overline{PermVar(T_M) \mapsto \lookup{\permVarMap}{Perm(M)}}
    \\\\
    \Sigma' = \envUpdate{l_1' \mapsto \rho(l_1)}{\overline{l_2' \mapsto \rho(l_2)}}\envUpdate{\permVarMap'}{\permVarMap}{\envUpdate{\psi, \rho(l_1)}{\psi}{\Sigma}}
}
{
    \stepsTo{\Sigma, l_1.\generics{m}{M}(\overline{l_2})}{\Sigma', \fbox{$\overline{[l_2'/x]}[l_1'/\this]e$}^{\rho(l_1)}}
}

\inferrule*[right=E-privInv]
{
    \mu(\rho(l_1)) = \generics{C}{T}.S(\ldots)
	\\\\
    tdef(C, m) = \overline{T_{C_f}.T_{fdecl} \trans T_{fST}} \ T \ \generics{m}{T_M}(\overline{T_{C_x}.T_{x} \trans T_{xST} \ x}) \  T_{this} \trans \ T_{this}' \ e
	\\\\
	l_1' \nin dom(\rho) \and \overline{l_2' \nin dom(\rho)} \and
    \params{C} = \overline{T_D}
    \\\\
    \permVarMap' = \permVarMap, \overline{PermVar(T_D) \mapsto \lookup{\permVarMap}{Perm(T)}}, \overline{PermVar(T_M) \mapsto \lookup{\permVarMap}{Perm(M)}}
    \\\\
    \Sigma' = \envUpdate{l_1' \mapsto \rho(l_1)}{\overline{l_2' \mapsto \rho(l_2)}}\envUpdate{\permVarMap'}{\permVarMap}{\envUpdate{\psi, \rho(l_1)}{\psi}{\Sigma}}
}
{
	\stepsTo{\Sigma, l_1.\generics{m}{M}(\overline{l_2})}{\Sigma', \overline{[l_2'/x]}[l_1'/\this]e}
}
\end{mathpar}

\begin{center}
{\fbox{\begin{minipage}{10cm}
\framebox{$\lookup{\permVarMap}{T_{ST}} = T_{ST}$}
\begin{mathpar}
    \inferrule*{ }
        { \lookup{\permVarMap}{p} = \permVarMap(p) }

    \inferrule*{ \nonVar{T_{ST}} }
        { \lookup{\permVarMap}{T_{ST}} = T_{ST} }
\end{mathpar}
\end{minipage}
}}
\end{center}

\begin{mathpar}
\inferrule*[right=E-$\nearrow_{owned}$]
{
    \mu(\rho(l)) = \generics{C}{T}.S'(\ldots)
}
{
    \stepsTo{\Sigma, l \nearrow_{owned} S(\overline{l'})}{\envUpdate{\mu[\rho(l) \mapsto \generics{C}{T}.S(\overline{\rho(l')})]}{\mu}{\Sigma}, \unit}
}
\end{mathpar}

In E-$\nearrow_{shared}$, a shared object can transition state if it is not statelocked or the transition does not actually change which state the object is in.
\begin{mathpar}
\inferrule*[right=E-$\nearrow_{shared}$]
{
    \mu(\rho(l)) = \generics{C}{T}.S'(\ldots) \and
	\rho(l) \notin \phi \lor S = S'
}
{
    \stepsTo{\Sigma, l \nearrow_{shared} S(\overline{l'})}{\envUpdate{\mu[\rho(l) \mapsto \generics{C}{T}.S(\overline{\rho(l')})]}{\mu}{\Sigma}, \unit}
}

\inferrule*[right=E-assert]
{
}
{
	\stepsTo{\Sigma, \assertExpr{s}{T_{ST}}}{\Sigma, \unit}
}

\inferrule*[right=E-IsIn-PermVar]
{
    \permVarMap(p) = T_{ST}
}
{
    \stepsTo{\Sigma, \text{if } l \text{ is in\textsubscript{P} } p \text{ then } e_1 \text{ else } e_2}{\Sigma, \text{if } l \text{ is in\textsubscript{P} } T_{ST} \text{ then } e_1 \text{ else } e_2}
}

\inferrule*[right=E-IsIn-Perm-Then]
{
    \text{Perm} \in \{ \Owned, \Unowned, \Shared \} \and
    \subperm{\cdot}{P}{\text{Perm}}
}
{
    \stepsTo{\Sigma, \text{if } l \text{ is in\textsubscript{P} } \text{Perm} \text{ then } e_1 \text{ else } e_2}{\Sigma, e_1 }
}

\inferrule*[right=E-IsIn-Perm-Else]
{
    \text{Perm} \in \{ \Owned, \Unowned, \Shared \} \and
    \notsubperm{\cdot}{\text{Perm}}{P}
}
{
    \stepsTo{\Sigma, \text{if } l \text{ is in\textsubscript{P} } \text{Perm} \text{ then } e_1 \text{ else } e_2}{\Sigma, e_2 }
}

\inferrule*[right=E-IsIn-Unowned]
{
}
{
    \stepsTo{\Sigma, \text{if } l \text{ is in\textsubscript{Unowned} } \overline{S} \text{ then } e_1 \text{ else } e_2}{\Sigma, e_2 }
}

\inferrule*[right=E-IsIn-Owned-Then]
{
	\mu(\rho(l)) = \generics{C}{T}.S'(\ldots) \and
    S' \in \overline{S}
}
{
    \stepsTo{\Sigma, \text{if } l \text{ is in\textsubscript{owned} } \overline{S} \text{ then } e_1 \text{ else } e_2}{\Sigma, e_1 }
}
\end{mathpar}

In E-IsIn-Shared-Then, we check $\rho(l) \notin \phi$ because the static semantics that correspond generate a temporary owning reference. If we didn't check, or we allowed nested checks, we would generate multiple distinct temporary owning references.

\begin{mathpar}
\inferrule*[right=E-IsIn-Shared-Then]
{
    \mu(\rho(l)) = \generics{C}{T}.S(\ldots) \and \rho(l) \notin \phi
}
{
	\stepsTo{\Sigma, \text{if } l \text{ is in\textsubscript{shared} } \overline{S} \text{ then } e_1 \text{ else } e_2}{\envUpdate{\phi, \rho(l)}{\phi}{\Sigma}, \fbox{$e_1$}_{\rho(l)} }
}

\inferrule*[right=E-IsIn-Else]
{
    \mu(\rho(l)) = \generics{C}{T}.S'(\ldots) \and S' \nin \overline{S}
}
{
	\stepsTo{\Sigma, \text{if } l \text{ is in}_p \  \overline{S} \text{ then } e_1 \text{ else } e_2}{ \Sigma, e_2 }
}

\inferrule*[right=E-disown]
{
}
{
	\stepsTo{\Sigma, \disown s}{\Sigma, \unit}
}

\inferrule*[right=E-pack]
{
}
{
	\stepsTo{\Sigma, \pack}{\Sigma, \unit}
}
\end{mathpar}

Silica maintains two contexts that FT does not. $\phi$ records which objects are currently state-locked because evaluation is currently inside the body of a dynamic state check. $\psi$ records which objects are receivers of invocations that are on the stack in order to detect reentrancy. In order to remove objects from these two contexts, we introduce two new constructs,  \phibox{e}{o} and \psibox{e}{o}. They each permit the boxed expression to first evaluate to a value, and then afterward remove the corresponding object reference from the appropriate context.

\begin{mathpar}
\inferrule*[right=T-state-mutation-detection]
{
    \ty{\typeBounds}{\Delta}{s}{e}{T}{\Delta'}
}
{
    \ty{\typeBounds}{\Delta}{s}{\phibox{e}{o}}{T}{\Delta'}
}
\inferrule*[right=T-reentrancy-detection]
{
    \ty{\typeBounds}{\Delta}{s}{e}{T}{\Delta'}
}
{
    \ty{\typeBounds}{\Delta}{s}{\psibox{e}{o}}{T}{\Delta'}
}

\inferrule*[right=E-box-$\phi$]
{
}
{
	\stepsTo{\Sigma, \phibox{$v$}{o}}{\envUpdate{(\phi \setminus o)}{\phi}{\Sigma}, v}
}

\inferrule*[right=E-box-$\psi$]
{
}
{
	\stepsTo{\Sigma, \psibox{$v$}{o}}{\envUpdate{(\psi \setminus o)}{\psi}{\Sigma}, v}
}

\inferrule*[right=E-box-$\phi$-congr]
{
	\stepsTo{\Sigma, e}{\Sigma', e'}
}
{
	\stepsTo{\Sigma, \phibox{$e$}{o}}{\Sigma', \phibox{$e'$}{o}}
}

\inferrule*[right=E-box-$\psi$-congr]
{
	\stepsTo{\Sigma, e}{\Sigma', e'}
}
{
	\stepsTo{\Sigma, \psibox{$e$}{o}}{\Sigma', \psibox{$e'$}{o}}
}

\end{mathpar}

\subsection{Silica Soundness and Asset Retention}
In this section, we outline the proof of type soundness. We also state the \textit{asset retention} theorem, which formally states the property that owned references to assets can only be dropped with the \disown operation. Full proofs can be found in the appendix.

As with FT, we extend type contexts to include object references $o$ and indirect references $l$.

\begin{align*}
b &\in x \; | \; l \; | \; o \\
\Delta & ::= \overline{b : T} \\
\end{align*}

Then we must use revised rules that take these references into account:

\begin{mathpar}
\inferrule*[right=T-Lookup]{\splitType{T_1}{T_2}{T_3}}{\ty{\typeBounds}{\Delta, b: T_1}{s}{b}{T_2}{\Delta, b: T_3}}
\end{mathpar}

We need typing for \unit:

\begin{mathpar}
\inferrule*[right=T-()]
{
}
{
	\ty{\typeBounds}{\Delta}{s}{\unit}{\Unit}{\Delta}
}

\end{mathpar}

We extend simple expressions to include indirect references:

\begin{tabular}{l r l l}
s & \bnfdef 	& x \\
&	\bnfalt 	& l \\
\end{tabular}

\textit{Global consistency} defines consistency among static and runtime environments. It requires that every indirect reference to an object in $\rho$ maps to a legitimate indirect reference in $\mu$ and that $\rho$ maps indirect references to appropriately-typed values. It also requires that every type in the static context correspond with an indirect reference in the indirect reference context. The permission variables must be available for lookup in $\permVarMap$ and map to concrete permissions or states. Finally, every object in the heap must have only compatible aliases, as expressed by \textit{reference consistency}.

\framebox{\ok{\typeBounds, \Sigma, \Delta}} Global Consistency
\begin{mathpar}
\inferrule*
{
	range(\rho) \subset dom(\mu) \cup \{ \unit \}  \\\\ 
	dom(\Delta) \subset dom(\rho) \cup dom(\mu) \\\\ 
	\{l \; | \; (l: \Unit) \in \Delta \} \subset \{ l \; | \; \rho(l) = \unit \} \\\\ 
	\{l \; | \; (l: boolean) \in \Delta \} \subset (\{ l \; | \; \rho(l) \in \{true, false \} \} \\\\ 
	\{l \; | \; (l: T) \in \Delta \} \subset \{ l \; | \; \rho(l) = o \} \\\\ 
    \PermVar{\typeBounds} \subset \{ p \; | \; \permVarMap(p) = T_{ST} \} \\\\ 
    \forall s : T_C.T_{ST} \in \Delta, \exists C, \overline{T} ~ \text{s.t.} ~ T_C = \generics{C}{T} \\\\ 
	\Sigma, \Delta \vdash \ok{dom(\mu)}
}
{
	\ok{\typeBounds, \Sigma, \Delta}
}
\end{mathpar}

\textit{Reference consistency} expresses the requirement that all aliases to a given object must be compatible with each other and consistent with the actual type of the object in the heap. It also requires that objects in the heap have the right number of fields. The fact that the fields must reference objects of appropriate type is implied by the requirement that all references must reference objects of types consistent with the reference types.

\framebox{$\Sigma, \Delta \vdash \ok{o}$} Reference Consistency
\begin{mathpar}
\inferrule*
{
    \mu(o) = \generics{C}{T}.S(\overline{o'}) \and |\overline{o'}| = |stateFields(C, S)| \\\\
	refTypes(\Sigma, \Delta, o) = \overline{D} \and
    \subtype{\cdot}{\generics{C}{T}}{\overline{D}} \\\\
    \forall T_1, T_2 \in \overline{D}, \compatibleTypes{T_1}{T_2} \text{ or } StateLockCompatible(T_1, T_2)
}
{
	\Sigma, \Delta \vdash \ok{o}
}
\\
\text{where}
\\
StateLockCompatible(T_1, T_2) \triangleq o \in \Sigma_\phi \land ((i \neq j) \implies owned(T_i) \land T_j = \generics{C}{T}.Shared)
\end{mathpar}

StateLockCompatible is defined in order to allow the original Shared alias (via which the state was checked) to co-exist with the state-specifying reference. This would not normally be permitted, but is safe while $o \in \Sigma_\phi$ because the shared alias cannot be used to mutate typestate while that is the case.

The relation \compatibleTypes{}{} defines compatibility between pairs of aliases:

\framebox{\compatibleTypes{T_1}{T_2}} Alias Compatibility
\begin{mathpar}

\inferrule*[right=SymCompat]{\compatibleTypes{T_2}{T_1}}
	{\compatibleTypes{T_1}{T_2}}

\inferrule*[right=SubtypeCompat]{\compatibleTypes{\generics{C}{T}.T_{ST}}{\generics{C}{T}.T_{ST}'}}
    {\compatibleTypes{\generics{C}{T}.T_{ST}}{\generics{I}{T}.T_{ST}}}

\inferrule*[right=ParamCompat]{\compatibleTypes{\generics{C}{T}.T_{ST}}{\generics{C}{T}.T_{ST}'}}
    {\compatibleTypes{\generics{C}{T}.T_{ST}}{\generics{C}{T'}.T_{ST}'}}

\inferrule*[right=UUCompat]{ }
    {\compatibleTypes{T_C.Unowned}{T_C.Unowned}}

\inferrule*[right=USCompat]{ }
    {\compatibleTypes{T_C.Unowned}{T_C.Shared}}

\inferrule*[right=UOCompat]{ }
    {\compatibleTypes{T_C.Unowned}{T_C.Owned}}

\inferrule*[right=UStatesCompat]{ }
    {\compatibleTypes{T_C.Unowned}{T_C.\overline{S}}}

\inferrule*[right=SCompat]{ }
    {\compatibleTypes{T_C.\Shared}{T_C.\Shared}}

\end{mathpar}

$refTypes$ computes the set of types of referencing aliases to a given object in a given static and dynamic context. References may be from fields of objects in the heap; from indirect references; and from variables in the static context. Fields of objects in the heap include both fields whose types are specified in their declarations and fields whose types are overridden temporarily in the static context $\Delta$.

\begin{align*}
refTypes(\Sigma, \Delta, o) &= refFieldTypes(\mu, o) \concat envTypes(\Sigma, o) \concat ctxTypes(\Delta, o) \\
refFieldTypes (\mu, o) &= \concat_{o' \in dom(\mu)} [T_i \; |\; \mu(o') = \generics{C}{T}.S(\overline{o}), ft(\Delta, \generics{C}{T}, S) = \overline{T f}\text{ and } o \in \overline{o}]\\
ft (\Delta, C, S) &= [T f \; | \; s.f: T \in \Delta] \cup (allFields(C, S) \setminus [T f \; | \; s.f \in dom(\Delta)])\\
envTypes (\Sigma, \Delta, o) &= \concat_{l \in dom(\rho)} [T \; |\; \Sigma_\rho(l) = o \text{ and } (l: T) \in \Delta]\\
ctxTypes (\Delta, o) &= [T \; | \; o : T \in \Delta] \\
\end{align*}

\begin{definition}[$<^l$]
    A context $\Delta$ is l-stronger than a context $\Delta'$ with respect to $\typeBounds$ and $\Sigma$ (denoted $\lStronger{\Delta}{\typeBounds, \Sigma}{\Delta'}$) if and only if for all $l' : T' \in \Delta'$,  there is some $T$ and $l$ such that $\subtype{\typeBounds}{T}{T'}$, $l : T \in \Delta$, \sameOwnership{T}{T'}, and $\Sigma'_\rho(l) = \Sigma'_\rho(l')$.
\end{definition}

Note that this differs from the definition of $<^l$ given in FT. Here, the indirect reference in the two contexts need not match. This weakening is necessary because permissions are split in invocations. After an invocation, the new expression typechecks in a context that may  retain some of the permissions from the original reference (whereas the remaining permissions were transferred to the invocation, i.e. retained in a \textit{different} indirect reference). This means that although the permissions are still at least as strong in the new context, the strongest permission may be held by a different indirect reference than in the original.

\begin{corollary}[$<^l$-reflexivity]
For all $\Delta, \typeBounds, \Sigma$, $\lStronger{\Delta}{\typeBounds, \Sigma'}{\Delta}$.
\end{corollary}
\begin{proof}
Trivial application of the definition because $<:$ is reflexive.
\end{proof}

\begin{lemma}[Canonical forms]
If \ty{\typeBounds}{\Delta}{s}{v}{T}{\Delta'}, then:
\begin{enumerate}
    \item If $T = \generics{C}{T'}.\overline{S}$, then $v = \generics{C}{T'}.S(\overline{s})$.
    \item If $T = \Unit$, then $v = \unit$.
\end{enumerate}
\end{lemma}
\begin{proof}
	By inspection of the typing rules.
\end{proof}

\begin{lemma}[Memory consistency]
\label{memory-consistency}
If $\ok{\typeBounds, \Sigma, \Delta}$, then:
\begin{enumerate}
    \item If $l : \generics{C}{T'}.S \in \Delta$, then $\exists o. \rho(l) = o$ and $\mu(o) = \generics{C}{T'}.S(\overline{s})$.
	\item If $\ty{\typeBounds}{\Delta}{s}{e}{T}{\Delta'}$, and $l$ is a free variable of $e$, then $l \in dom(\rho)$.
\end{enumerate}
\end{lemma}

\begin{theorem}[Progress]
If e is a closed expression and $\ty{\typeBounds}{\Delta}{s}{e}{T}{\Delta'}$, then at least one of the following holds:
\begin{enumerate}
\item e \text{is a value}
\item For any environment $\Sigma$ such that $\ok{\typeBounds, \Sigma, \Delta}$, $\stepsTo{\Sigma, e}{\Sigma', e'}$ for some environment $\Sigma'$
\item $e$ is stuck at a bad state transition --- that is, $e = \mathbb{E}[l \nearrow_{Shared} S(\overline{s})]$ where $\mu(\rho(l)) = \generics{C}{T'}.S'(\ldots)$, $S \neq S'$, $\rho(l) \in \phi$, and $\ty{\typeBounds}{\Delta}{s}{l}{\generics{C}{T'}.Shared}{\Delta'}$.
\item $e$ is stuck at a reentrant invocation -- that is, $e = \mathbb{E}[l.m(\overline{s})]$ where $\mu(\rho(l)) = \generics{C}{T'}.S(\ldots)$, $\rho(l) \in \psi$.
\item $e$ is stuck in a nested dynamic state check -- that is, $e = \mathbb{E}[\ifExpr{s}{shared}{T\textsubscript{ST}}{e_1}{e_2}]$ where  $\mu(\rho(l)) = \generics{C}{T}.S(\ldots)$ and $\rho(l) \in \phi$.

\end{enumerate}
\end{theorem}
\begin{proof}
The proof, which proceeds by induction on the typing derivation, can be found in the appendix.
\end{proof}

\begin{theorem}[Preservation]
If e is a closed expression, $\ty{\typeBounds}{\Delta}{s}{e}{T}{\Delta''}$, \ok{\typeBounds, \Sigma, \Delta}, and \stepsTo{\Sigma, e}{\Sigma', e'} then for some $\Delta'$, $\ty{\typeBounds'}{\Delta'}{s}{e'}{T'}{\Delta'''}$, \ok{\typeBounds', \Sigma', \Delta'}, and  $\lStronger{\Delta'''}{\typeBounds, \Sigma'}{\Delta''}$.
\end{theorem}
\begin{proof}
The proof, which proceeds by induction on the dynamic semantics, can be found in the appendix.
\end{proof}

Informally, \textit{asset retention} is the property that if a well-typed expression $e$ takes a step in an appropriate dynamic context, then owning references to assets are only dropped if $e$ is a \disown operation.
\begin{theorem}[Asset retention]

Suppose:
\begin{enumerate}
	\item \ok{\typeBounds, \Sigma, \Delta}
	\item $o \in dom(\mu)$
	\item $refTypes(\Sigma, \Delta, o) = \overline{D}$
	\item \ty{\typeBounds}{\Delta}{s}{e}{T}{\Delta'}
	\item $e$ is closed
	\item \label{steps} \stepsTo{\Sigma, e}{\Sigma', e'}
	\item $refTypes(\Sigma', \Delta', o) = \overline{D'}$
    \item \label{nd} $\exists T' \in \overline{D}$ such that $\nonDisposable{\typeBounds}{T'}$
    \item \label{d} $\forall T' \in \overline{D'}: \disposable{\typeBounds}{T'}$
\end{enumerate}

Then in the context of a well-typed program, either $\nonDisposable{\typeBounds}{T}$ or $e = \mathbb{E}[\disown \ s]$, where $\rho(s) = o$.
\end{theorem}
\begin{proof}
The proof, which proceeds by induction on the typing derivation, can be found in the appendix.
\end{proof}

\section{Evaluation}
\label{sec:evaluation}
Beyond the formative user studies that helped us design the language, we wanted to ensure that Obsidian can be used to specify typical smart contracts in a concise and reasonable way. Therefore, we undertook two case studies to assess the extent to which Obsidian is suitable for implementing appropriate smart contracts.

Obsidian's type system has significant implications for the design and implementation of software relative to a traditional object-oriented language. We were interested in evaluating several research questions using the case studies:

\begin{enumerate}[label=(RQ\arabic*)]
\item Does the aliasing structure in real blockchain applications allow use of ownership (and therefore typestate)? If so, what are the implications on architecture? Or, alternatively, do so many objects need to be \Shared that the main benefit of typestate is that it helps ensure that programmers insert dynamic tests when required?
\item To what extent does the use of typestate reduce the need for explicit state checks and assertions, which would otherwise be necessary?
\item Can realistic systems be built with Obsidian?
\item To what extent do realistic systems have constructs that are naturally expressed as states and assets?
\end{enumerate}

To address the research questions above, we were interested in implementing a blockchain application in Obsidian. To obtain realistic results, we looked for a domain in which:
\begin{itemize}
\item Use of a blockchain platform for the application provided significant advantages over a traditional, centralized platform.
\item We could engage with a real client to ensure that the requirements were driven by real needs, not by convenience of the developer or by the appropriateness of language features.
\item The application seemed likely to be representative in structure of a large class of blockchain applications.
\end{itemize}

\subsection{Case study 1: Parametric Insurance}
\subsubsection{Motivation}

A summary of this case study based on an earlier version of the language was previously described by \citet{Koronkevich18:Obsidian}\footnote{Unpublished draft.}, but here we provide substantially more analysis and the implementation addresses more use cases, and the previous manuscript was not published in an archival form.

In \textit{parametric insurance}, a buyer purchases a claim, specifying a \textit{parameter} that governs when the policy will pay out. The parameter is chosen so that whether conditions satisfy the parameter can be determined objectively. For example, a farmer might buy drought insurance as parametric insurance, specifying that if the soil moisture index (a property derived from weather conditions) in a particular location drops below $m$ in a particular time window, the policy should pay out. The insurance is then priced according to the risk of the specified event. In contrast, traditional insurance would require that the farmer summon a claims adjuster, who could exercise subjective judgment regarding the extent of the crop damage. Parametric insurance is particularly compelling in places where the potential policyholders do not trust potential insurers, who may send dishonest or unfair adjusters. In that context, potential policyholders may also be concerned with the stability and trustworthiness of the insurer: what if the insurer pockets the insurance premium and goes bankrupt, or otherwise refuses to pay out legitimate claims? 

In order to build a trustworthy insurance market for farmers in parts of the world without trust between farmers and insurers, the World Bank became interested in deploying an insurance marketplace on a blockchain platform. We partnered with the World Bank to use this application as a case study for Obsidian. We used the case study both to \textit{evaluate} Obsidian as well as to \textit{improve} Obsidian, and we describe below results in both categories.

The case study was conducted primarily by an undergraduate who was not involved in the language design,  with assistance and later extensions by the language designers. The choice to have an undergraduate do the case study was motivated by the desire to learn about what aspects of the language were easy or difficult to master. It was also motivated by the desire to reduce bias; a language designer studying their own language might be less likely to observe interesting and important problems with the language. 

We met regularly with members of the World Bank team to ensure that our implementation would be consistent with their requirements. We began by eliciting requirements, structured according to their expectations of the workflow for participants. 

\subsubsection{Requirements}
The main users of the insurance system are \textbf{farmers}, \textbf{insurers}, and \textbf{banks}. Banks are necessary in order to mediate financial relationships among the parties. We assume that farmers have local accounts with their banks, and that the banks can transfer money to the insurers through the existing financial network. Basic assumptions of trust drove the design:
\begin{itemize}
\item Farmers trust their banks, with whom they already do business, but do not trust insurers, who may attempt to pocket their premiums and disappear without paying out policies.
\item Insurers do not trust farmers to accurately report on the weather; they require a trusted weather service to do that. They do trust the implementation of the smart contracts to pay out claims when appropriate and to otherwise refund payout funds to the insurers at policy expiration.
\item There exists a mutually trusted weather service, which can provide signed evidence of weather events.
\end{itemize}

\subsubsection{Design}
Because blockchains typically require all operations to be deterministic and all transactions to be invoked externally, we derived the following design:
\begin{itemize}
\item Farmers are responsible for requesting claims and providing acceptable proof of a relevant weather event in order to receive a payout.
\item Insurers are responsible for requesting refunds when policies expire.
\item A trusted, off-blockchain weather service is available that can, on request, provide signed weather data relevant to a particular query.
\end{itemize}
An alternative approach would involve the weather service handling weather subscriptions. The blockchain insurance service would emit events indicating that it subscribed to particular weather data, and the weather service would invoke appropriate blockchain transactions when relevant conditions occurred. However, this design is more complex and requires trusting the weather service to push requests in a timely manner. Our design is simpler but requires that policyholders invoke the claim transactions, passing appropriate signed weather records.


Our design of the application allows farmers to start the exchange by requesting bids from insurers. Then, to offer a bid, insurers are required to specify a premium and put the potential payout in escrow; this ensures that even if the insurer goes bankrupt later, the policy can pay out if appropriate. If the farmer chooses to purchase a policy, the farmer submits the appropriate payment.

Later, if a weather event occurs that would justify filing a claim, a farmer requests a signed weather report from the weather service. The farmer submits a claim transaction to the insurance service, which sends the virtual currency to the farmer. The farmer could then present the virtual currency to their real-world bank to enact a deposit.

\subsubsection{Results}

The implementation consists of 545 non-comment, non-whitespace lines of Obsidian code. For simplicity, the implementation is limited to one insurer, who can make one bid on a policy request. An overview of the invocations that are sent and results that are received in a typical successful bid and claim scenario is shown in Fig. \ref{case-study-diagram}. All of the objects reside in the blockchain except as noted. The full code for this case study is available online\footnote{https://github.com/mcoblenz/Obsidian/tree/master/resources/case\_studies/Insurance}.

\begin{figure}
\includegraphics[width=\textwidth]{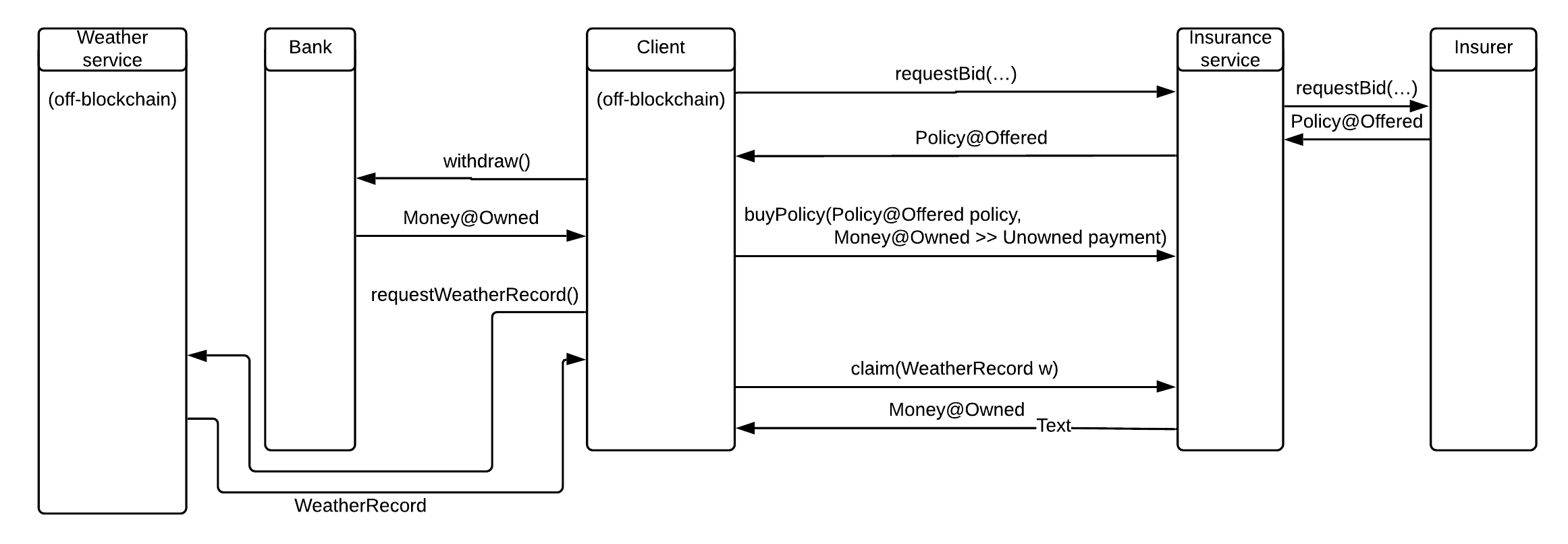}
\caption{Invocations sent and results returned in a typical successful bid/claim scenario.}
\label{case-study-diagram}
\end{figure}

We made several observations about Obsidian. In some cases, we were able to leverage our observations to improve the language. In others, we learned lessons about the implications of the type system on application design and architecture.

First, in the version of the language that existed when the case study started, Obsidian included an \textit{explicit ownership transfer operator} \code{<-}. In that version of the language, passing an owned reference as an argument would only transfer ownership to the callee if the argument was decorated with \code{<-}. For example, \code{deposit(<-m)} would transfer ownership of the reference \code{m} to the \code{deposit} transaction, but \code{deposit(m)} would be a type error because \code{deposit} requires an \Owned reference. While redundant with type information, we had included the \code{<-} operator because we thought it would reduce confusion, but we noticed while using the language (both in the case study and in smaller examples) that its presence was onerous. We removed it, which was a noticeable simplification.

Second, in that version of the language, \code{asset} was a property of contracts. We noticed in the insurance case study that it is more appropriate to think of \code{asset} as a property of states, since some states own assets and some do not. In the case study, an instance of the \code{PolicyRecord} contract holds the insurer's money (acting as an escrow) while a policy is active, but after the policy is expired or paid, the contract no longer holds money (and therefore no longer needs to itself be an asset). It is better to not mark extraneous objects as assets, since assets must be explicitly discarded, and only assets can own assets. Each of those requirements imposes a burden on the programmer. This burden can be helpful in detecting bugs, but should not be borne when not required. We changed the language so that \code{asset} applies to individual states rather than only entire contracts. 

Third, the type system in Obsidian has significant implications on architecture. In a traditional object-oriented language, it is feasible to have many aliases to an object, with informal conventions regarding relationships between the object and the referencing objects. UML also distinguishes between composition, which implies ownership, and aggregation, which does not, reinforcing the idea that ownership in the sense in which Obsidian uses it is common and useful in typical object-oriented designs. Because of the use of ownership in Obsidian, using typestate with a design that does not express ownership sometimes requires refining the design so that it does. In the case study, we found applying ownership useful in refining our design. For example, when an insurance policy is purchased, the insurance service must hold the payout virtual currency until either the policy expires or it is paid. While the insurance service holds the currency, it must associate the currency for a policy with the policy itself. Does the policy, therefore, own the \code{Money}? If so, what is the relationship between the client, who purchased the policy and has certain kinds of control over it, and the \code{Policy}, which cannot be held by the (untrusted) client? We resolved this question by adding a new object, the \code{PolicyRecord}. A \code{PolicyRecord}, which is itself \Owned by the insurance service, has an \Unowned reference to the \code{Policy} and an \Owned reference to a \code{Money} object. This means that \code{PolicyRecord} is an \code{asset} when it is active (because it owns \code{Money}, which is itself an \code{asset}) but \code{Policy} does not need to be an asset. We found that thinking about ownership according to the Obsidian type system helped us refine and clarify our design. Without ownership, we might have chosen a less carefully-considered design.

It is instructive to compare the Obsidian implementation to a partial Solidity implementation, which we wrote for comparison purposes. Figure \ref{solidity-comparison} shows an example of why parts of the Obsidian implementation are substantially shorter. Note how the Solidity implementation requires repeated run time tests to make sure each function only runs when the receiver is in the appropriate state. Obsidian code only invokes those transactions when the \code{Policy} object is in appropriate state; the runtime executes an equivalent dynamic check to ensure safety when the transactions are invoked from outside Obsidian code. Also, the Solidity implementation has \code{cost} and \code{expirationTime} fields in scope when inappropriate, so they need to be initialized repeatedly. In the Obsidian implementation, they are only set when the object is in the \code{Offered} state. Finally, the Solidity implementation must track the state manually via \code{currentState} and the \code{States} type, whereas this is done automatically in the Obsidian implementation. However, Solidity supports some features that are convenient and lead to more concise code: the Solidity compiler automatically generates getters for public fields, whereas Obsidian requires the user to write them manually, and built-in arrays can be convenient. However, the author of the Solidity implementation must be very careful to manage money manually; any money that is received by transactions must be accounted for, or the money will be stuck in the contract forever. Solidity also lacks a math library; completing the implementation would require us to provide our own square root function (which we use to compute distances).

We showed our implementation to our World Bank collaborators, and they agreed that it represents a promising design. There are  various aspects of the full system that are not part of the case study, such as properly verifying cryptographic signatures of weather data, communicating with a real weather service and a real bank, and supporting multiple banks and insurers. However, in only a cursory review, one of the World Bank economists noticed a bug in the Obsidian code: the code always approved a claim requests even if the weather did not justify a claim according to the policy's parameters. This brings to light two important observations. First, Obsidian, despite being a novel language, is readable enough to new users that they were able to understand the code. Second, type system-based approaches find particular classes of bugs, but other classes of bugs require either traditional approaches or formal verification to find.

\begin{figure}
\begin{subfigure}[b]{.49\textwidth}
\begin{lstlisting}[numbers=none, framexleftmargin=0em, xleftmargin=-2em, basicstyle=\tiny\ttfamily]
contract Policy {
  state Offered {
    int cost;
    int expirationTime;
  }

  state Active;
  state Expired;
  state Claimed;

  Policy@Offered(int c, int expiration) {
    ->Offered(cost = c, expirationTime = expiration);
  }

  transaction activate(Policy@Offered >> Active this) {
    ->Active;
  }
  
  

  transaction expire(Policy@Offered >> Expired this) {
    ->Expired;
  }
} 
#\vspace{4em}#  
\end{lstlisting}
\subcaption{Obsidian implementation of a Policy contract.}
\end{subfigure}
\begin{subfigure}[b]{.49\textwidth}
\begin{lstlisting}[numbers=none, framexleftmargin=0em, xleftmargin=0em, basicstyle=\tiny\ttfamily, language=Solidity]
contract Policy {
  enum States {Offered, Active, Expired}
  States public currentState;
  uint public cost;
  uint public expirationTime;
  
  constructor (uint _cost, uint _expirationTime) public {
    cost = _cost;
    expirationTime = _expirationTime;
    currentState = States.Offered;
  }
  
  function activate() public {
    require(currentState == States.Offered,
        "Can't activate Policy not in Offered state.");
    currentState = States.Active;
    cost = 0;
    expirationTime = 0;
  }
  
  function expire() public {
    require(currentState == States.Offered,
        "Can't expire Policy not in Offered state.");
    currentState = States.Expired;
    cost = 0;
    expirationTime = 0;
  }
}
\end{lstlisting}
\subcaption{Solidity implementation of a Policy contract.}
\end{subfigure}
\caption{Comparison between Obsidian and Solidity implementations of a Policy contract from the insurance case study.}
\label{solidity-comparison}
\end{figure}

\subsection{Case study 2: Shipping}
\subsubsection{Motivation}
Supply chain tracking is one of the commonly-proposed applications for blockchains \citep{SupplyChain}. As such, we were interested in what implications Obsidian's design would have on an application that tracks shipments as they move through a supply chain. We collaborated with partners at IBM Research to conduct a case study of a simple shipping application. Our collaborators wrote most of the code, with occasional Obsidian help from us. We updated the implementation to use the polymorphic LinkedList contract, which became available only after the original implementation was done.

\subsubsection{Results}
The final implementation\footnote{https://github.com/laredo/Shipping} consists of 141 non-comment, non-whitespace, non-printing lines of Obsidian code. We found it very encouraging that they were able to write the case study with relatively little input from us, especially considering that Obsidian is a research prototype with extremely limited documentation. Although this is smaller than the insurance case study, we noticed some interesting relationships between the Obsidian type system and object-oriented design.

Fig. \ref{shipping-1} summarizes an early design of the Shipping application\footnote{This version corresponds with git commit 8106e406e8ca005f8878dea5ac78e54b439fe509 in the Shipping repository.}, focusing on a particular ownership problem. The implementation does not compile; the compiler reports three problems. First, \code{LegList}'s \code{arrived} transaction attempts to invoke setArrival via a reference of type \code{Leg@Unowned}; this is disallowed because setArrival changes the state of its receiver, which is unsafe through an \Unowned reference. Second, \code{append} in \code{LegList} takes an \Unowned leg to append, but uses it to transition to the \code{HasNext} state, which requires an \Owned object. Third, \code{Transport}'s \code{depart} transaction attempts to append a new \code{Leg} to its \code{legList}. It does so by calling the \code{Leg} constructor, which takes a \Shared \code{Transport}. But calling this constructor passing an owned reference (\code{this}) causes the caller's reference to become \Shared, not Owned, which is inconsistent with the type of \code{depart}, which requires that \code{this} be owned (and specifically in state \code{InTransport}).

Fig. \ref{shipping-2} shows the final design of the application. This version passes the type checker. Note how a \code{LegList} contains only \code{Arrived} references to \code{Leg} objects. One \code{Leg} may be \code{InTransit}, but that is owned by the \code{Transport} when it is in an appropriate state (also \code{InTransit}). Each \code{Leg} has an \Unowned reference to its \code{Transport}, allowing the \code{TransportList} to own the \code{Transport}. A \code{TransportList} likewise only contains objects in \code{Unload} state; one \code{Transport} in \code{InTransport} state is referenced at the \code{Shipment} level. 

\begin{figure}
\includegraphics[width=\textwidth]{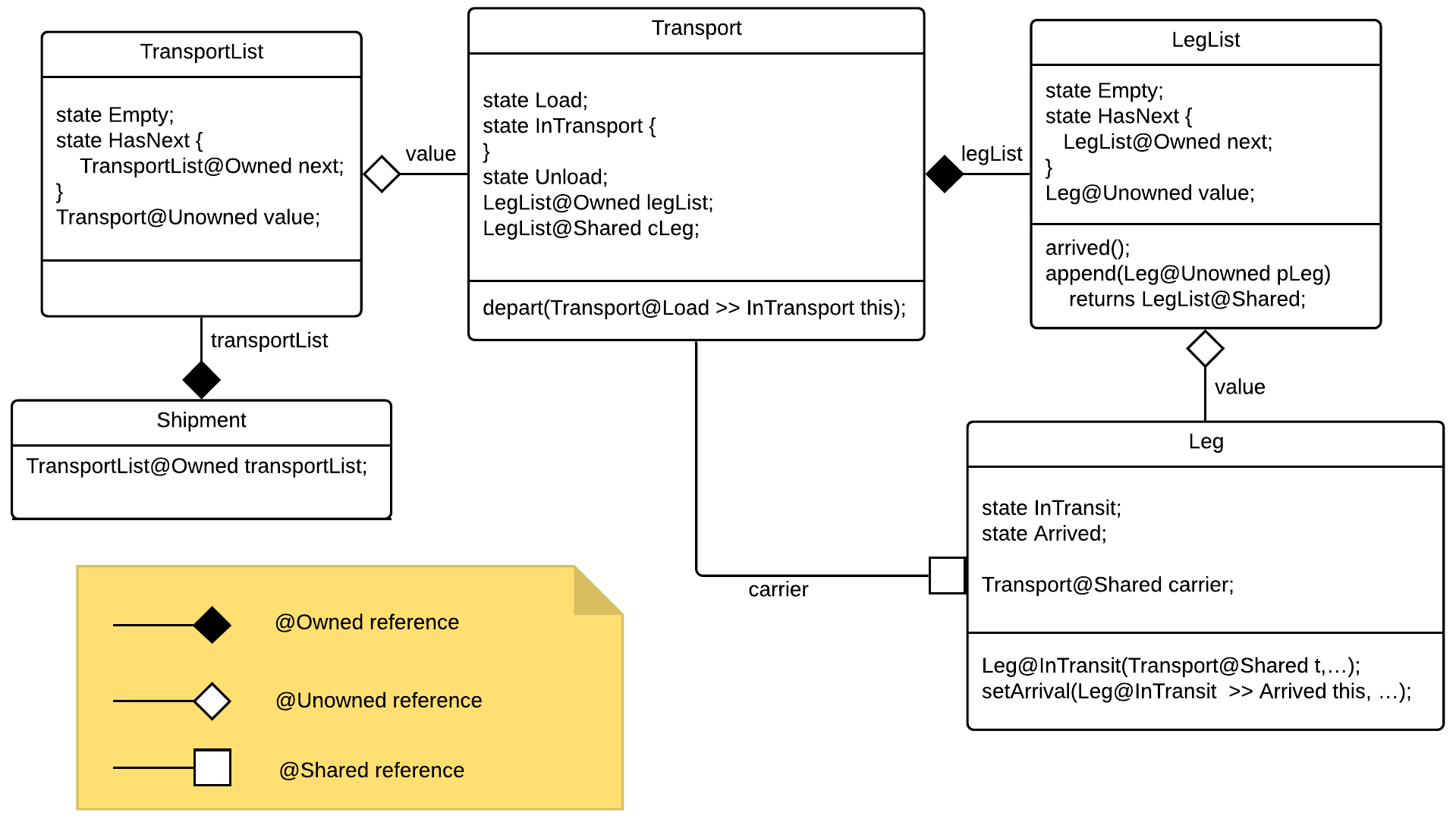}
\caption{Initial design of the Shipping application (which does not compile).}
\label{shipping-1}
\bigskip

\includegraphics[width=\textwidth]{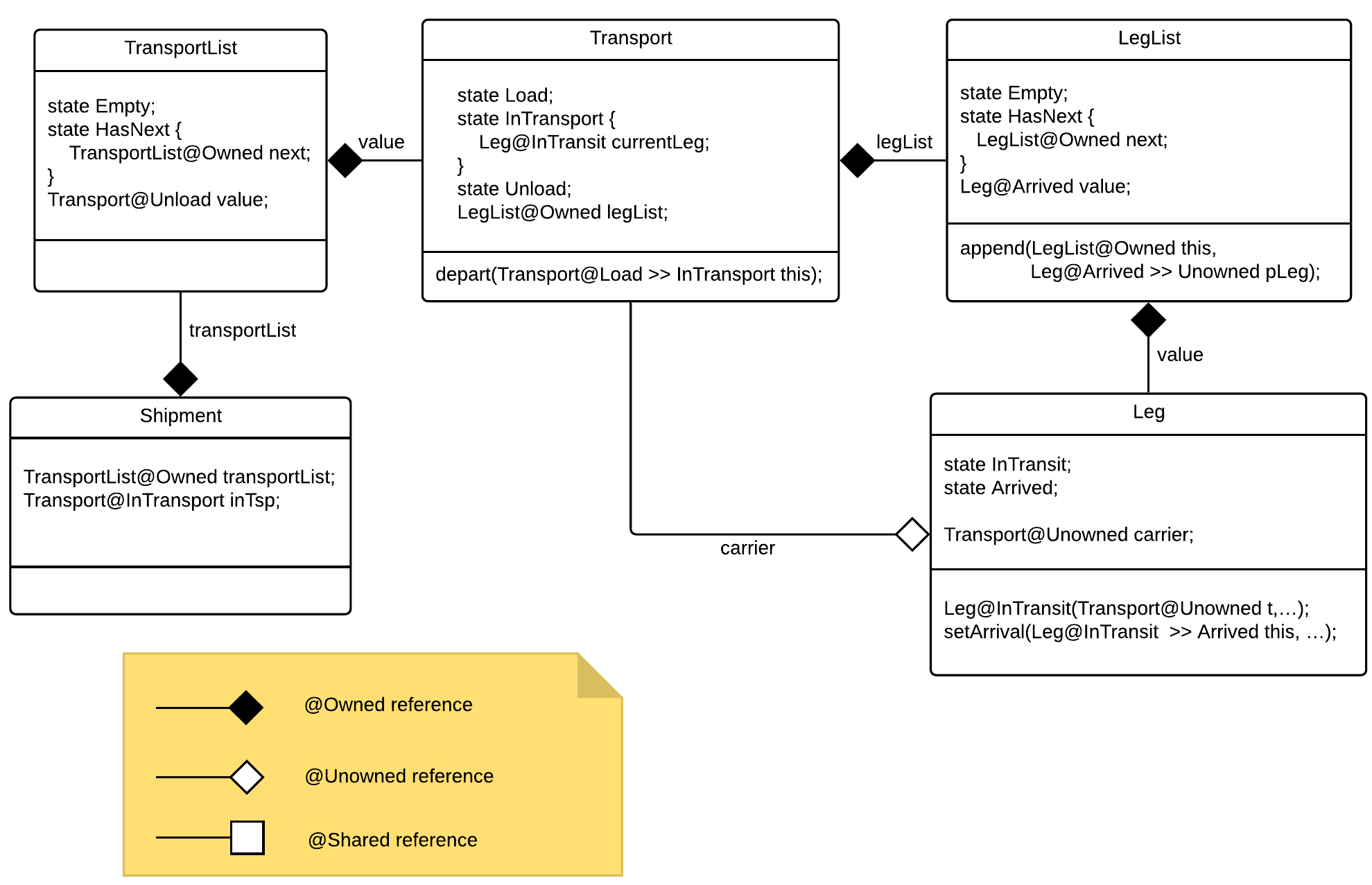}
\caption{Revised design of the Shipping application.}
\label{shipping-2}
\end{figure}

We argue that although the type checker forced the programmer to revise the design, the revised design is better. In the first design, collections (\code{TransportList} and \code{LegList}) contain objects of dissimilar types. In the revised design, these collections contain only objects in the same state. This change is analogous to the difference between dynamically-typed languages, such as LISP, in which collections may have objects of inconsistent type, and statically-typed languages, such as Java, in which the programmer reaps benefits by making collections contain objects of consistent type. The typical benefit is that when one retrieves an object from the collection, there is no need to case-analyze on the element's type, since all of the elements have the same type. This means that there can be no bugs that arise from neglecting to case-analyze, as can happen in the dynamically-typed approach. 

The revised version also reflects a better division of responsibilities among the components. For example, in the first version (Fig. \ref{shipping-1}), \code{LegList} is responsible for both maintaining the list of legs as well as recording when the first leg arrived. This violates the \textit{single responsibility principle} \citep{martin2003agile}. In the revised version, \code{LegList} only maintains a list of \code{Leg} objects; updating their states is implemented elsewhere.

One difficulty we noticed in this case study, however, is that sometimes there is a conceptual gap between the relatively low-level error messages given by the compiler and the high-level design changes needed in order to improve the design. For example, the first error message in the initial version of the application shown in Fig. \ref{shipping-1} is: \code{Cannot invoke setArrival on a receiver of type Leg@Owned; a receiver of type Leg@InTransit is required.} The programmer is required to figure out what changes need to be made; in this case, the \code{arrived} transaction should not be on \code{LegList}; instead, \code{LegList} should only include legs that are already in state \code{Arrived}. We hypothesize that more documentation and tooling may be helpful to encourage designers to choose designs that will be suitable for the Obsidian type system.

We also implemented a version of the Shipping application in Solidity, which required 197 non-comment, non-whitespace lines, so the Obsidian version took 72\% as many lines as the Solidity version. This makes the Obsidian version The translation was straightforward; we translated each state precondition to a runtime assertion, and we flattened fields in states to the top (contract) level. Although the types no longer express the structural constraints that exist in the Obsidian version, the general structure of the code and data structures was identical except that we implemented containers with native Solidity arrays rather than linked lists. As with the prior case study, the additional length required for Solidity was generally due to runtime checks of properties that were established statically in the Obsidian version.

\subsection{Case Study Summary}
We asked six research questions above. We return to them now and summarize what we found.
\begin{enumerate}
\item The aliasing structure in the blockchain applications we implemented \textit{does} allow use of ownership and typestate. However, it forces the programmer to carefully choose an ownership structure, rather than using ad hoc aliases. This can be both restrictive but also result in a simpler, cleaner design.
\item Implementing smart contracts in Solidity typically requires a couple of lines of assertions for every function in smart contracts that are designed to use states. This makes the Obsidian code more concise, although some features that Obsidian currently lacks (such as auto-generated getters) improve concision in Solidity.
\item We and our collaborators were able to successfully build nontrivial smart contracts in Obsidian, despite the fact that Obsidian is a research prototype without much documentation. 
\item The applications that we chose benefited from representation with assets and states, since they represented objects of value and the transactions that were possible at any given time depended on the state of the object. Of course, not every application of smart contracts has this structure.
\end{enumerate}

\section{Future Work}
\label{sec:future-work}
Obsidian is a promising smart contract language, but it should not exist in isolation. Authors of applications for blockchain systems (known as \textit{distributed applications}, or \textit{Dapps}) need to be able to integrate smart contracts with front-end applications, such as web applications. Typically, developers need to invoke smart contract transactions from JavaScript. We would like to build a mechanism for JavaScript applications to safely invoke transactions on Obsidian smart contracts. One possible approach is to embed Obsidian code in JavaScript to enable native interaction, coupled with a mapping between Obsidian objects and JSON.

Obsidian currently has limited IDE support; we plan to improve our existing extension for Visual Studio Code so that programmers can receive live feedback on errors while they edit Obsidian code.

In the current implementation, Obsidian clients invoke all remote transactions sequentially. This means that another remote user might run intervening transactions, violating assumptions of the client program. More discussion of approaches to address this can be found in \S \ref{client-programs}.

The type system-oriented approach in Obsidian is beneficial for many users, but it does not lead to verification of domain-specific program properties. In the future, it would be beneficial to augment Obsidian with a verification mechanism so that users can prove relevant properties of their programs formally.

We plan to conduct a summative user study in which we compare Obsidian to Solidity. We hope to show that programmers using Obsidian successfully complete relevant programming tasks without inserting bugs that the Solidity programmers accidentally insert.

Finally, Obsidian currently only supports Hyperledger Fabric. We would like to target Ethereum as well in order to demonstrate generality of the language as well as to enable more potential users to use the language.

\section{Conclusions}
\label{sec:conclusions}
With Obsidian we have shown how:
\begin{itemize}
\item Typestate can be combined with assets to provide relevant safety properties for smart contracts, including asset retention.
\item A unified approach for smart contracts and client programs can provide safety properties that cannot be provided using the approaches that are currently in use.
\item A core calculus can encode key features of Obsidian and form a sound foundation for the language.
\item Applications can be built successfully with typestate and assets, with useful implications on architecture and object-oriented design.
\end{itemize}

Obsidian represents a promising approach for smart contract programming, including sound foundations and an implementation that enables real programs to execute on a blockchain platform. By formalizing useful safety guarantees and providing them in a programming language that was designed with user input, we hope to significantly improve safety of smart contracts that real programmers write. By combining techniques from human-computer interaction, traditional principles of type system design, and evaluation via case studies, we can obtain a language that is much better than if we used only one of those techniques alone.

\begin{acks}                           
  We appreciate the help of Eliezer Kanal at the Software Engineering Institute, who helped start this project; Jim Laredo, Rick Hull, Petr Novotny, and Yunhui Zheng at IBM, who provided useful technical and real-world insight; and David Gould and Georgi Panterov at the World Bank, with whom we worked on the insurance case study.
  
  This material is based upon work supported by the
  \grantsponsor{GS100000001}{National Science
    Foundation}{http://dx.doi.org/10.13039/100000001} under Grants
  \grantnum{GS100000001}{CNS-1423054} and \grantnum{GS100000001}{CCF-1814826}, by the U.S. Department of Defense, and by \grantsponsor{Ripple}{Ripple}{https://www.ripple.com}. In addition, the first author is supported by an IBM PhD Fellowship. 
  Any opinions, findings, and
  conclusions or recommendations expressed in this material are those
  of the author and do not necessarily reflect the views of the
  National Science Foundation.  
\end{acks}

\bibliography{obsidian}

\appendix
\section{Appendix}
\subsection{Main Soundness Theorems}

\begin{theorem}[Progress]
If e is a closed expression and $\ty{\typeBounds}{\Delta}{s}{e}{T}{\Delta'}$, then at least one of the following holds:
\begin{enumerate}
\item e \text{is a value}
\item For any environment $\Sigma$ such that $\ok{\typeBounds, \Sigma, \Delta}$, $\stepsTo{\Sigma, e}{\Sigma', e'}$ for some environment $\Sigma'$
\item $e$ is stuck at a bad state transition --- that is, $e = \mathbb{E}[l \nearrow_{Shared} S(\overline{s})]$ where $\mu(\rho(l)) = \generics{C}{T'}.S'(\ldots)$, $S \neq S'$, $\rho(l) \in \phi$, and $\ty{\typeBounds}{\Delta}{s}{l}{\generics{C}{T'}.Shared}{\Delta'}$.
\item $e$ is stuck at a reentrant invocation -- that is, $e = \mathbb{E}[l.m(\overline{s})]$ where $\mu(\rho(l)) = \generics{C}{T'}.S(\ldots)$, $\rho(l) \in \psi$.
\item $e$ is stuck in a nested dynamic state check -- that is, $e = \mathbb{E}[\ifExpr{s}{shared}{T\textsubscript{ST}}{e_1}{e_2}]$ where  $\mu(\rho(l)) = \generics{C}{T}.S(\ldots)$ and $\rho(l) \in \phi$.

\end{enumerate}
\end{theorem}

\begin{proof} By induction on the derivation of \ty{\typeBounds}{\Delta}{s}{e}{T}{\Delta'}.
\begin{description}
\item[Case: T-lookup.]
	$e = b$. We case-analyze on $b$.
	\begin{description}
		\item [Subcase: $b = x$]. Then $b$ is not closed. Contradiction.
        \item [Subcase: $b = l$]. Suppose $\ok{\typeBounds, \Sigma, \Delta}$. By global consistency, $l \in dom(\Sigma_\rho)$. Then \stepsTo{b}{\Sigma_\rho(l)} by rule E-lookup.
		\item [Subcase: $b = o$]. Then $b$ is a value.
	\end{description}
\item[Case: T-let.]
	Because $e$ is closed, $e = \text{let } x: T = e_1 \text{ in } e_2$. Otherwise, since $e$ is closed, $e_1$ is closed, and the induction hypothesis applies to $e_1$. This leaves several cases:
		\begin{description}
			\item[Case: $e_1$ is a value $v$] The properties of the context permit creating a fresh indirect reference $l$ that is not in $\rho$. By E-let, $\stepsTo{\Sigma, \text{let } x: T = v \text{ in } e}{\envUpdate{\rho[l \mapsto v]}{\rho}{\Sigma}, [l/x]e}$.
			\item[Case: $\stepsTo{\Sigma, e_1}{\Sigma', e_1'}$.] Then E-letCongr applies, and $\stepsTo{\Sigma, e}{\Sigma', \text{let } x: T = e_1' \text{ in } e_2}$.
			\item[Case: $e_1$ is stuck] with $e_1 = \mathbb{E}[l \nearrow_{Shared} S(\overline{s})]$.
			 Then
			 \begin{align*}
				e &= \text{let } x: T = \mathbb{E}[l \nearrow_{Shared} S(\overline{s})] \text{ in } e_2\\
				e &= \mathbb{E'}[l \nearrow_{Shared} S(\overline{s})]
  			 \end{align*}
			 \item[Case: $e_1$ is stuck] with $e_1 = \mathbb{E}[\mathbb{E}[l.m(\overline{s})]$.
			 Then
			 \begin{align*}
				e &= \text{let } x: T = \mathbb{E}[\mathbb{E}[l.m(\overline{s})] \text{ in } e_2\\
				e &= \mathbb{E'}[\mathbb{E}[l.m(\overline{s})]
  			 \end{align*}
		\end{description}
		
\item[Case: New-state.]
    Because $e$ is closed, $e = \text{new } \generics{C}{T}.S(\overline{l})$ (any variables $x$ would be free, so all parameters must be locations). The properties of the context permit creating a fresh object reference $o$ that is not in $\mu$. $\overline{l}$ are a free locations of $e$, so by memory consistency (\ref{memory-consistency}), $\overline{l} \in dom(\rho)$, and $\overline{\rho(l)}$ is well-defined.
    By E-new, $\stepsTo{\Sigma, \text{new } \generics{C}{T}.S(\overline{l})}{\envUpdate{\mu[o \mapsto \generics{C}{T}.S(\overline{\rho(l)})]}{\mu}{\Sigma}, o}$.

\item[Case: T-this-field-def.]
	Because $e$ is closed, $e = l.f_i$. By assumption:
	\begin{enumerate}
		\item \ty{\typeBounds}{\Delta}{l}{l.f_i}{T}{\Delta'}.
		\item $\ok{\typeBounds, \Sigma, \Delta}$
	\end{enumerate}
    By memory consistency, $\rho(l) = o$ for some $o$ and $\mu(o) = \generics{C}{T'}.S(\overline{s'})$. Note that $1 \leq i \leq |\overline{s'}|$ by well-typedness of $s.f_i$ and global consistency. By rule E-field, $\stepsTo{\Sigma, s.f_i}{\Sigma, s'_i}$.

\item[Case: T-this-field-ctxt.]
	Identical to the \textit{This-field-def} case.

\item[Case: T-field-update.]
    Because $e$ is closed, $e = l.f_i := l'$.
    By memory consistency, $\mu(\rho(l)) = \generics{C}{T'}.S(\overline{l''})$.
    $fields(\generics{C}{T'}.S)$ is ambiently available.
    By E-fieldUpdate, \stepsTo{\Sigma, l.f_i := l'}{\envUpdate{\mu[\rho(l) \mapsto \generics{C}{T'}.S(l''_1, l''_2, \ldots, l''_{i-1}, l', l''_{i+1}, \ldots, l''_{|l''|})]}{\mu}{\Sigma}, \unit}.
    
\item[Case: T-inv.]
    Because $e$ is closed, $e = l_1.\generics{m}{T}(\overline{l_2})$.
    By memory consistency, $\mu(\rho(l)) = \generics{C}{T'}.S(\ldots)$.
    The transaction is ambiently available.
    We generate fresh $l_1'$ and $l_2'$ so that they are not in $dom(\rho)$.
    If $rho(l_1) \in \psi$, then e is stuck at a reentrant invocation.
    Otherwise, let
    \begin{enumerate}
        \item $\Sigma' = \Sigma[l_1' \mapsto \rho(l_1)][\overline{l_2' \mapsto \rho(l_2)}]$
        \item $\permVarMap' = \overline{PermVar(T_D) \mapsto Perm(T)}, \overline{PermVar(T_G) \mapsto Perm(T_M)}$
        \item $\Sigma'' = \envUpdate{\permVarMap'}{\permVarMap}{\envUpdate{\psi, \rho(l_1)}{\psi}{\Sigma'}}$
        \item $e' = tdef(C, m)$
    \end{enumerate}

    Then by E-Inv,
    	$\stepsTo{\Sigma, e}{\Sigma'', \fbox{$\overline{[l_1'/x]}[l_2'/\this]e'$}^{\rho(l_1)}}$

\item[Case: T-privInv.]
	Analogous to the \textit{Public-Invoke} case, using rule E-Inv-Private, except that the invocation is never stuck (E-Inv-Private does not check that $rho(l_1) \nin \psi$).

\item[Case: T-$\nearrow_p$.]
    Because $e$ is closed, $e = l \nearrow_{p} S(\overline{l'})$.
    By assumption, $l : \generics{C}{T'}.T_{ST}$.
    By memory consistency, $l \in dom(\rho)$ and $\mu(\rho(l)) = \generics{C}{T'}.S(\ldots)$.
    We case-analyze on $T_{ST}$.

		\begin{description}
            \item [Subcase: $T_{ST} = \overline{S}$ or $T_{ST} = Owned$.] \mbox{} \\
                By E-$\nearrow_{owned}$,
                	\stepsTo{\Sigma, l \nearrow_{owned} S(\overline{l'})}{\envUpdate{\mu[\rho(l) \mapsto \generics{C}{T'}.S(\overline{l'})]}{\mu}{\Sigma}, \unit}
            \item [Subcase: $T_{ST} = Shared$.] \mbox{} 
            	\begin{description}
                	\item[Case: $\rho(l) \nin \phi$.] Then by E-$\nearrow_{shared}$, \\
						\stepsTo{\Sigma, l \nearrow_{shared} S(\overline{l'})}{\envUpdate{\mu[\rho(l) \mapsto \generics{C}{T'}.S(\overline{l'})]}{\mu}{\Sigma}, \unit}.
                	\item[Case: $\mu(\rho(l)) = \generics{C}{T'}.S(\ldots)$.] Then by E-$\nearrow_{shared}$,\\
	                	\stepsTo{\Sigma, l \nearrow_{shared} S(\overline{l'})}{\envUpdate{\mu[\rho(l) \mapsto \generics{C}{T'}.S(\overline{l'})]}{\mu}{\Sigma}, \unit}.
                	\item[Case: $e$ is stuck at a bad state transition.]
                In that case, we have \\
                 $e = \mathbb{E}[l \nearrow_{\Shared} S(\overline{l'})]$ where $\mu(\rho(l)) = \generics{C^*}{T^*}.S'(\ldots)$, $S \neq S'$, $\rho(l) \in \phi$, and $\ty{\typeBounds}{\Delta}{s}{l}{\generics{C}{T'}.\Shared}{\Delta'}$. $C = C^*$ due to memory consistency.
                \end{description}
            \item [Subcase: $T_{ST} = \Unowned$.] This case is impossible because it contradicts the antecedent $T_{ST} \neq \Unowned$ of \textit{T-$\nearrow_p$}.
		\end{description}

\item[Case: T-assertStates.]
    Because $e$ is closed, $e = \assertExpr{l}{\overline{S}} $. By rule E-assert, \stepsTo{\Sigma, \text{assert } l \text{ in } \overline{S}}{\Sigma, \unit}.

\item[Case: T-assertPermission.]
    Because $e$ is closed, $e = \assertExpr{l}{T_{ST}}$. By rule E-assert, \stepsTo{\Sigma, \text{assert } l \text{ in } T_{ST}}{\Sigma, \unit}.

\item[Case: T-assertInVar.]
    Because $e$ is closed, $e = \assertExpr{l}{T_{ST}}$. By rule E-assert, \stepsTo{\Sigma, \text{assert } l \text{ in } T_{ST}}{\Sigma, \unit}.

\item[Case: T-assertInVarAlready.]
    Because $e$ is closed, $e = \assertExpr{l}{T_{ST}}$. By rule E-assert, \stepsTo{\Sigma, \text{assert } l \text{ in } T_{ST}}{\Sigma, \unit}.

\item[Case: T-isInStaticOwnership.]
    Because $e$ is closed, $e = \ifExpr{l}{owned}{S}{e_1}{e_2}$. By memory consistency, there exists S' such that $\mu(\rho(l)) = \generics{C}{T'}.S'(\ldots)$.
	\begin{description}
		\item [Subcase: $S' = S$.] Then by E-IsIn-Dynamic-Match-Owned, \stepsTo{\Sigma, e}{\Sigma, e_1 }.
		\item [Subcase: $S' \neq S$.] By IsIn-Dynamic-Else, \stepsTo{\Sigma, e}{\Sigma, e_2}.
	\end{description}

\item[Case: T-isInDynamic.]
    Because $e$ is closed, $e = \ifExpr{l}{shared}{S}{e_1}{e_2}$.
    By memory consistency, $l \in dom(\rho)$ and there exists S' such that $\mu(\rho(l)) = \generics{C}{T'}.S'(\ldots)$.

    By inversion, we have $l : \generics{C}{T'}.\Shared$.

	\begin{description}
		\item [Subcase: $S' = S$.] Then if $\rho(l) \in \phi$ then we are stuck in a nested dynamic state check. Otherwise, by E-IsIn-Dynamic-Match-Shared, \stepsTo{\Sigma, e}{{\envUpdate{\phi, \rho(l)}{\phi}{\Sigma}, \fbox{$e_1$}_{\rho(l)} }}.
		\item [Subcase: $S' \neq S$.] By IsIn-Dynamic-Else, \stepsTo{\Sigma, e}{\Sigma, e_2}.
	\end{description}

\item[Case: T-IsIn-PermVar.]
    Because $e$ is closed, $e = \ifExpr{l}{P}{p}{e_1}{e_2}$.
    By assumption $\ok{\typeBounds, \Sigma, \Delta}$, so $\permVarMap(p) = T_{ST}$ for some $T_{ST}$.
    Then $\stepsTo{\Sigma, \ifExpr{l}{P}{p}{e_1}{e_2}}{\Sigma, \ifExpr{l}{P}{T_{ST}}{e_1}{e_2}}$.

\item[Case: T-IsIn-Perm-Then.]
    Because $e$ is closed, $e = \ifExpr{l}{p}{\text{Perm}}{e_1}{e_2}$.
    By inversion, \subperm{\typeBounds}{P}{\text{Perm}}.
    As both $P$ and $\text{Perm}$ are permissions, not variables, we have $\subperm{\cdot}{P}{\text{Perm}}$, so by E-IsIn-Permission-Else \stepsTo{\Sigma, e}{\Sigma, e_1}.

\item[Case: T-IsIn-Perm-Else.]
    Because $e$ is closed, $e = \ifExpr{l}{p}{\text{Perm}}{e_1}{e_2}$.
    By inversion, \subperm{\typeBounds}{\text{Perm}}{P}, and $P \neq \text{Perm}$.
    As both $P$ and $\text{Perm}$ are permissions, not variables, we have $\subperm{\cdot}{\text{Perm}}{P}$, so by E-IsIn-Permission-Else \stepsTo{\Sigma, e}{\Sigma, e_2}.

\item[Case: T-IsIn-Unowned.]
    Because $e$ is closed, $e = \ifExpr{l}{p}{\text{Perm}}{e_1}{e_2}$.
    In this case, by E-IsIn-Unowned $\stepsTo{\Sigma, e}{e_2}$.

\item[Case: T-disown.]
	Because $e$ is closed, $e = \disown \ l$.
    By rule \textit{disown}, \stepsTo{\Sigma, \disown l}{\Sigma, \unit}.

\item[Case: T-pack.]
	By \textit{pack}, \stepsTo{\Sigma, \pack}{\Sigma, \unit}.

\item[Case: T-state-mutation-detection.]
	Because $e$ is closed, $e = \phibox{e'}{o}$, where $e'$ is also closed.
    If $e'$ is a value $v$, then by E-Box-$\phi$, \stepsTo{\Sigma, \phibox{$v$}{o}}{\envUpdate{(\phi \setminus o)}{\phi}{\Sigma}, v}.
    Otherwise, by the induction hypothesis, either \stepsTo{\Sigma, e'}{\Sigma', e''}, or $e'$ is stuck with an appropriate evaluation context. In the former case, by E-box-$\phi$-congr, \stepsTo{\Sigma, \fbox{$e'$}_o}{\Sigma', \fbox{$e''$}_o}. In the latter case, $e$ is stuck with an appropriate evaluation context.

\item[Case: T-reentrancy-detection.]
	Because $e$ is closed, $e = \psibox{e'}{o}$, where $e'$ is also closed.
    If $e'$ is a value $v$, then by E-Box-$\psi$, \stepsTo{\Sigma, \psibox{$v$}{o}}{\envUpdate{(\psi \setminus o)}{\psi}{\Sigma}, v}.
    Otherwise, by the induction hypothesis, either \stepsTo{\Sigma, e'}{\Sigma', e''}, or $e'$ is stuck with an appropriate evaluation context. In the former case, by E-box-$\psi$-congr, \stepsTo{\Sigma, \psibox{$e'$}{o}}{\Sigma', \psibox{$e''$}{o}}. In the latter case, $e$ is stuck with an appropriate evaluation context.

\end{description}
\end{proof}

\begin{theorem}[Preservation]
If e is a closed expression, $\ty{\typeBounds}{\Delta}{s}{e}{T}{\Delta''}$, \ok{\typeBounds, \Sigma, \Delta}, and \stepsTo{\Sigma, e}{\Sigma', e'} then for some $\Delta'$, $\ty{\typeBounds'}{\Delta'}{s}{e'}{T'}{\Delta'''}$, \ok{\typeBounds', \Sigma', \Delta'}, and $\lStronger{\Delta'''}{\typeBounds, \Sigma'}{\Delta''}$.
\end{theorem}

\begin{proof}
Proof proceeds by induction on the dynamic semantics.

\begin{description}

\item[Case: E-lookup.] $e = l$.
	We case-analyze on $T$.
	\begin{description}
		\item[Subcase: $T = \Unit$]
		$ $ \\
		By assumption, \ok{\typeBounds, \Sigma, \Delta}, and \ty{\typeBounds}{\Delta}{s}{l}{T}{\Delta''}. By assumption and E-Lookup, \stepsTo{\Sigma, l}{\Sigma, \Sigma_\rho(l)}. The fact that $\Sigma_\rho(l) = \unit$ follows directly from global consistency. Then by T-(), \ty{\typeBounds}{\Delta}{s}{\unit}{\Unit}{\Delta}. Global consistency is immediate because the contexts are unchanged, and $\lStronger{\Delta'''}{\typeBounds, \Sigma'}{\Delta''}$ by $<^l$-reflexivity.

        \item[Subcase: $T = \generics{C}{T'}.T_{ST}$]
		$ $ \\
            By inversion, $\Delta = \Delta_0, l: T_0$, $\splitType{T_0}{T}{T_2}$, and $\Delta'' = \Delta_0, l: T_2$. By rule E-lookup, $e' = \Sigma_\rho(l)$. The fact that $\Sigma_\rho(l) = o$ for some $o$ follows directly from global consistency. $l \neq o$ by construction and if $o$ occurs in $\Delta_0$, then we apply the strengthening lemma to generate a new proof of \ty{\typeBounds}{\Delta}{s}{e}{T}{\Delta'} in which $o$ does not occur.
            Thus, $\Delta' = \Delta'', o: T_2$ is a valid typing context. Then by Var, $\ty{\typeBounds}{\Delta'', o: T\textsubscript{2}}{s}{o}{T\textsubscript{2}}{\Delta'', o: T'}$ for some $T'$.
		Now, $\Delta'$ is the same as $\Delta$ except that some instances of $T_0$ have been replaced with $T_2$. The required consistency is obtained from the Split Compatibility lemma. We have $\lStronger{\Delta'''}{\typeBounds, \Sigma'}{\Delta''}$ because the two contexts differ only on $o$, which is not relevant to the $<^l$ relation.

        \item[Subcase: $T = \generics{I}{T'}.T_{ST}$ or $T = X.T_{ST}$]
        $ $ \\
            By memory consistency, this case is impossible.
	\end{description}

\item[Case: E-new.] $e = \new \ \generics{C}{T'}.S(\overline{l})$ because $e$ is closed (any variables would be free, so they must not exist).
    By assumption, \ty{\typeBounds}{\Delta}{s}{\new \ \generics{C}{T'}.S(\overline{l})}{\generics{C}{T'}.S}{\Delta''}; also, $e' = o$, and $o \nin dom(\mu)$. Let $\Delta' = \Delta'', o : \generics{C}{T'}.S$.
        By \textit{Var}, \ty{\typeBounds}{\Delta'', o: \generics{C}{T'}.S}{s}{o}{\generics{C}{T'}.S}{\Delta'', o: \generics{C}{T'}.Unowned}.
        Since $o$ is fresh and $\ok{\typeBounds, \Sigma, \Delta}$, there are no references to $o$ in the previous contexts, so all of the aliases are trivially consistent.
        We also have $\overline{\subtype{\typeBounds}{T}{stateFields(\generics{C}{T'}, S)}}$, where $\overline{l : T} \in \Delta$, which implies the required field property for reference consistency.
        By the split compatibility lemma, we have \ok{\typeBounds, \Sigma, \Delta'}. We have $\lStronger{\Delta'''}{\typeBounds, \Sigma'}{\Delta''}$ because the two contexts differ only on $o$, which is not relevant to the $<^l$ relation.

\item[Case: E-let.] $e$ = let $x: T_1$ = $v$ in $e_2$
	    By assumption:
	    	\begin{enumerate}
			 	\item \stepsTo{\Sigma, \text{let } x: T_1 = v \text{ in } e_2}{\envUpdate{\rho[l \mapsto v]}{\rho}{\Sigma}, [l/x]e}
				\item \ok{\typeBounds, \Sigma, \Delta}
				\item \ty{\typeBounds}{\Delta}{s}{\text{let } x: T_1 = v \text{ in } e_2}{T}{\Delta''}
			\end{enumerate}
	    \begin{description}
	    	\item[Subcase: v = o.]\mbox{}
				\begin{enumerate}
				\item By inversion:
        	    	\begin{enumerate}
        				\item \label {let-o-ty} \ty{\typeBounds}{\Delta}{s}{o}{T_1}{\Delta^*}
        				\item \label{let-e2-ty} \ty{\typeBounds}{\Delta^*, x: T_1}{s}{e_2}{T}{\Delta^{**}, x: T_1'}
                        \item $\disposable{\typeBounds}{T_1'}$
        				\item $l \notin dom(\rho)$
        			\end{enumerate}

        	    \item Let $\Delta' = \Delta^*, l : T_1$.
                By the substitution lemma (\ref{substitution}) applied to \ref{let-e2-ty},\\
                	 \ty{\typeBounds}{\Delta'}{s}{[l/x]e_2}{T}{\Delta^{**}, l: T_1'}.
	 
                \item By global consistency and \ref{let-o-ty}, $T_1$ is consistent with all other references in $refTypes(\Sigma, \Delta, o)$.

                Now, note that by global consistency, all references were previously compatible with $T_1$. $\Sigma'$ now includes a reference to the same object with indirect reference l, which corresponds with $l : T_1 \in \Delta'$. The only rule that could have been used in \ref{let-o-ty} is T-lookup, which split \splitType{T_1}{T_1'}{T_3} and replaced $o: T_1 \in \Delta$ with $o: T_3 \in \Delta'$. By the split compatibility lemma (\ref{split-compatibility}), $T_3$ is compatible with all other aliases to $o$, and in particular with $T_1'$.
                \item $\lStronger{\Delta^{**}, l: T_1'}{\typeBounds, \Sigma'}{\Delta^{**}, x: T_1'}$ because $l \nin dom(\Delta^{**}, x: T_1')$.
                \end{enumerate}
			\item[Subcase: v = \unit.]
				By inversion:
        	    \begin{enumerate}
	    			\item \ty{\typeBounds}{\Delta}{s}{\unit}{\Unit}{\Delta}
					\item \ty{\typeBounds}{\Delta, x : \Unit}{s}{e_2}{T}{\Delta^*, x : T_1'}
                    \item $\disposable{\typeBounds}{\Unit}$
					\item $l \notin dom(\rho)$
				\end{enumerate}

				Let $\Delta' = \Delta^*, l : \Unit$. By the substitution lemma (\ref{substitution}) \ty{\typeBounds}{\Delta^*, l : \Unit}{s}{[l/x]e_2}{T}{\Delta^{**}, l: T_1'}. Then the extensions to the contexts do not affect permissions, so they must be compatible, and \ok{\typeBounds, \Sigma', \Delta'}. $\lStronger{\Delta^{**}, l: T_1'}{\typeBounds, \Sigma'}{\Delta^{**}, x: T_1'}$ because $l \nin dom(\Delta^{**}, x: T_1')$.
		\end{description}
\item[Case: E-letCongr.] $e$ = let $x: T_1$ = $e_1$ in $e_2$. 
			\begin{enumerate}
    			\item By assumption:
    	    	\begin{enumerate}
    				\item \ok{\typeBounds, \Sigma, \Delta}
    				\item \ty{\typeBounds}{\Delta}{s}{\text{let } x: T_1 = e_1 \text{ in } e_2}{T}{\Delta''}
    			\end{enumerate}
			\item By inversion:
    			\begin{enumerate}
     				\item \stepsTo{\Sigma, e_1}{\Sigma^*, e_1'}.
    				\item \ty{\typeBounds}{\Delta}{s}{e_1}{T_1}{\Delta^*}
                    \item \label{secongr-body-check} \ty{\typeBounds}{\Delta^*, x: T_1}{s}{e_2}{T_2}{ \Delta^{**}, x: T_1'}
                    \item \label{congr-disposable} $\disposable{\typeBounds}{T_1'}$
    			\end{enumerate}

            \item By the induction hypothesis:
                \begin{enumerate}
                    \item \label{congr-e1'-ty} \ty{\typeBounds'^*}{\Delta'^{*}}{s}{e_1'}{T_1}{\Delta''^{*}} for some $\typeBounds'^*, \Delta'^*$, and $\Delta''^*$
                    \item \ok{\typeBounds'^*, \Sigma^*, \Delta'^*}
                    \item \label{secongr-new-context-stronger} $\lStronger{\Delta''^{*}}{\typeBounds, \Sigma'}{\Delta^*}$
                \end{enumerate}
            \item \label{congr-new-e2-ty} By \ref{l-stronger-substitution} with \ref{secongr-new-context-stronger} and \ref{secongr-body-check},  we have \ty{\typeBounds}{\Delta''^{*}, x: T_1}{s}{e_2}{T_2}{\Delta^{***}, x: T_1'}, with \lStronger{\Delta^{***}}{\Sigma^*}{\Delta^{**}}.
                
            \item Let $\Delta' = \Delta'^{*}$ and let $\typeBounds^{**} = \typeBounds, \typeBounds'^*$. 
                
                Then, by rule Let with \ref{congr-e1'-ty}, \ref{congr-new-e2-ty}, and \ref{congr-disposable}, \ty{\typeBounds{**}}{\Delta'}{s}{\text{let } x: T_1 = e_1' \text{ in } e_2}{T}{\Delta^{***}}, where \lStronger{\Delta^{***}}{\Sigma'}{\Delta''}.
                \item By \ref{l-stronger-consistency}, \ok{\Gamma^{**}, \Sigma', \Delta'}.
			\end{enumerate}

\item[Case: E-Inv.] $e = l_1.\generics{m}{M}(\overline{l_2})$ because $e$ is closed.
			\begin{enumerate}
    			\item By assumption, and because $e$ is closed:
    	    	\begin{enumerate}
                    \item \label{invoke-step} \stepsTo{\Sigma, l_1.\generics{m}{M}(\overline{l_2})}{\envUpdate{\psi, \rho(l_1)}{\psi}{\Sigma''}, \fbox{$\overline{[l_2'/x]}[l_1'/\this]e$}^{\rho(l_1)}}
    				\item \label{invoke-orig-ok} \ok{\typeBounds, \Sigma, \Delta}
                    \item \label{invoke-orig-ty} \ty{\typeBounds}{\Delta_0, l_1 : \generics{C}{T}.T_{STl1}, \overline{l_2: T_{l2}}}{s}{l_1.\generics{m}{M}(\overline{l_2})}{T}{\Delta_0, l_1: T_{l1}', \overline{l_2: T_{l2}'}}
    			\end{enumerate}
    			\item By inversion:
    			\begin{enumerate}
					\item $l_1' \nin dom(\rho)$
					\item $\overline{l_2' \nin dom(\rho)}$
                    \item $\params{C} = \overline{T_D}$
					\item $\Sigma'' = \Sigma[l_1' \mapsto \rho(l_1)][\overline{l_2' \mapsto \rho(l_2)}]$
                    \item $\permVarMap' = \permVarMap, \overline{PermVar(T_D) \mapsto Perm(T)}, \overline{PermVar(T_M) \mapsto Perm(M)}$
                    \item $\Sigma''' = \envUpdate{\permVarMap'}{\permVarMap}{\envUpdate{\psi, \rho(l_1)}{\psi}{\Sigma'}}$
                    \item $\mu(\rho(l_1)) = \generics{C}{T}.S(\ldots)$
    				\item $\rho(l_1) \notin \psi$
                    \item $tdef(C,m) = \generics{m}{T_M}(\overline{T_x \trans T_{xST} \ x}) \ T_{this} \trans \ T_{this}' \ e'$
                    \item \label{subt-l1} $\subtype{\typeBounds}{\generics{C}{T}.T_{STl1}}{\generics{C}{T}.T_{this}}$
                    \item \label{subt-l2} $\overline{\subtype{\typeBounds}{T_{l2}}{C_x.T_x}}$
                    \item $T_{l1}' = \funcArg{\generics{C}{T}.T_{STl1}}{\generics{C}{T}.T_{this}}{\generics{C}{T}.T_{this}'}$
					\item $\overline{T_{l2}' = \funcArg{T_{l2}}{T_x}{C_x.T_{xST}}}$
    			\end{enumerate}

    			\item We assume that the transaction is well-typed in its contract: \\
			 {\okIn{T \ \generics{m}{M}(\overline{C_x.T_x \trans T_{xST} \ x}) T_{this} \trans \ T_{this}' \ e}{C}}. As a result, we additionally have (by inversion):
    			\begin{enumerate}
                    \item \label{eTypeInvoke} \ty{\overline{T_D}, \overline{T_G}}{this: \generics{C}{T}.T_{this}, {\ \overline{x: C_x.T_x}}}{s_1}{e}{T}{this: \generics{C}{T}.T_{this}', \overline{x: C_x.T_{xST}}}
    			\end{enumerate}

                    Then by the substitution lemma for interfaces (\ref{lem:subcon-subs}), we also have
                \begin{enumerate}
                    \item \ty{\overline{T_D}, \overline{T_G}}{this: \generics{C}{T}.T_{this}, {\ \overline{x: \generics{C'}{T'}.T_x}}}{s_1}{e}{T}{this: \generics{C}{T}.T_{this}', \overline{x: \generics{C'}{T'}.T_{xST}}}
                \end{enumerate}

                where $\overline{l_2 : \generics{C'}{T'}.T_{ST}'}$, by global consistency.

			\item \label{apply-subst-invoke}
                Let $\typeBounds' = \typeBounds, \overline{T_D}, \overline{T_M}$.
                By the substitution lemma (\ref{substitution}) on \ref{eTypeInvoke}, we have: \\
                    $\ty{\typeBounds'}{l_1': \generics{C}{T}.T_{this}, {\ \overline{l_2': \generics{C'}{T'}.T_{x}}}}{s_1}{[l_2'/x][l_1'/\this]e}{T}{l_1': \generics{C}{T}.T_{this}', \overline{l_2': \generics{C'}{T'}.T_{xST}}}$

			\item Let:
			\begin{align*}
				T_{l1R} &= \funcArgResidual{\generics{C}{T}.T_{STl1}}{\generics{C}{T}.T_{this}}{\generics{C}{T}.T_{this}'}\\
				\overline{T_{l2R}} &= \overline{\funcArgResidual{T_{l2}}{T_x}{C_x.T_{xST}}}\\
                \Delta' &= \Delta, l_1: T_{l1R}, {\ \overline{l_2: T_{l2R}}}, l_1': \generics{C}{T}.T_{this}, \overline{l_2' : \generics{C'}{T'}.T_{l2}'}
			\end{align*}
			Note that $l_1$ and $l_2$ do not occur free in $[\overline{l_2'/x}][l_1'/\this]e$ because otherwise (\ref{eTypeInvoke}) would not have been the case. Then we have (by weakening \ref{apply-subst-invoke}):
                    $ \ty{\typeBounds'}{\Delta'}{s}{[\overline{l_2'/x}][l_1'/\this]e}{T}{\Delta, l_1: T_{l1R}, \overline{l_2: T_{l2R}}, l_1': \generics{C}{T}.T_{this}', \overline{l_2': \generics{C'}{T'}.T_{xST}}}$

			\item By rule Reentrancy-detection: \\
                $\ty{\typeBounds'}{\Delta'}{s}{\psibox{[\overline{l_2'/x}][l_1'/\this]e}{\rho(l)}}{T}{\Delta, l_1: T_{l1R}, \overline{l_2: T_{l2R}}, l_1': \generics{C}{T}.T_{this}', \overline{l_2': \generics{C'}{T'}.T_{xST}}}$
			  which corresponds to the evaluation step in \ref{invoke-step}.
                This also gives us that every indirect reference has a contract type, as required by global consistency.

          \item Consider:
          	\begin{align*}
			 	T_{l1R} &= \funcArgResidual{\generics{C}{T}.T_{STl1}}{\generics{C}{T}.T_{this}}{\generics{C}{T}.T_{this}'}\\
                T_{l1}' &= \funcArg{\generics{C}{T}.T_{STl1}}{\generics{C}{T}.T_{this}}{\generics{C}{T}.T_{this}'}
  			\end{align*}
				If $T_{l1}' \neq \generics{C}{T}.T_{this}'$, there are two possibilities, both with $\generics{C}{T}.T_{this} = \Unowned$. If $\generics{C}{T}.T_{STl1} = T_C.\Shared$, then $T_{l1R} = T_C.\Shared$; otherwise, $\maybeOwned{T_{l1R}}$. In both cases, \sameOwnership{T_{l1R}}{T'_{l1}} and \subtype{\typeBounds}{T_{l1R}}{T'_{l1}}. The same argument holds for $l_2$ and its type. Therefore:\\          
          
          \lStronger{\Delta, l_1: T_{l1R}, \overline{l_2: T_{l2R}}, l_1': \generics{C}{T}.T_{this}', \overline{l_2': \generics{C'}{T'}.T_{xST}}}{\typeBounds, \Sigma'}
				 {\Delta_0, l_1: T_{l1}', \overline{l_2: T_{l2}'}}

			\item 	
			
                By assumption of $\ok{\typeBounds, \Sigma, \Delta}$, $\permVarMap$ contains mappings for each $p \in \PermVar{\typeBounds}$.
                Note that $\permVarMap'$ additionally contains mappings for each $\overline{T_G}$ and $\overline{T_D}$, so $\PermVar{\typeBounds'} \subset \{ p \; | \; \permVarMap(p) = T_{ST} \}$, as required by global consistency.
                Finally, to show \ok{\typeBounds', \Sigma', \Delta'}, we need to show that the new types for $l_1$ and $\overline{l_2}$ are compatible with the aliases in $\Delta'$. 
                
                First consider $T_{l1R}$ and $\generics{C}{T}.T_{this}$, which alias the object originally referenced with type $\generics{C}{T}.T_{STl1}$. By assumption (\ref{invoke-orig-ty} and \ref{invoke-orig-ok}),  $\generics{C}{T}.T_{STl1}$ is compatible with all existing aliases in $\Sigma$.  Note that $T_{l1R} = \funcArgResidual{\generics{C}{T}.T_{STl1}}{\generics{C}{T}.T_{this}}{\generics{C}{T}.T_{this}'}$.
                
                Consider the cases for $T_{l1R}$:
			\begin{description}
                \item[Case: FuncArg-owned-unowned.] Previously, $l_1 : \generics{C}{T}.T_{STl1}$ was in $\Delta$, and \ok{\typeBounds', \Sigma, \Delta}. Now, $\Delta'$ includes both $\generics{C}{T}.T_{this}$ and $\generics{C}{T}.T_{STl1}$. But $T_{this} = \Unowned$, which is compatible with all other references.
                \item[Case: FuncArg-shared-unowned.] Previously, $l_1 : \generics{C}{T}.Shared$ was in $\Delta$, and \ok{\typeBounds', \Sigma, \Delta}. Now, $\Delta'$ includes both $\generics{C}{T}.T_{this}$ and $\generics{C}{T}.\Shared$. But $T_{this} = \Unowned$, which is compatible with \Shared.
                \item[Case: FuncArg-other.] Previously, $l_1: \generics{C}{T}.T_{STl1}$ was in $\Delta$, and \ok{\typeBounds', \Sigma, \Delta}. Now, $\Delta'$ includes both $\generics{C}{T}.T_{this}$ and $\generics{C}{T}.\Unowned$. But $\Unowned$ is compatible with all other references.
			\end{description}

			The corresponding argument applies to $l_2'$.
		\end{enumerate}

    \item[Case: E-Inv-Private.] $e = l_1.\generics{m}{M}(\overline{l_2})$ because $e$ is closed.

This case is similar to the E-Inv case, except that the fields are treated in a manner analogous to arguments: the field states are part of the initial context; they are transformed via $funcArg$; and the resulting types are in the output context.

\item[Case: E-IsIn-Dynamic-Match-Owned.] $e = \ifExpr{x}{owned}{T_{ST}}{e_1}{e_2}$ because $e$ is closed.
	 \begin{enumerate}
		 \item By assumption, and because $e$ is closed:
    	    \begin{enumerate}
		    	\item \ok{\typeBounds, \Sigma, \Delta}
                \item \label{isin-ty} \ty{\typeBounds}{\Delta_0, l: \generics{C}{T}.T_{ST}}{s}{\text{if } l \text{ is in\textsubscript{owned} } S \text{ then } e_1 \text{ else } e_2}{T_1}{\Delta''}
				\item \stepsTo{\Sigma, \text{if } l \text{ is in\textsubscript{owned} } S \text{ then } e_1 \text{ else } e_2}{\Sigma, e_1 }
			\end{enumerate}

        \item By inversion:
            \begin{enumerate}
                \item \label{l-type-dynamic-match} $\mu(\rho(l)) = \generics{C}{T}.S(\ldots)$
                \item \label{isin-inv-ty1} \ty{\typeBounds}{\Delta_0, l: \generics{C}{T}.S}{s}{e_1}{T_1}{\Delta^*}
                \item $S \in states(\generics{C}{T})$
                \item $\overline{S_x} = \possibleStates{\typeBounds}{\generics{C}{T}.T_{ST}}$
                \item $\subperm{\typeBounds}{T_{ST}}{Owned}$
                \item \label{isin-inv-ty2} \ty{\typeBounds}{\Delta_0, x: \generics{C}{T}.(\overline{S_x} \setminus S)}{s}{e_2}{T_1}{\Delta^{**}}
                \item $\Delta'' = merge(\Delta^*, \Delta^{**})$
            \end{enumerate}

        \item Let $\Delta' = \Delta_0, l: \generics{C}{T}.S$.
            By \ref{isin-inv-ty1}, \ty{\typeBounds}{\Delta'}{s}{e_1}{T_1}{\Delta^*}.
        \item The difference between $\Delta$ and $\Delta'$ is that in $\Delta'$, the type of $l$ is $\generics{C}{T}.S$. To show that \ok{\typeBounds, \Sigma', \Delta'}, we need to show that $\mu(\rho(l)) = \generics{C}{T}.S(\ldots)$. But this is given by (\ref{l-type-dynamic-match}).
        \item
        	By the merge subtyping lemma \ref{merge-subtyping}, if $l: T \in merge(\Delta^*, \Delta^{**})$, then $l: T' \in \Delta^*$ with \subtype{\typeBounds}{T'}{T} and \sameOwnership{T'}{T}. Thus, \lStronger{\Delta^*}{\typeBounds; \Sigma}{\Delta''}.
	\end{enumerate}

\item[Case: E-IsIn-Dynamic-Match-Shared.] $e = \text{if } l \text{ is in\textsubscript{shared} } \overline{S} \text{ then } e_1 \text{ else } e_2$
 \begin{enumerate}
		 \item By assumption, and because $e$ is closed:
    	    \begin{enumerate}
                \item \ok{\typeBounds, \Sigma, \Delta_0, l : \generics{C}{T}.Shared}
                \item \ty{\typeBounds}{\Delta_0, l: \generics{C}{T}.Shared}{s}{\text{if } l \text{ is in\textsubscript{shared} } S \text{ then } e_1 \text{ else } e_2}{T_1}{\Delta''}
				\item \stepsTo{\Sigma, \text{if } l \text{ is in\textsubscript{shared} } S \text{ then } e_1 \text{ else } e_2}{\envUpdate{\phi, \rho(l)}{\phi}{\Sigma}, \fbox{$e_1$}_{\rho(l)} }
			\end{enumerate}
		\item By inversion:
		\begin{enumerate}
            \item \label{dyn-match-shared-ty} \ty{\typeBounds}{\Delta_0, l: \generics{C}{T}.S}{s}{e_1}{T_1}{\Delta^*, l: \generics{C}{T}.T_{ST}}
            \item $\bound{\typeBounds}{T_{ST}} \neq \Unowned$
			\item $S \in \stateNames{C}$
            \item \ty{\typeBounds}{\Delta_0, l: \generics{C}{T}.Shared}{s}{e_2}{T_1}{\Delta^{**}, l: \generics{C}{T}.Shared}
            \item $\Delta'' = merge(\Delta^*, \Delta^{**}), l: \generics{C}{T}.Shared$
            \item \label{l-type-dynamic-match-shared} $\mu(\rho(l)) = \generics{C}{T}.S(\ldots)$
            \item $\rho(l) \notin \phi$
		\end{enumerate}
    \item Let $\Delta' = \Delta_0, l: \generics{C}{T}.S$.
        By State-mutation-detection and \ref{dyn-match-shared-ty}, \ty{\typeBounds}{\Delta'}{s}{\fbox{$e_1$}_{\rho(l)}}{T_1}{\Delta'''}.
     \item The difference between $\Delta$ and $\Delta'$ is that in $\Delta'$, the type of $l$ is $\generics{C}{T}.S$. By (\ref{l-type-dynamic-match-shared}), we know that $\mu(\rho(l)) = \generics{C}{T}.S(\ldots)$. However, there may be other aliases to $\rho(l)$ that have $\Shared$ permission. Since $\rho(l)$ is in the $\phi$ context of $\Sigma'$, any other references to $\rho(l)$ must be compatible with $\generics{C}{T}.Shared$, so we have \ok{\typeBounds, \Sigma', \Delta'} via $StateLockCompatible$.
     \item By the merge subtyping lemma \ref{merge-subtyping}, if $l: T \in merge(\Delta^*, \Delta^{**})$, then $l: T' \in \Delta^*$ with \subtype{\typeBounds}{T'}{T}. Thus, \lStronger{\Delta^*}{\typeBounds; \Sigma}{\Delta''}.

\end{enumerate}

\item[Case: E--IsIn-Dynamic-Else.] $e = \text{if } l \text{ is in}_p \  \overline{S} \text{ then } e_1 \text{ else } e_2$
\begin{enumerate}
	\item By assumption, and because $e$ is closed:
	\begin{enumerate}
		\item \ok{\typeBounds, \Sigma, \Delta}
		\item \stepsTo{\Sigma, \text{if } l \text{ is in}_p \  S \text{ then } e_1 \text{ else } e_2}{ \Sigma, e_2 }
		\item \ty{\typeBounds}{\Delta}{s}{\text{if } l \text{ is in\textsubscript{p} } S \text{ then } e_1 \text{ else } e_2}{T_1}{\Delta''}
	\end{enumerate}
	\item By inversion:
		\begin{enumerate}
			\item \label{dyn-else-in-s'} $\mu(\rho(l)) = \generics{C}{T}.S'(\ldots)$
			\item $S' \nin \overline{S}$
		\end{enumerate}
	\item By inversion, either:
            	\begin{enumerate}
            		\item \label{isin-else-ty-shared} \ty{\typeBounds}{\Delta_0, l: \generics{C}{T}.Shared}{s}{e_2}{T_1}{\Delta^*}; or:
					\item \label{isin-else-ty-owned} \ty{\typeBounds}{\Delta_0, l: \generics{C}{T}.\overline{S_x} \setminus \overline{S}}{s}{e_2}{T_1}{\Delta^{**}}
				\end{enumerate}
    \item If we are in case (\ref{isin-else-ty-shared}), let $\Delta' = \Delta$. Then by \ref{isin-else-ty-shared}, \ty{\typeBounds}{\Delta'}{s}{e_2}{T_1}{\Delta^*}. By assumption,  \ok{\typeBounds, \Sigma, \Delta'}.
    By the merge subtyping lemma \ref{merge-subtyping}, if $l: T \in merge(\Delta^*, \Delta^{**})$, then $l: T' \in \Delta^*$ with \subtype{\typeBounds}{T'}{T}. Thus, \lStronger{\Delta^*}{\typeBounds; \Sigma}{\Delta''}.
     \item Otherwise, let $\Delta' = \Delta_0, l: \generics{C}{T}.\overline{S_x} \setminus \overline{S}$. Then by \ref{isin-else-ty-owned}, \ty{\typeBounds}{\Delta'}{s}{e_2}{T_1}{\Delta^{**}}. By inversion, we had \ty{\typeBounds}{\Delta_0, l: \generics{C}{T}.T_{ST}}{s}{e_2}{T_1}{\Delta^{**}}. As a result, there are no other owning references to the object referenced by $l$, and the referenced object is in state $S'$ by (\ref{dyn-else-in-s'}). Since $S' \nin \overline{S}$, $\generics{C}{T}.\overline{S_x} \setminus \overline{S}$ is a consistent type for the reference, and \ok{\typeBounds, \Sigma, \Delta'}. By the merge subtyping lemma \ref{merge-subtyping}, if $l: T \in merge(\Delta^*, \Delta^{**})$, then $l: T' \in \Delta^{**}$ with \subtype{\typeBounds}{T'}{T}. Thus, \lStronger{\Delta^{**}}{\typeBounds; \Sigma}{\Delta''}.

\end{enumerate}

\item[Case: E-IsIn-PermVar] \mbox{}
	\begin{enumerate}
	\item By assumption, and because $e$ is closed:
    \begin{enumerate}
        \item \ok{\typeBounds, \Sigma, \Delta}
        \item \ty{\typeBounds}{\Delta}{s}{\text{if } l \text{ is in\textsubscript{Perm} } p \text{ then } e_1 \text{ else } e_2}{T_1}{\Delta''}
        \item $\stepsTo{\Sigma, \text{if } l \text{ is in\textsubscript{Perm} } p \text{ then } e_1 \text{ else } e_2}{\Sigma, \text{if } l \text{ is in\textsubscript{Perm} } T_{ST} \text{ then } e_1 \text{ else } e_2}$
    \end{enumerate}

    \item By inversion:
    \begin{enumerate}
        \item $\permVarMap(p) = T_{ST}$
        \item \label{isin-permvar-inv-then} $\ty{\typeBounds}{\Delta, l: T_C.p}{s}{e_1}{T_1}{\Delta'}$
        \item \label{isin-permvar-inv-else} $\ty{\typeBounds}{\Delta, l: T_C.T_{ST}'}{s}{e_2}{T_1}{\Delta''}$
        \item $\Delta_f = merge(\Delta', \Delta'')$
        \item $\text{Perm} = \ToPermission{T_{ST}'}$
    \end{enumerate}

    \item In order to perform substitution for type parameters, we must have proved $\subsOk{\typeBounds}{T}{T_G}$, so we must have $\subperm{\typeBounds}{T_{ST}}{p}$.
    Then by~\ref{isin-permvar-inv-then} and the permission variable substitution lemma~\ref{lem:permvar-subs}, we have $\ty{\typeBounds}{\Delta, l : T_C.T_{ST}}{s}{e_1}{T_1}{\Delta'}$.
    \item We proceed by case analysis on $T_{ST}$.

    \begin{description}
        \item[Subcase: $T_{ST} = \overline{S}$] ~

            If $P = \Unowned$, then $T_{ST}' = \Unowned$, and by~\ref{isin-permvar-inv-else} we can apply T-IsIn-Unowned to show \ty{\typeBounds}{\Delta, l: T_C.\Unowned}{s}{\text{if } l \text{ is in\textsubscript{Unowned} } \overline{S} \text{ then } e_1 \text{ else } e_2}{T_1}{\Delta_f'}.

            If $P = \Shared$, then $T_{ST}' = \Shared$, and by~\ref{isin-permvar-inv-else} we can apply T-IsIn-Dynamic to show \ty{\typeBounds}{\Delta, l: T_C.Shared}{s}{\text{if } l \text{ is in\textsubscript{shared} } \overline{S} \text{ then } e_1 \text{ else } e_2}{T_1}{\Delta_f'}.

            If $P = \Owned$, then $\subtype{\typeBounds}{T_{ST}'}{\Owned}$, so $\maybeOwned{T_C.T_{ST}'}$, and $\subperm{\typeBounds}{\overline{S_x}}{T_{ST}'}$, where $\overline{S_x} = \possibleStates{\typeBounds}{T_C.T_{ST}}$.
            Then by the subtype substitution lemma lemma~\ref{subtype-substitution} and by~\ref{isin-permvar-inv-else} we have $\ty{\typeBounds}{\Delta, l: T_C.\left( \overline{S_x} \setminus \overline{S} \right)}{s}{e_2}{T_1}{\Delta''}$.
            Now we can apply T-IsIn-StaticOwnership to get $\ty{\typeBounds}{\Delta, l: T_C.T_{ST}'}{s}{\text{if } l \text{ is in\textsubscript{owned} } \overline{S} \text{ then } e_1 \text{ else } e_2}{T_1}{\Delta_f'}$.

        \item[Subcase: $T_{ST} = P$] ~
        
            If $\subperm{\typeBounds}{\text{Perm}}{P}$, then by IsIn-Permission-Then, \\
            \ty{\typeBounds}{\Delta, l: T_C.T_{ST}}{s}{\text{if } l \text{ is in\textsubscript{\text{Perm}} } P \text{ then } e_1 \text{ else } e_2}{T_1}{\Delta_f'}.

            Otherwise, $\notsubperm{\typeBounds}{\text{Perm}}{P}$, so by~\ref{isin-permvar-inv-else} and IsIn-Permission-Else,
            \ty{\typeBounds}{\Delta, l: T_C.T_{ST}}{s}{\text{if } l \text{ is in\textsubscript{P} } \text{Perm} \text{ then } e_1 \text{ else } e_2}{T_1}{\Delta_f'}.

        \item[Subcase: $T_{ST} = q$] This case is impossible, because $\permVarMap$ only maps to nonvariable states and permissions.
    \end{description}

    \item In all cases, global consistency is maintained because the environment does not change, \lStronger{\Delta_f'}{\typeBounds, \Sigma'}{\Delta'} by reflexivity.
	\end{enumerate}
\item[Case: E-IsIn-Permission-Then]
    By assumption, and because $e$ is closed:
    \begin{enumerate}
        \item \ok{\typeBounds, \Sigma, \Delta}
        \item \ty{\typeBounds}{\Delta}{s}{\text{if } l \text{ is in\textsubscript{P} } \text{Perm} \text{ then } e_1 \text{ else } e_2}{T_1}{\Delta''}
        \item \stepsTo{\Sigma, \text{if } l \text{ is in\textsubscript{P} } \text{Perm} \text{ then } e_1 \text{ else } e_2}{\Sigma, e_1 }
    \end{enumerate}

    By inversion:
    \begin{enumerate}
        \item $\text{Perm} \in \{ \Owned, \Unowned, \Shared \}$
        \item \subperm{\cdot}{P}{\text{Perm}}
    \end{enumerate}

    To prove that $e$ is well-typed, we must have used either IsIn-Permission-Then or IsIn-Permission-Else.
    However, we know that $\subperm{\cdot}{P}{\text{Perm}}$, so we must have used IsIn-Permission-Else.
    Then by inversion of IsIn-Permission-Then, we have $\ty{\typeBounds}{\Delta_0, x: T_C.T_{ST}}{s}{e_1}{T_1}{\Delta'''}$.

    Let $\Delta' = \Delta_0, x: T_C.T_{ST}$.
    Global consistency is maintained because the environment has not changed, and $\lStronger{\Delta'''}{\typeBounds, \Sigma'}{\Delta''}$ by $<^l$-reflexivity.

\item[Case: E-IsIn-Permission-Else]
    By assumption, and because $e$ is closed:
    \begin{enumerate}
        \item \ok{\typeBounds, \Sigma, \Delta}
        \item \ty{\typeBounds}{\Delta}{s}{\text{if } l \text{ is in\textsubscript{P} } \text{Perm} \text{ then } e_1 \text{ else } e_2}{T_1}{\Delta''}
        \item \stepsTo{\Sigma, \text{if } l \text{ is in\textsubscript{P} } \text{Perm} \text{ then } e_1 \text{ else } e_2}{\Sigma, e_2 }
    \end{enumerate}

    By inversion:
    \begin{enumerate}
        \item $\text{Perm} \in \{ \Owned, \Unowned, \Shared \}$
        \item \subperm{\cdot}{\text{Perm}}{P}
        \item $P \neq \text{Perm}$
    \end{enumerate}

    To prove that $e$ is well-typed, we must have used either IsIn-Permission-Then or IsIn-Permission-Else.
    However, we know that $\subperm{\cdot}{\text{Perm}}{P}$ and $P \neq \text{Perm}$, so we must have used IsIn-Permission-Else.
    Then by inversion of IsIn-Permission-Else, we have $\ty{\typeBounds}{\Delta, x: T_C.T_{ST}}{s}{e_2}{T_1}{\Delta'}$.
    Global consistency is maintained because the environment has not changed, and $\lStronger{\Delta'''}{\typeBounds, \Sigma'}{\Delta''}$ by $<^l$-reflexivity.

\item[Case: E-IsIn-Unowned]
    By assumption, and because $e$ is closed:
    \begin{enumerate}
        \item \ok{\typeBounds, \Sigma, \Delta}
        \item \ty{\typeBounds}{\Delta}{s}{\text{if } l \text{ is in\textsubscript{Unowned} } \overline{S} \text{ then } e_1 \text{ else } e_2}{T_1}{\Delta''}
        \item \stepsTo{\Sigma, \text{if } l \text{ is in\textsubscript{Unowned} } \overline{S} \text{ then } e_1 \text{ else } e_2}{\Sigma, e_2 }
    \end{enumerate}

    By inversion:
    \begin{enumerate}
        \item \label{isin-unowned-inv-else} \ty{\typeBounds}{\Delta, x: T_C.T_{ST}}{s}{e_2}{T_1}{\Delta''}
    \end{enumerate}

    $e_2$ is well typed by~\ref{isin-unowned-inv-else}.
    Global consistency is maintained because the environment has not changed, and $\lStronger{\Delta'''}{\typeBounds, \Sigma'}{\Delta''}$ by $<^l$-reflexivity.

\item[Case: E-Box-$\phi$.] $e = \phibox{v}{o}$.
	\begin{enumerate}
    	\item By assumption, and because $e$ is closed:
    	\begin{enumerate}
    		\item \ok{\typeBounds, \Sigma, \Delta}
    		\item \ty{\typeBounds}{\Delta}{s}{\phibox{v}{o}}{T}{\Delta''}
    		\item \stepsTo{\Sigma, \phibox{$v$}{o}}{\envUpdate{(\phi \setminus o)}{\phi}{\Sigma}, v}
    	\end{enumerate}
    	\item By inversion:
    	\begin{enumerate}
    		\item \label{box-ty} \ty{\typeBounds}{\Delta}{s}{v}{T}{\Delta''}
    	\end{enumerate}
		\item Note that \phibox{e}{o} can only arise in the context of a shared-mode dynamic state test. Therefore, $\Delta$ must be of the form $\Delta_0, l: \generics{C}{T}.Shared$ and $\Delta''$ must be of the form $\Delta_0'', l: \generics{C}{T}.Shared$.
    	\item Since $v$ is a value, either $v = \unit$ or there exists $o'$ such that $v = o'$. If $v = \unit$, then let $\Delta' = \cdot$. By T-(), \ty{\typeBounds}{\cdot}{\unit}{\Unit}{\cdot}. Otherwise, $v = o'$ and by Var, there exists $o' : T_1 \in \Delta$ with $\splitType{T_1}{T_2}{T_3}$. In that case, let $\Delta' = \Delta, o' : T_1$. The proof proceeds as in the E-lookup rule: by Var, there exists $\Delta''' = o': T_3$  such \ty{\typeBounds}{\Delta'}{s}{o'}{T}{\Delta'''}.
		\item $\Delta'''$ differs from $\Delta''$ only on bindings for $o'$, which is not relevant to the $<^l$ relation, so $\lStronger{\Delta'''}{\typeBounds, \Sigma'}{\Delta''}$ by $<^l$-reflexivity. 
		\item \ok{\typeBounds, \Sigma, \Delta'} by the split compatibility lemma.
	\end{enumerate}

\item[Case: E-Box-$\phi$-congr.] $e = \phibox{e}{o}$.
	\begin{enumerate}
		\item By assumption, and because $e$ is closed:
    	\begin{enumerate}
			 \item \label{phi-ok} \ok{\typeBounds, \Sigma, \Delta}
			 \item \ty{\typeBounds}{\Delta}{s}{\phibox{e}{o}}{T}{\Delta''}
			 \item \stepsTo{\Sigma, \phibox{$e$}{o}}{\Sigma', \phibox{$e'$}{o}}
		\end{enumerate}
		\item By inversion:
		\begin{enumerate}
			\item \stepsTo{\Sigma, e}{\Sigma', e'}
			\item \ty{\typeBounds}{\Delta}{s}{e}{T}{\Delta''}
		\end{enumerate}
		\item Let $\Delta' = \Delta$. By \ref{phi-ok}, \ok{\typeBounds, \Sigma, \Delta'}. Note that $\Delta''' = \Delta''$. $\lStronger{\Delta'''}{\typeBounds, \Sigma'}{\Delta''}$ by $<^l$-reflexivity. By State-mutation-detection, \ty{\typeBounds}{\Delta'}{s}{\phibox{e'}{o}}{T}{\Delta''}.
	\end{enumerate}

\item[Case: E-Box-$\psi$.] $e = \psibox{v}{o}$.
	\begin{enumerate}
		\item By assumption, and because $e$ is closed:
    	\begin{enumerate}
			\item \label{psi-ok} \ok{\typeBounds, \Sigma, \Delta}
			\item \ty{\typeBounds}{\Delta}{s}{\psibox{v}{o}}{T}{\Delta''}
			\item \stepsTo{\Sigma, \psibox{$v$}{o}}{\envUpdate{(\psi \setminus o)}{\psi}{\Sigma}, v}
		\end{enumerate}
		\item By inversion:
		\begin{enumerate}
			\item \label{psi-ty} \ty{\typeBounds}{\Delta}{s}{v}{T}{\Delta''}
		\end{enumerate}
		\item Let $\Delta' = \Delta$. By \ref{psi-ty}, \ty{\typeBounds}{\Delta'}{s}{v}{T}{\Delta''}. $\Sigma' = \envUpdate{(\psi \setminus o)}{\psi}{\Sigma}$. Note that the definition of consistency does not depend on $\Sigma_\psi$. With \ref{psi-ok}, we conclude that \ok{\typeBounds, \Sigma', \Delta'}. Note that $\Delta''' = \Delta''$. $\lStronger{\Delta'''}{\typeBounds, \Sigma'}{\Delta''}$ by $<^l$-reflexivity.
	\end{enumerate}

\item[Case: E-Box-$\psi$-congr.] $e = \psibox{e}{o}$.
	\begin{enumerate}
		\item By assumption, and because $e$ is closed:
    	\begin{enumerate}
			\item \label{psi-ok2} \ok{\typeBounds, \Sigma, \Delta}
			\item \ty{\typeBounds}{\Delta}{s}{\psibox{e}{o}}{T}{\Delta''}
			\item \stepsTo{\Sigma, \psibox{$e$}{o}}{\Sigma', \psibox{$e'$}{o}}
		\end{enumerate}
		\item By inversion:
		\begin{enumerate}
			\item \label{psi-ty} \ty{\typeBounds}{\Delta}{s}{e}{T}{\Delta''}
			\item \stepsTo{\Sigma, e}{\Sigma', e'}
		\end{enumerate}
		\item Let $\Delta' = \Delta$. By \ref{psi-ok2}, \ok{\typeBounds, \Sigma, \Delta'}. Note that $\Delta''' = \Delta''$. $\lStronger{\Delta'''}{\typeBounds, \Sigma'}{\Delta''}$ by $<^l$-reflexivity. By Reentrancy-detection, \ty{\typeBounds}{\Delta'}{s}{\psibox{e'}{o}}{T}{\Delta''}.
	\end{enumerate}

\item[Case: E-State-Transition-Static-Ownership.] $e = l \nearrow_{owned} S(\overline{l'})$
	\begin{enumerate}
		\item By assumption, and because $e$ is closed:
    	\begin{enumerate}
			\item \ok{\typeBounds, \Sigma, \Delta}
			\item \stepsTo{\Sigma, l \nearrow_{owned} S(\overline{l'})}{\envUpdate{\mu[\rho(l) \mapsto \generics{C}{T}.S(\overline{\rho(l')})]}{\mu}{\Sigma}, \unit}
            \item \label{transition-static-ty} \ty{\typeBounds}{\Delta_0, l: \generics{C}{T}.T_{ST}}{l}{l \nearrow_{owned} S(\overline{x})}{\Unit}{\Delta^*, l : \generics{C}{T}.S}
		\end{enumerate}
		\item By inversion:
		\begin{enumerate}
			\item $\subperm{\typeBounds}{T_{ST}}{\Owned}$
    		\item \ty{\typeBounds}{\Delta_0}{l}{\overline{x}}{\overline{T}}{\Delta^{*}}
    		\item $\overline{\subtype{\typeBounds}{T}{type(stateFields(\generics{C}{T}, S'))}}$
			\item $unionFields(\generics{C}{T}, T_{ST}) = \overline{T_{fl} \ f_l}$
			\item $fieldTypes_l(\Delta^*; \overline{T_{fl} \ f_l}) = \overline{T_{fl}'}$
			\item $\overline{\disposable{\typeBounds}{T_{fl}'}}$
    	\end{enumerate}
    \item Let $\Delta' = \Delta, l : \generics{C}{T}.S$. By T-(), \ty{\typeBounds}{\Delta}{l}{\unit}{\Unit}{\Delta}. To show that \ok{\typeBounds, \Sigma', \Delta'}, it suffices to show that any $T \in refTypes(\Sigma', \Delta', \rho(l))$ that specifies state specifies type $\generics{C}{T}.S'$. But note that by \ref{transition-static-ty}, $l$ is in the original typing context with an owning type. Since \ok{\typeBounds, \Sigma, \Delta}, and $\generics{C}{T}.T_{ST} \in refTypes(\Sigma, \Delta, \rho(l))$, the only owning alias to the object referenced by $l$ is $l$ itself. Replacing $l : \generics{C}{T}.T_{ST}$ in $\Delta$ with $l: \generics{C}{T}.S$ replaces the type of the only owning alias with $\generics{C}{T}.S$, which is consistent with $\mu(\rho(l)) = \generics{C}{T}.S(\overline{l})$. $\lStronger{\Delta'''}{\typeBounds, \Sigma'}{\Delta''}$ by $<^l$-reflexivity.
	\end{enumerate}

\item[Case: E-State-Transition-Shared.] $e = l \nearrow_{shared} S(\overline{l'})$
\begin{enumerate}
\item By assumption, and because $e$ is closed:
    	\begin{enumerate}
			\item \ok{\typeBounds, \Sigma, \Delta}
            \item \stepsTo{\Sigma, l \nearrow_{shared} S(\overline{l'})}{\envUpdate{\mu[\rho(l) \mapsto \generics{C}{T}.S(\overline{\rho(l')})]}{\mu}{\Sigma}, \unit}
        \end{enumerate}
        \item \label{shared-typing-not-owned} Now, assume typing rule $\nearrow_{shared}$ applied, since if $\nearrow_{owned}$ applied, then the argument for case E-State-Transition-Static-Ownership (above) applies. Then:
        \begin{enumerate}
            \item \label{transition-static-ty-shared} \ty{\typeBounds}{\Delta_0, l: \generics{C}{T}.T_{ST}}{l}{l \nearrow_{shared} S(\overline{x})}{\Unit}{\Delta^*, l : \generics{C}{T}.S}
		\end{enumerate}
		\item By inversion:
		\begin{enumerate}
			\item $\subperm{\typeBounds}{T_{ST}}{\Shared}$. By \ref{shared-typing-not-owned}, we assume therefore $T_{ST} = \Shared$.
    		\item \ty{\typeBounds}{\Delta_0}{l}{\overline{x}}{\overline{T}}{\Delta^{*}}
    		\item $\overline{\subtype{\typeBounds}{T}{type(stateFields(\generics{C}{T}, S'))}}$
			\item $unionFields(\generics{C}{T}, T_{ST}) = \overline{T_{fl} \ f_l}$
			\item $fieldTypes_l(\Delta^*; \overline{T_{fl} \ f_l}) = \overline{T_{fl}'}$
			\item $\overline{\disposable{\typeBounds}{T_{fl}'}}$
            \item $	\rho(l) \notin \phi \lor \mu(\rho(l)) = \generics{C}{T}.S(\ldots)$
		\end{enumerate}
		\item There are two subcases.
			\begin{description}
                \item[Subcase: $\rho(l) \nin \phi$.] Let $\Delta' = \Delta$. By T-(), \ty{\typeBounds}{\Delta}{l}{\unit}{\Unit}{\Delta}. Now, all existing aliases to the object referenced by $\rho(l)$ were compatible with the previous reference, which was of type $\generics{C}{T}.Shared$. As a result, none of those references restricted the state of the object, and the new state (in $\Sigma'$) is consistent with $\Delta$.
                \item[Subcase: $\mu(\rho(l)) = \generics{C}{T}.S(\ldots)$.] Let $\Delta' = \Delta$. By T-(), \ty{\typeBounds}{\Delta}{l}{\unit}{\Unit}{\Delta}.  All references to the object referenced by $\rho(l)$ have the same type in $\Sigma'$ as they did in $\Sigma$ because neither the contract nor the state of the object have changed, and we have \ok{\typeBounds, \Sigma', \Delta'}.
			\end{description}
		\item In both cases, $\lStronger{\Delta'''}{\typeBounds, \Sigma'}{\Delta''}$ by $<^l$-reflexivity.
\end{enumerate}

\item[Case: E-Field.] $e = l.f_i$.
	\begin{enumerate}
		\item By assumption, and because $e$ is closed:
    	\begin{enumerate}
			\item \ok{\typeBounds, \Sigma, \Delta}
			\item \stepsTo{\Sigma, l.f_i}{\Sigma, o_i}
		\end{enumerate}
		\item By inversion:
		\begin{enumerate}
            \item $\mu(\rho(l)) = \generics{C}{T}.S(\overline{s})$
		\end{enumerate}
		\item Now, there are two subcases because there are two possible type judgments for $e$.
		\begin{description}
			\item [Subcase: this-field-def] \mbox{} 
				\begin{enumerate}
					\item By assumption: \ty{\typeBounds}{\Delta_0, l: T}{l}{l.f}{T_2}{\Delta_0, l: T, l.f: T_3}
					\item By inversion:
						\begin{enumerate}
							\item $l.f \notin Dom(\Delta)$
							\item $T_1 \ f \in intersectFields(T)$
							\item \splitType{T_1}{T_2}{T_3}
						\end{enumerate}
					\item Let $\Delta' = \Delta_0, l: T, l.f: T_3, o_i : T_2$. Then by Var, \ty{\typeBounds}{\Delta'}{s}{o_i}{T_2}{\Delta'''} for some $\Delta'''$. \ok{\typeBounds, \Sigma, \Delta'} because $T_2$ is a consistent permission for $o_i$ per the split compatibility lemma (as in the E-lookup case). $\Delta'''$ agrees with $\Delta''$ on all $l$, so $\lStronger{\Delta'''}{\typeBounds, \Sigma'}{\Delta''}$ by $<^l$-reflexivity.
				\end{enumerate}
			\item [Subcase: this-field-ctxt] \mbox{} 
				\begin{enumerate}
					\item By assumption: \ty{\typeBounds}{\Delta_0, l: T, l.f: T_1}{l}{l.f}{T_2}{\Delta_0, l: T, l.f: T_3}
					\item By inversion: \splitType{T_1}{T_2}{T_3}
				\end{enumerate}
					\item Let $\Delta' = \Delta_0, l: T, l.f: T_3, o_i : T_2$. Then by Var, \ty{\typeBounds}{\Delta'}{s}{o_i}{T_3}{\Delta'''} for some $\Delta'''$. \ok{\typeBounds, \Sigma, \Delta'} because $T_2$ is a consistent permission for $o_i$ per the split compatibility lemma. $\Delta'''$ agrees with $\Delta''$ on all $l$, so $\lStronger{\Delta'''}{\typeBounds, \Sigma'}{\Delta''}$ by $<^l$-reflexivity.
		\end{description}
	\end{enumerate}

\item[Case: E-FieldUpdate.] $e = l.f_i := l'$.
	\begin{enumerate}
		\item By assumption, and because $e$ is closed:
    	\begin{enumerate}
			\item \ok{\typeBounds, \Sigma, \Delta}
			\item \stepsTo{\Sigma, l.f_i := l'}{\envUpdate{\mu[\rho(l) \mapsto \generics{C}{T}.S(o_1, o_2, \ldots, o_{i-1}, \rho(l'), o_{i+1}, \ldots, o_{|\overline{l.f}|})]}{\mu}{\Sigma}, \text{\unit}}
			\item \ty{\typeBounds}{\Delta}{l}{l.f_i := l'}{\text{\Unit}}{\Delta^{**}, l.f_i: T_C.T_{ST}}
		\end{enumerate}
		\item By inversion:
		\begin{enumerate}
			\item $\mu(\rho(l)) = \generics{C}{T}.S(\overline{o})$
            \item $fields(\generics{C}{T}.S) = \overline{T \ f}$
   			\item \ty{\typeBounds}{\Delta}{l}{l.f_i}{T_C.T_{ST}}{\Delta^*}
		    \item \ty{\typeBounds}{\Delta^*}{l}{l.f_i}{T_C.T_{ST}'}{\Delta^{**}}
            \item \disposable{\typeBounds}{T_C.T_{ST}}
		\end{enumerate}
		\item Let $\Delta' = \Delta^*, l.f_i : T_C.T_{ST}$. By T-(), \ty{\typeBounds}{\Delta'}{l}{\unit}{\Unit}{\Delta'}. 
		\item Note that $\Sigma' = \envUpdate{\mu[\rho(l) \mapsto \generics{C}{T}.S(o_1, o_2, \ldots, o_{i-1}, \rho(l'), o_{i+1}, \ldots, o_{|l|})]}{\mu}{\Sigma}$. 
		By the same argument used in the proof of preservation for the \textit{E-lookup} case, \ok{\typeBounds, \Sigma, \Delta^*} and likewise \ok{\typeBounds, \Sigma, \Delta^{**}}. To show \ok{\typeBounds, \Sigma', \Delta'}, we note that the only change relative to $\Sigma$ and $\Delta^{**}$ is regarding the type of $l.f_i$. $\rho(l)$ has the same number of fields in $\Sigma'$ as in $\Sigma$. Although $\rho(l)$ may now have an additional reference to $\rho(l')$ that did not exist before, this reference is compatible with all of the other references in $refTypes(\Sigma', \Delta^*, \rho(l'))$ because if the new reference is owned, this is only because $T_C.T_{ST}$ was owned, which was previously accounted for in $refTypes(\Sigma', \Delta^*, \rho(l'))$, and that ownership has been removed in $\Delta^{**}$.
		\item $\Delta'''$ agrees with $\Delta''$ on all $l$, so $\lStronger{\Delta'''}{\typeBounds, \Sigma'}{\Delta''}$ by $<^l$-reflexivity.
	\end{enumerate}

\item[Case: E-Assert.] $e = \assertExpr{x}{T_{ST}}$.
		\begin{enumerate}
		\item By assumption, and because $e$ is closed:
    	\begin{enumerate}
			\item \ok{\typeBounds, \Sigma, \Delta}
			\item \stepsTo{\Sigma, \text{assert } l \text{ in } T_{ST}}{\Sigma, \unit}
		\end{enumerate}
		\item There are two subcases:
		\begin{description}
            \item[Subcase: $T_{ST} = \overline{S}$.] By assumption, \ty{\typeBounds}{\Delta_0, l : \generics{C}{T}.\overline{S}}{s}{\assertExpr{l}{\overline{S'}}}{\Unit}{\Delta_0, l : \generics{C}{T}.\overline{S}}.
				Let $\Delta' = \Delta$. By T-(), \ty{\typeBounds}{\Delta'}{l}{\unit}{\Unit}{\Delta'}. Since $\Sigma' = \Sigma$, $\Delta' = \Delta$, and \ok{\typeBounds, \Sigma, \Delta}, we have \ok{\typeBounds, \Sigma', \Delta'}.
            \item[Subcase: $T_{ST} \neq \overline{S}$.] By assumption, \ty{\typeBounds}{\Delta_0, l : \generics{C}{T}.T_{ST}}{s}{\assertExpr{l}{T_{ST}}}{\Unit}{\Delta_0, l : \generics{C}{T}.T_{ST}}. Let $\Delta' = \Delta$. By T-(), \ty{\typeBounds}{\Delta'}{l}{\unit}{\Unit}{\Delta'}. Since $\Sigma' = \Sigma$, $\Delta' = \Delta$, and \ok{\typeBounds, \Sigma, \Delta}, we have \ok{\typeBounds, \Sigma', \Delta'}.
		\end{description}
		\end{enumerate}

\item[Case: E-Disown.] $e = \disown l$.
	\begin{enumerate}
		\item By assumption, and because $e$ is closed:
    	\begin{enumerate}
			\item \ok{\typeBounds, \Sigma, \Delta}
			\item \stepsTo{\Sigma, \disown l}{\Sigma, l}
		\end{enumerate}
		\item There are two subcases:
		\begin{description}
			\item[Subcase: \ty{\typeBounds}{\Delta_0, l: \generics{C}{T}.\overline{S}}{s}{\disown \ l}{\Unit}{\Delta_0, l: T'}.] By inversion, \splitType{\generics{C}{T}.\overline{S}}{T}{T'}. Let $\Delta' = \Delta''$.  By T-(), \ty{\typeBounds}{\Delta'}{s}{\unit}{\Unit}{\Delta'}. Although the split compatibility lemma does not precisely apply here, an analogous argument does: any other alias to the object referenced by $l$ was previously compatible with $\generics{C}{T}.\overline{S}$, so we can see by case analysis of the definitions of compatibility and splitting that such aliases are also compatible with $T'$.
			\item[Subcase: {\ty{\typeBounds}{\Delta_0, l: \generics{C}{T}.Owned}{s}{\disown \ l}{\Unit}{\Delta_0, l: T'}}.] By inversion, \splitType{\generics{C}{T}.Owned}{T}{T'}. By T-(), \ty{\typeBounds}{\Delta'}{l}{\unit}{\Unit}{\Delta'}. Although the split compatibility lemma does not precisely apply here, an analogous argument does: any other alias to the object referenced by $l$ was previously compatible with $\generics{C}{T}.Owned$, so we can see by case analysis of the definitions of compatibility and splitting that such aliases are also compatible with $T'$.
		\end{description}
		\item In both subcases, $\Delta''' = \Delta''$, so $\lStronger{\Delta'''}{\typeBounds, \Sigma'}{\Delta''}$ by $<^l$-reflexivity.
	\end{enumerate}

\item[Case: E-Pack.] $e = \pack$.
	\begin{enumerate}
		\item By assumption, and because $e$ is closed:
    	\begin{enumerate}
			\item \ok{\typeBounds, \Sigma, \Delta}
			\item \stepsTo{\Sigma, \pack \ s}{\Sigma, \unit}
			\item \ty{\typeBounds}{\Delta_0, l: T, \overline{l.f: T_f}}{l}{\pack}{\Unit}{\Delta, l: T}. \textit{(Note that $\overline{l.f: T_f}$ can be any subset of the declared fields, including the empty subset.)}
		\end{enumerate}
		\item By inversion:
		\begin{enumerate}
			\item $l.f \notin dom(\Delta_0)$
			\item $contractFields(T) = \overline{T_{decl} \ f}$
            \item $\overline{\subtype{\typeBounds}{T_f}{T_{decl}}}$
		\end{enumerate}
		\item Let $\Delta' = \Delta$. By T-(), \ty{\typeBounds}{\Delta'}{l}{\unit}{\Unit}{\Delta'}. Note that every $T_f$ is a subtype of $T_{decl}$. The impact on $refTypes(\Sigma', \Delta', o)$ is that types defined for fields will replace types defined in $\Delta$. But because every replacement is a supertype of the type that it replaces, we have \ok{\typeBounds, \Sigma', \Delta'} by the \textit{subtype compatibility} lemma (\ref{subtype-compat}).
	\end{enumerate}
\end{description}
\end{proof}

\begin{theorem}[Asset retention]

Suppose:
\begin{enumerate}
	\item \ok{\typeBounds, \Sigma, \Delta}
	\item $o \in dom(\mu)$
	\item $refTypes(\Sigma, \Delta, o) = \overline{D}$
	\item \ty{\typeBounds}{\Delta}{s}{e}{T}{\Delta'}
	\item $e$ is closed
	\item \label{steps} \stepsTo{\Sigma, e}{\Sigma', e'}
	\item $refTypes(\Sigma', \Delta', o) = \overline{D'}$
    \item \label{nd} $\exists T' \in \overline{D}$ such that $\nonDisposable{\typeBounds}{T'}$
    \item \label{d} $\forall T' \in \overline{D'}: \disposable{\typeBounds}{T'}$
\end{enumerate}

Then in the context of a well-typed program, either $\nonDisposable{\typeBounds}{T}$ or $e = \mathbb{E}[\disown \ s]$, where $\rho(s) = o$.
\end{theorem}

\begin{proof}
By induction on the typing derivation.

\begin{description}
\item[Case: T-lookup.]
    In (\ref{steps}), the only rule that could have applied is E-lookup, which leaves $\Sigma$ unchanged. $\Delta'$ is the same as $\Delta$ except that some instances of $T_1$ have been replaced by $T_3$. If $\nonDisposable{\typeBounds}{T}$, it is proved. Otherwise, $\disposable{\typeBounds}{T}$, and by the definition of split, $\disposable{\typeBounds}{T_1}$ and $\disposable{\typeBounds}{T_3}$, so there was no change in disposability in $\Delta'$, contradicting the conjunction of (\ref{nd}) and (\ref{d}).

\item[Case: T-Let.] $e = \text{let } x: T = e_1 \text{ in } e_2$.
	There are two subcases, depending on the rule that was used for \stepsTo{\Sigma, e}{\Sigma', e'}:
	\begin{description}
		\item[Subcase: E-let.] $\Sigma'$ has a new mapping for a new indirect reference $l$, which may cause an additional alias to an object, but all previous aliases are preserved, so it cannot be the case that all non-disposable references are gone.
		\item[Subcase: E-letCongr.] The induction hypothesis applies to $e_1$ because \stepsTo{\Sigma, e_1}{\Sigma', e_1'}. This suffices to prove the case because there are no changes to $\Delta$.
	\end{description}

\item[Case: T-new.] By rule E-New, \stepsTo{\Sigma, \text{new } \generics{C}{T'}.S(\overline{l})}{\envUpdate{\mu[o \mapsto \generics{C}{T'}.S(\overline{\rho(l)})]}{\mu}{\Sigma}, o}. By inversion, $\ty{\typeBounds}{\Delta}{s}{\overline{s'}}{\overline{T}}{\Delta'}$.  By the induction hypothesis, any nondisposable references in $\Delta$ are preserved in $\Delta'$. The new $\Sigma'$ also preserves any existing nondisposable references.

\item[Case: T-this-field-def.]
    Rule \textit{E-field} leaves $\Sigma$ unchanged. By the \textit{split non-disposability} lemma (\ref{split-non-disposability}), if $\disposable{\typeBounds}{T_1}$, then $\disposable{\typeBounds}{T_3}$. No other types are changed in the typing context.

\item[Case: T-this-field-ctxt.]
	Same argument as for \textit{This-field-def}.

\item[Case: T-fieldUpdate.] Although $\Sigma'$ replaces a field, which may reference an object, the reference that was overwritten was disposable (by inversion).

\item[Case: T-inv.]
	The changes in $\Delta$ consist of replacing types with the results of $funcArg$. $\Sigma'$ has additional aliases to objects, but additional aliases cannot cause loss of owning references. We consider the cases for $funcArg$:
	\begin{description}
		\item[FuncArg-owned-unowned.] This case preserves ownership in the output type.
		\item[FuncArg-shared-unowned.] The input type here is not owned.
		\item[FuncArg-other.] If $owned(T_C.T_{STinput-decl})$, then in $\Delta$, the corresponding variable is an owning type. By substitution, ownership of the object will be maintained in the next context.
	\end{description}
	This represents a contradiction with the assumption that ownership was lost.

\item[Case: T-privInv.]
	This case is analogous to the case for Public-Invoke, but with additional aliases changed due to fields.

\item[Case: T-$\nearrow_p$.] $e = s \nearrow_{p} S'(\overline{x})$. 
	A rule $E-\nearrow_{p}$ applied, replacing an object that previously had a type consistent with $\generics{C}{T_A}.T_{ST}$  in $\mu$ with one that references an object in state $S'$. The new static context contains an owning reference to the new object, so ownership of $s$ was not lost.  For the dynamic context $\Sigma'$, it suffices to examine the references from fields of the old object ($\mu(\rho(l))$). It remains to consider the fields that were overwritten, but these all had types that were disposable (by inversion of T-$\nearrow_p$).

\item[Case: T-assertStates.] This rule causes no change in either $\Delta$ or $\Sigma$, which is a contradiction.
\item[Case: T-assertPermission.] This rule causes no change in either $\Delta$ or $\Sigma$, which is a contradiction.
\item[Case: T-assertInVar.] This rule causes no change in either $\Delta$ or $\Sigma$, which is a contradiction.
\item[Case: T-assertInVarAlready.] This rule causes no change in either $\Delta$ or $\Sigma$, which is a contradiction.

\item[Case: T-isInStaticOwnership.] $e = \ifExpr{x}{owned}{S}{e_1}{e_2}$.
	\begin{enumerate}
        \item If E-IsIn-Owned-Then applies, $\stepsTo{\Sigma, e}{\Sigma, e_1}$ (and by the preservation lemma, $e_1$ is well-typed). By the same argument as for T-lookup, no ownership was lost in $\Delta'$ and $\Delta''$; any consumed ownership is now in $T_1$. From the \textit{merging preserves nondisposability} lemma (\ref{merging-preserves-nondisposability}), we find a contradiction with the assumption that a type has changed from nondisposable to disposable in this step.
		\item Otherwise, E-IsIn-Else applies, and the same argument applies to $e_2$.
	\end{enumerate}
	
\item[Case: T-isInDynamic.] $e = \ifExpr{x}{shared}{S}{e_1}{e_2}$.  \mbox{}
    The same argument as in the T-isInStaticOwnership case applies, except that the situation is even simpler because $\Delta$ and $\Delta'$ agree that $x : T_C.Shared$.
    
\item[Case: T-IsIn-PermVar.] The argument is the same as for T-isInDynamic.

\item[Case: T-IsIn-Perm-Then.] E-IsIn-Perm-Then applies, and, $\Sigma' = \Sigma$. By the same argument as in the T-Lookup case, no ownership was lost in $\Delta'$, which contradicts the assumption.

\item[Case: T-IsIn-Perm-Else.] The argument is the same as for T-IsIn-Perm-Then, but with E-IsIn-Perm-Else.

\item[Case: T-IsIn-Unowned.] The argument is the same as for T-IsIn-Perm-Then, but with E-IsIn-Unowned.
    
\item[Case: T-disown.] Then $e = \disown \ s$.

\item[Case: T-pack.] Note that \pack{} leaves $\Sigma$ unchanged; the only change is removing $\overline{s.f : T_f}$ from $\Delta$. But by inversion, $\overline{\sameOwnership{T_f}{T_{decl}}}$. As a result, no ownership can change from $\Delta$ to $\Delta'$, contradicting the assumptions.

\item[Case: T-state-mutation-detection.] $e = \phibox{e'}{o}$. The step must have been either via E-Box-$\phi$ or via E-Box-$\phi$-congr. 	\begin{description}
		\item[Case: E-Box-$\phi$.] The change in E-Box-$\phi$ and state-mutation-detection has no impact on ownership, so this contradicts the assumptions.
		\item[Case: E-Box-$\phi$-congr.] We have the required property by the induction hypothesis, since the present rules make no changes themselves to $\Delta'$ and $\Sigma'$, which were provided inductively.
	\end{description}
	
\item[Case: T-reentrancy-detection.] $e = \psibox{e'}{o}$. The step must have been either via E-Box-$\psi$ or via E-Box-$\psi$-congr.
	\begin{description}
		\item[Case: E-Box-$\psi$.] The change in E-Box-$\psi$ and state-mutation-detection has no impact on ownership, so this contradicts the assumptions.
		\item[Case: E-Box-$\psi$-congr.] We have the required property by the induction hypothesis, since the present rules make no changes themselves to $\Delta'$ and $\Sigma'$, which were provided inductively.
	\end{description}

\end{description}
\end{proof}

\subsection{Supporting Lemmas}

\begin{lemma}[Memory consistency]
\label{memory-consistency}
If $\ok{\typeBounds, \Sigma, \Delta}$, then:
\begin{enumerate}
    \item If $l : \generics{C}{T'}.S \in \Delta$, then $\exists o. \rho(l) = o$ and $\mu(o) = \generics{C}{T'}.S(\overline{s})$.
	\item If $\ty{\typeBounds}{\Delta}{s}{e}{T}{\Delta'}$, and $l$ is a free variable of $e$, then $l \in dom(\rho)$.
\end{enumerate}
\end{lemma}

\begin{proof} \mbox{}
\begin{enumerate}
    \item Assume $l : \generics{C}{T'}.S \in \Delta$.
        Then $\rho(l) = o$ follows by inversion of global consistency. $\mu(o) = \generics{C^*}{T^*}.S'(\overline{o'})$ follows by inversion of reference consistency (which itself follows by inversion of global consistency).
        By inversion of reference consistency, $\subtype{\cdot}{\generics{C^*}{T^*}.S'}{\overline{D}}$.
        By definition of refTypes, $\generics{C}{T'}.S \in \overline{D}$, so $\subtype{\cdot}{\generics{C^*}{T^*}.S'}{\generics{C}{T'}.S}$.
        This implies that $C = C^*$, $S = S'$, and $\subtype{\cdot}{\overline{T^*}}{\overline{T'}}$ (by definition of subtyping).
	\item By induction on the typing derivation, we prove that if $l$ is a free variable of $e$, then $l \in dom(\Delta)$. Then the conclusion follows immediately from the definition of global consistency. We consider some example cases:
		\begin{description}
			\item[Case: T-lookup.] $s'$ is a free variable, but $s': T_1 \in \Delta$.
			\item[Case: T-let.] Any free variables in $e$ must be in $e_1$ or $e_2$. The result is obtained by induction on $e_1$ and $e_2$.
			\item[Case: $s \nearrow_p S'(\overline{x})$.] $s$ is a free variable, but $s: \generics{C}{T_A}.T_{ST} \in \Delta$.
			\item[Case: T-assertStates.] $x$ is a free variable, but $x \in dom(\Delta)$.
		\end{description}
		The remaining cases are similar to the above.
\end{enumerate}
\end{proof}

\begin{lemma}[Split Non-disposability]
\label{split-non-disposability}
	If \splitType{T_1}{T_2}{T_3}, and $T_1$ is not disposable, then $T_2$ is not disposable.
\end{lemma}
\begin{proof}
	By inspection of the definition of \splitType{T_1}{T_2}{T_3} and $owned$. Note that in the \textit{Split-owned-shared} and \textit{Split-states-shared} cases, although $owned(T_1)$, C is not an asset, which makes $T_1$ disposable.
\end{proof}

\begin{lemma}[Subtype Compatibility]
\label{subtype-compat}
If $\compatibleTypes{T}{T'}$, and $\subtype{\typeBounds}{T'}{T''}$, then $\compatibleTypes{T}{T''}$.
\end{lemma}
\begin{proof}
By straightforward case analysis of the subtyping relation.
\end{proof}

\begin{lemma}[Subtyping reflexivity]
\label{subtyping-reflexivity}
For all types $T$, $\subtype{\typeBounds}{T}{T}$.
\end{lemma}
\begin{proof}
\begin{description}
	\item[Case: $\Unit$.] Rule Unit applies.
    \item[Case: $T_C.T_{ST}$]. By rule Refl in the definition of the subpermission relation, rule \textit{Matching-definitions} applies.
\end{description}
\end{proof}

\begin{lemma}[Exclusivity of isAsset/nonAsset]
    For all types $T$:
    \begin{enumerate}
        \item If $\isAsset{\typeBounds}{T}$ is provable, then $\nonAsset{\typeBounds}{T}$ is not provable.
        \item If $\nonAsset{\typeBounds}{T}$ is provable, then $\isAsset{\typeBounds}{T}$ is not provable.
    \end{enumerate}
\end{lemma}
\begin{proof}
By straightforward case analysis of the isAsset and nonAsset rules.
\end{proof}

\begin{lemma}[Exclusivity of isVar/nonVar]
    For all types $T$:
    \begin{enumerate}
        \item If $\isVar{T}$ is provable, then $\nonVar{T}$ is not provable.
        \item If $\nonVar{T}$ is provable, then $\isVar{T}$ is not provable.
    \end{enumerate}

    For all declaration types $T_C$:
    \begin{enumerate}
        \item If $\isVar{T_C}$ is provable, then $\nonVar{T_C}$ is not provable.
        \item If $\nonVar{T_C}$ is provable, then $\isVar{T_C}$ is not provable.
    \end{enumerate}

    For all permissions/states $T_{ST}$:
    \begin{enumerate}
        \item If $\isVar{T_{ST}}$ is provable, then $\nonVar{T_{ST}}$ is not provable.
        \item If $\nonVar{T_{ST}}$ is provable, then $\isVar{T_{ST}}$ is not provable.
    \end{enumerate}
\end{lemma}
\begin{proof}
By straightforward case analysis of the isVar and nonVar rules.
\end{proof}

\begin{lemma}[Exclusivity of maybeOwned/notOwned]
    For all types $T$:
    \begin{enumerate}
        \item If $\maybeOwned{T}$ is provable, then $\notOwned{T}$ is not provable.
        \item If $\notOwned{T}$ is provable, then $\maybeOwned{T}$ is not provable.
    \end{enumerate}
\end{lemma}
\begin{proof}
By straightforward case analysis of the ownedState and notOwned rules.
\end{proof}

\begin{definition}[Non-disposability]
\begin{mathpar}

\inferrule*[right=ND-owned]{\maybeOwned{T_C.T_{ST}} \and
    \isAsset{\typeBounds}{T_C.T_{ST}} }
    { \nonDisposable{\typeBounds}{T_C.T_{ST}} }

\end{mathpar}
\end{definition}

\begin{lemma}[Exclusivity of disposability and non-disposability]
\label{disposable-xor-nondisposable}
For all types $T$:
	\begin{enumerate}
        \item If $\disposable{\typeBounds}{T}$ is provable, then $\nonDisposable{\typeBounds}{T}$ is not provable.
        \item If $\nonDisposable{\typeBounds}{T}$ is provable, then $\disposable{\typeBounds}{T}$ is not provable.
	\end{enumerate}
\end{lemma}

\begin{proof}
	\begin{enumerate}
        \item Consider the cases for $\disposable{\typeBounds}{T}$.
			\begin{description}
                \item[Case: D-Owned.] Let $T = T_C.T_{ST}$.
                    By inversion, $\maybeOwned{T_C.T_{ST}}$ and $\nonAsset{\typeBounds}{T_C.T_{ST}}$.
                    Then we cannot prove $nonDisposable$, which requires $\isAsset{\typeBounds}{T_C.T_{ST}}$.

                \item[Case: D-not-owned.] There is no rule by which to prove $\nonDisposable{\typeBounds}{T}$.
                \item[Case: D-Unit.] There is no rule by which to prove $\nonDisposable{\typeBounds}{T}$.
			\end{description}
        \item To prove $\nonDisposable{\typeBounds}{T}$, we must use ND-Owned; so $T = T_C.T_{ST}$, and we must show that $\maybeOwned{T_C.T_{ST}}$ and $\isAsset{\typeBounds}{T_C.T_{ST}}$.
            But this directly contradicts the premises of D-not-owned and D-owned, and D-unit does not apply.
            So there is no rule by which to prove $\disposable{\typeBounds}{T}$.
	\end{enumerate}
\end{proof}

\begin{lemma}[Interface substitution]
    \label{lem:subcon-subs}
    If
    \begin{enumerate}
        \item $\ty{\typeBounds}{\Delta, s' : \generics{I}{T}.T_{ST}}{s}{e}{T}{\Delta'}$
        \item $\subtype{\typeBounds}{\generics{C}{T'}}{\generics{I}{T}}$
        \item $\ok{C}$
    \end{enumerate}
    then $\ty{\typeBounds}{\Delta, s' : \generics{C}{T'}.T_{ST}}{s}{e}{T}{\Delta''}$, where
    \begin{enumerate}
        \item if $s' : \generics{I}{T}.T_{ST}' \in \Delta'$, then $\Delta'' = \Delta', s' : \generics{C}{T'}.T_{ST}' $
        \item otherwise $\Delta'' = \Delta'$.
    \end{enumerate}
\end{lemma}
\begin{proof}
    By induction on the typing derivation.
    The relevant cases are Inv and P-Inv; all other cases will be identical, because $x$ has the same permission or state.
    Because interfaces don't have fields, there must not be any field assignments or access involving $x$, so we don't need to consider those.

    \begin{description}
        \item[Case: Inv] In this case, $e = s_1.\generics{m}{T_M}(\overline{s_2})$.

            If $s' = s_1$, with the assumption that $\okIn{m}{C}$, we have:
            \begin{enumerate}
                \item $\transaction{\typeBounds}{\generics{m}{T_M}}{\generics{I}{T}} = T \ \generics{m}{T_M'}(\overline{T_{C_x}.T_x \trans T_{xST} \ x}) \ T_{this} \trans \ T_{this}'$
                \item $\ok{C}$
                \item $\subperm{\typeBounds}{T_{ST}}{T_{this}}$
                \item $\overline{\subtype{\typeBounds}{T_{s2}}{T_{C_x}.T_x}}$
	            \item $T_{s1}' = funcArg(T_C.T_{STs1}, T_C.T_{this}, T_C.T_{this}')$
                \item $\overline{T_{s2}' = funcArg(T_{s2}, T_x, T_{C_x}.T_{xST})}$
            \end{enumerate}

            So then
            \begin{enumerate}
                \item $\transactionName{m} \in \transactionNames{C}$
                \item $tdef(m, C) = M = T' \ \generics{m}{T_M'}(\overline{T_{C_x}'.T_x' \trans T_{xST}' \ x}) \ T_{this}^{*} \trans \ T_{this}^{**}$
                \item $\implementOk{\typeBounds}{\generics{I}{T}}{M}$.
            \end{enumerate}

            By definition of implementOk, this implies that the invocation is still well-typed.

            If $s' \in \overline{s_2}$, then $\subtype{\typeBounds}{\generics{I}{T}.T_{ST}^{*}}{T}$ for some argument of type $T$.
            But then because subtyping is transitive, $\generics{C'}{T'}.T_{ST}^*$ is also a subtype of $T$, so the invocation is still safe.

        \item[Case: P-Inv] Identical to the Inv case, except that we cannot invoke a private transaction on $s'$, as interfaces do not have private transactions, so $s'$ must be one of the arguments.
    \end{description}
\end{proof}

\begin{lemma}[Permission Variable Substitution]
    \label{lem:permvar-subs}
    Suppose
    \begin{enumerate}
        \item $\subperm{\typeBounds}{T_{ST}}{p}$
        \item $\ty{\typeBounds}{\Delta, x : T_C.p}{s}{e}{T}{\Delta'}$
    \end{enumerate}

    Then $\ty{\typeBounds}{\Delta, x : T_C.T_{ST}}{s}{e}{T}{\Delta'}$.
\end{lemma}
\begin{proof}
   Follows from \ref{subtype-substitution}
\end{proof}

\begin{lemma}[Exclusivity of subpermission]
\label{lem:subperm-exclusivity}
    For any permissions $P$ and $P'$:
    \begin{enumerate}
        \item If $\subperm{\cdot}{P}{P'}$ is provable, then $\notsubperm{\cdot}{P}{P'}$ is not provable.
        \item If $\notsubperm{\cdot}{P}{P'}$ is provable, then $\subperm{\cdot}{P}{P'}$ is not provable.
    \end{enumerate}
\end{lemma}
\begin{proof}
By case analysis of the subpermission rules, we can see that every pair of permissions is related.
The only way that $\subperm{\cdot}{P}{P'}$ \textbf{and} $\subperm{\cdot}{P'}{P}$ can be true is if $P = P'$, but then we cannot prove $\notsubperm{\cdot}{P}{P'}$.
\end{proof}

\begin{lemma}[Split compatibility]
\label{split-compatibility}
If \ty{\typeBounds}{\Delta}{s}{\overline{s'}}{\overline{T}}{\Delta'} and \ok{\typeBounds, \Sigma, \Delta} then \ok{\typeBounds, \Sigma, \Delta'}.
\end{lemma}
\begin{proof}
For one expression, it suffices to show that replacing $T$ with $T_3$ in $\Delta$ leaves the remaining context consistent with $\Sigma$. The proof of this is by cases of splitting; this is theorem \texttt{splittingRespectsHeap} in heapLemmasforSplitting.agda. For multiple expressions, simply iterate the argument.
\end{proof}

\begin{lemma}[Substitution]
\label{substitution}
If \ty{\typeBounds}{\Delta, x: T_x}{s}{e}{T'}{\Delta', x : T_x'}, then \ty{\typeBounds}{\Delta, l: T_x}{s}{[l/x]e}{T'}{\Delta', l: T_x'}
\end{lemma}

\begin{proof}
Substitute $l$ for $x$ throughout the previous proof.
\end{proof}

\begin{lemma}[Subtype replacement]
\label{subtype-replacement}
	If 
		\begin{itemize}
			\item \ty{\typeBounds}{\Delta, x: T_x}{s}{e}{T'}{\Delta', x : T_x'}
			\item \subtype{\typeBounds}{T_x''}{T_x} 
			\item \sameOwnership{T_x''}{T_x}
		\end{itemize}
		then \ty{\typeBounds}{\Delta, x: T_x''}{s}{e}{T'}{\Delta', x: T_x'''} where \subtype{\typeBounds}{T_x'''}{T_x'}.
\end{lemma} 
\begin{proof}
By induction on the typing derivation and the subtyping derivation. Relevant cases include:

\begin{description}
        \item[Case: T-lookup.]  \mbox{}
        	\begin{enumerate}
				\item By assumption:
            	\begin{enumerate}
    				\item $\subtype{\typeBounds}{T_x''}{T_x}$
    			\end{enumerate}				
        
        \item By inversion of T-lookup: 
        	\begin{enumerate}
				\item \label{subtype-split} $\splitType{T_x}{T'}{T_x'}$
			\end{enumerate}				
		
		\item Note that it suffices to show that \splitType{T_x''}{T'}{T_x'''}. Consider the cases for \ref{subtype-split}:
			\begin{description}
				\item[Case: Split-unowned]
					$T_x' = T_C.\Unowned$. Split-Unowned applies to $T_x''$, resulting in $T_x''' = T_C.\Unowned$.
				\item[Case: Split-shared]
					By assumption, $T_x = T_C.\Shared$. If $T_x'' = T_C.\Shared$, then the result follows by Split-Shared. Otherwise, $\maybeOwned{T_x''}$, but this contradicts the assumption that \sameOwnership{T_x''}{T_x}.
					
				\item[Case: Split-owned-shared]
					By inversion of maybeOwned, we have the following cases:
						\begin{description}
							\item[Subcase: $T_x = T_C.p$]
								All subtypes of $T_C.p$ are themselves maybeOwned and nonAsset, so Split-owned-shared applies.
							\item[Subcase: $T_x = T_C.\Owned$]
								All subtypes of $T_C.\Owned$ are themselves maybeOwned and nonAsset, so Split-owned-shared applies.
							\item[Subcase: $T_x = T_C.\overline{S}$]
								All subtypes of $T_C.\overline{S}$ are themselves maybeOwned and nonAsset, so Split-owned-shared applies.
						\end{description}
				\item[Case: Split-unit]
					$T_x = T_x' = \Unit$. Split-unit applies for $T_x''$, since the only subtype of $\Unit$ is $\Unit$. Then $T_x''' = \Unit$, which is a subtype of $T_x'$.
			\end{description}
		\end{enumerate}
		
        \item[Case: T-IsIn-StaticOwnership.] \subtype{\typeBounds}{T_x''}{T_x} results in a smaller set of initial possible states for $x$, resulting in a potentially smaller set of possible states for $x$ in the resulting context. This explains why it is not necessarily the case that $T_x''' = T_x''$.
        
\end{description}
\end{proof}

\begin{corollary}[Subtype substitution]
\label{subtype-substitution}
	If 
		\begin{itemize}
			\item \ty{\typeBounds}{\Delta, x: T_x}{s}{e}{T'}{\Delta', x : T_x'} and 
			\item \subtype{\typeBounds}{T_x''}{T_x} 
		\end{itemize}
		then \ty{\typeBounds}{\Delta, l: T_x''}{s}{[l/x]e}{T'}{\Delta', l: T_x'''} where \subtype{\typeBounds}{T_x'''}{T_x'}.
\end{corollary} 
\begin{proof}
	Follows by applying both \ref{subtype-replacement} and \ref{substitution}.
\end{proof}

\begin{corollary}[l-stronger substitution]
\label{l-stronger-substitution}
If \ty{\typeBounds}{\Delta}{s}{e}{T}{\Delta'} and $\lStronger{\Delta''}{\typeBounds, \Sigma'}{\Delta}$, then \ty{\typeBounds}{\Delta''}{s}{e}{T}{\Delta'''} with $\lStronger{\Delta'''}{\typeBounds, \Sigma'}{\Delta'}$.
\end{corollary}
\begin{proof}
By induction on $\Delta'$, applying \ref{subtype-substitution} and the definition of $\lStronger{\Delta''}{\typeBounds, \Sigma'}{\Delta}$.
\end{proof}

\begin{lemma}[l-stronger consistency]
\label{l-stronger-consistency}
If  $\lStronger{\Delta'}{\typeBounds, \Sigma'}{\Delta}$ and \ok{\typeBounds, \Sigma, \Delta} then \ok{\typeBounds, \Sigma, \Delta'}.
\end{lemma}
\begin{proof}
By induction on $\Delta$ and application of subtype compatibility (\ref{subtype-compat}).
\end{proof}

\begin{lemma}[Strengthening]
If \ty{\typeBounds}{\Delta, s': T_0}{s}{e}{T}{\Delta', s': T_1}, and $s'$ does not occur free in $e$, then \ty{\typeBounds}{\Delta}{s}{e}{T}{\Delta'}.
\end{lemma}

\begin{proof}
By induction on the typing derivation. Since $s'$ does not occur free in $e$, $s'$ must not be needed in either proof.
\end{proof}

\begin{lemma}[Weakening]
If \ty{\typeBounds}{\Delta}{s}{e}{T}{\Delta'}, and $s'$ does not occur free in $e$, then \ty{\typeBounds}{\Delta, s': T_0}{s}{e}{T}{\Delta', s': T_1}.
\end{lemma}

\begin{proof}
By induction on the typing derivation. Since $s'$ does not occur free in $e$, $s'$ must not be needed in either proof.
\end{proof}

\begin{lemma}[Merge consistency]
If \ok{\typeBounds, \Sigma, \Delta} and \ok{\typeBounds, \Sigma, \Delta'}, then \ok{\typeBounds, \Sigma, merge(\Delta, \Delta')}.
\end{lemma}
\begin{proof}

By induction on $merge(\Delta, \Delta')$.
\begin{description}
\item[Case: Sym.] By the induction hypothesis, \ok{\typeBounds, \Sigma, merge(\Delta', \Delta)}, and $merge(\Delta', \Delta) = merge(\Delta, \Delta')$.

\item[Case: $\oplus$.] By inversion, $\Delta = \Delta'', x : T$ and $\Delta' = \Delta''', x : T'$. Because $\Delta''$ is a subset of $\Delta'', x : (T \oplus T')$, by the induction hypothesis, \ok{\typeBounds, \Sigma, merge(\Delta'', \Delta''')} (the induction hypothesis applies because $dom(\Delta'') \subset dom(\Delta'', x : (T \oplus T'))$ and $\Delta''(x') = \Delta(x')$ for all $x \neq x'$). Therefore, by the definition of consistency, it suffices to show that $T \oplus T'$ is compatible with all $T'' \in refTypes(\Sigma, \Delta'')$. We assume that either $\compatibleTypes{T''}{T}$ or $\compatibleTypes{T''}{T'}$.
	\begin{description}
		\item[Subcase: $T \oplus T = T$.] Anything compatible with T is trivially compatible with T.
        \item[Subcase: $T_C.Owned \oplus T_C.\overline{S} = T_C.Owned$.] If $T''$ is compatible with $T_C.Owned$, then it is proved. Otherwise, $T''$ is compatible with $T_C.\overline{S}$ (by inspection of the definition of $\leftrightarrow$).
            In particular, $T''$ must be $T_C.Unowned$, in which case rule UnownedOwnedCompat applies.
        \item[Subcase: $T_C.Shared \oplus T_C.Unowned = T_C.Unowned$]. If $T''$ is compatible with $T_C.Unowned$, then it is proved.
            Otherwise, $T''$ is compatible with $T_C.Shared$, and by definition of $\leftrightarrow$, either $T'' = T_C.Shared$ or $T'' = T_C.Unowned$.
            The later case was already addressed, and in the former case, SharedCompat gives $\compatibleTypes{T''}{T_C.Shared}$.
		\item[Subcase: $T_C.\overline{S} \oplus T_C.\overline{S'} = T_C.(S \cup S')$.] The only compatibility rule that could have applied was UnownedStatesCompat, and it still applies to $T_C.(S \cup S')$.
        \item[Subcase: $\generics{C}{T}.T_{ST} \oplus \generics{I}{T'}.T_{ST}' = \generics{I}{T \oplus T'}.T_{ST} \oplus \generics{I}{T \oplus T'}.T_{ST}' = \generics{I}{T^*}.T_{ST}^*$.]
            If $T''$ is compatible with $\generics{C}{T}.T_{ST}$, then it will also be compatible with $\generics{I}{T^*}.T_{ST}^*$ by SubtypeCompat, ParamCompat, and application of one of the other subcases for $T_{ST}^*$.
            If $T''$ is compatible with $\generics{I}{T'}.T_{ST}'$, then it will also be compatible with $\generics{I}{T^*}.T_{ST}^*$ by ParamCompat and application of one of the other subcases for $T_{ST}^*$.

        \item[Subcase: $\generics{D}{T}.T_{ST} \oplus \generics{D}{T'}.T_{ST}' = \generics{D}{T \oplus T'}.T_{ST} \oplus \generics{D}{T \oplus T'}.T_{ST}' = \generics{D}{T^*}.T_{ST}^*$.]
            If $T''$ is compatible with $\generics{D}{T}.T_{ST}$, then it will also be compatible with $\generics{D}{T^*}.T_{ST}^*$ by ParamCompat, and application of one of the other subcases for $T_{ST}^*$.
            If $T''$ is compatible with $\generics{D}{T'}.T_{ST}$, then it will also be compatible with $\generics{D}{T^*}.T_{ST}^*$ by ParamCompat, and application of one of the other subcases for $T_{ST}^*$.
	\end{description}

\item[Case: Dispose-disposable.] Eliminating a variable from a context that is already consistent with $\Sigma$ leaves a context that is still consistent with $\Sigma$. Note that this rule does not allow removing bindings of the form $x.f: T$ because removing those bindings could result in inconsistencies, since then the types of those fields would (incorrectly) be assumed to be according to their declarations.

\end{description}
\end{proof}

\begin{lemma}[Merge Subtyping]
\label{merge-subtyping}
If $l: T \in merge(\Delta^*, \Delta^{**})$, then $l: T_1 \in \Delta^*$ and $l: T_2 \in \Delta^{**}$ with \subtype{\typeBounds}{T_1}{T}, \subtype{\typeBounds}{T_2}{T}, \sameOwnership{T_1}{T}, and \sameOwnership{T_2}{T}.
\end{lemma}
\begin{proof}
By induction on the merge judgment.

\begin{description}
	\item[Case: Sym.] The conclusion follows immediately from the induction hypothesis.
	\item[Case: $\oplus$.] In each subcase, the conclusion follows from the induction hypothesis and the $\oplus$ subtyping lemma (\ref{oplus-subtyping}).
	\item[Case: Dispose-disposable.] $d \nin merge(\Delta^*, \Delta^{**})$, so the conclusion follows immediately from the induction hypothesis.
\end{description}
\end{proof}

\begin{lemma}[$\oplus$ subtyping]
\label{oplus-subtyping}
	If $T_1 \oplus T_2 = T$, then \subtype{\typeBounds}{T_1}{T} and \subtype{\typeBounds}{T_2}{T}.
\end{lemma}
\begin{proof}
		\begin{description}
			\item[Case: $T \oplus T$.] It is proved by reflexivity of $<:$.
			\item[Case: $T_C.Owned \oplus T_C.\overline{S}$.] \subtype{\typeBounds}{T_C.Owned}{T_C.Owned} and \subtype{\typeBounds}{T_C.\overline{S}}{T_C.Owned}.
			\item[Case: $T_C.Shared \oplus T_C.Unowned$.] \subtype{\typeBounds}{T_C.Shared}{T_C.Unowned} and \subtype{\typeBounds}{T_C.Unowned}{T_C.Unowned}.
			\item[Case: $T_C.\overline{S} \oplus T_C.\overline{S'}$.] \subtype{\typeBounds}{T_C.\overline{S}}{T_C.(S \cup S')} and \subtype{\typeBounds}{T_C.\overline{S'}}{T_C.(S \cup S')}.
			\item[Case: $\generics{C}{T}.T_{ST} \oplus \generics{I}{T}.T_{ST}'$.] \subtype{\typeBounds}{\generics{C}{T}.T_{ST}}{\generics{I}{T}.(T_{ST} \oplus T_{ST}')} and \subtype{\typeBounds}{\generics{I}{T}.T_{ST}'}{\generics{I}{T}.(T_{ST} \oplus T_{ST}')} by rule \textit{implements-interface} and the induction hypothesis.
			\item[Case: $\generics{D}{T}.T_{ST} \oplus \generics{D}{T}.T_{ST}'$.] \subtype{\typeBounds}{\generics{D}{T}.T_{ST}}{\generics{D}{T}.(T_{ST} \oplus T_{ST}')} and \subtype{\typeBounds}{\generics{D}{T}.T_{ST}'}{\generics{D}{T}.(T_{ST} \oplus T_{ST}')} by rule Matching-Declarations and the induction hypothesis.
		\end{description}
	Note that in each of the above cases, \sameOwnership{T_1}{T_2}.

\end{proof}

\begin{theorem}[Unicity of ownership]
    If $\ok{\typeBounds, \Sigma, \Delta}$, and $o \mapsto \generics{C}{T}.S(\ldots) \in \mu$, and $refTypes(\Sigma, \Delta, o) = \overline{D}$, then at most one $T \in \overline{D}$ is either $\generics{C}{T}.\overline{S}$ or $\generics{C}{T}.Owned$.
\end{theorem}

\begin{proof}
    By inversion of reference consistency, $\forall T_1, T_2 \in \overline{D}, \compatibleTypes{T_1}{T_2} \text{ or }(o \in \Sigma_\phi \text { and } T_i = \generics{C}{T}.S \text { and } T_j = \generics{C}{T}.Shared (i \neq j))$. Note that $\generics{C}{T}.Owned$ is not compatible with either $\generics{C}{T}.Owned$ or $\generics{C}{T}.\overline{S}$, and $\generics{C}{T}.\overline{S}$ is not compatible with $\generics{C}{T}.\overline{S}$. If there were more than one alias of type $\generics{C}{T}.Owned$ or $\generics{C}{T}.\overline{S}$, they would be incompatible, which would be a contradiction. Even if $o \in \Sigma_\phi$, the aliases are restricted to shared and state-specifying aliases, and never more than one state-specifying alias exists.
\end{proof}

\begin{lemma}[Merging preserves nondisposability]
\label{merging-preserves-nondisposability}
    Suppose $\Delta_1, \Delta_2$ are static contexts. If ($s : T \in \Delta_1$ or $s : T \in \Delta_2$), $\nonDisposable{\typeBounds}{T}$, and $merge(\Delta_1, \Delta_2) = \Delta$, then $s : T' \in \Delta$ such that $\nonDisposable{\typeBounds}{T'}$.
\end{lemma}

\begin{proof}
By case analysis on $merge(\Delta_1, \Delta_2)$.

\begin{description}
	\item[Case: Sym]. The induction hypothesis applies to $merge(\Delta, \Delta')$ since the lemma was stated symmetrically.
	\item[Case: $\oplus$.] Note that in all cases of the definition of $T_1 \oplus T_2 = T_3$, if either $owned(T_1)$ or $owned(T_2)$, then $owned(T_3)$as well.
    \item[Case: Dispose-disposable.] Without loss of generality, suppose $s : T \in \Delta_1$. By inversion, $x \nin \Delta_2$. By assumption, $\nonDisposable{\typeBounds}{T}$. But by inversion, $\disposable{\typeBounds}{T}$. This is a contradiction (\ref{disposable-xor-nondisposable}).
\end{description}

\end{proof}

\end{document}